\documentclass{article} 
\usepackage{iclr2026_conference,times}

\usepackage{amsmath,amsfonts,bm}

\def\eqref#1{equation~\ref{#1}}

\def\1{\bm{1}}

\DeclareMathAlphabet{\mathsfit}{\encodingdefault}{\sfdefault}{m}{sl}
\SetMathAlphabet{\mathsfit}{bold}{\encodingdefault}{\sfdefault}{bx}{n}

\DeclareMathOperator*{\argmax}{arg\,max}

\usepackage{hyperref}
\usepackage{url}
\usepackage[utf8]{inputenc} 
\usepackage[T1]{fontenc}    
\usepackage{hyperref}       
\usepackage{url}            
\usepackage{booktabs}       
\usepackage{amsfonts}       
\usepackage{nicefrac}       
\usepackage{microtype}      
\usepackage[dvipsnames, svgnames, x11names]{xcolor}
\usepackage{amsmath}
\usepackage{filecontents}

\usepackage{amsfonts}
\usepackage{amssymb}
\usepackage{algorithmic}
\usepackage{graphicx}
\usepackage{textcomp}
\usepackage{xcolor}

\usepackage{mathrsfs}
\usepackage{amsthm}
\usepackage{epstopdf}
\usepackage{balance}
\usepackage{multirow}

\usepackage{epsfig,endnotes}
\usepackage{grffile}
\usepackage[font=bf]{caption}
\usepackage{color, url}
\usepackage{xspace} 
\usepackage[ruled,linesnumbered]{algorithm2e}
\usepackage{balance}
\usepackage{bm}
\usepackage{rotating}

\usepackage{enumitem}
\usepackage{makecell}
\usepackage{pdfx}
\usepackage{bm}
\usepackage{subcaption}
\usepackage{float}
\usepackage{fancyhdr}
\newcommand{\myparatight}[1]{\noindent{\bf {#1}.}~}
\newcommand{\name}{\text{EnsembleSHAP}}

\newcommand\WARN[1]{\textcolor{red}{#1}}

\newtheorem{theorem}{Theorem}

\newtheorem{property}{Property}
\renewcommand{\mathbf}[1]{\bm{#1}}

 \newenvironment{proofsketch}
{\begin{proof}[Proof sketch]}
{\end{proof}}

\title{\Large \bf {\name}: Faithful and Certifiably Robust Attribution for Random Subspace Method}

\author{Yanting Wang \& Jinyuan Jia \\
Pennsylvania State University\\
\texttt{\{yanting,jinyuan\}@psu.edu} \\
}

\iclrfinalcopy
\begin{document}

\maketitle

\begin{abstract}
Random subspace method has wide security applications such as providing certified defenses against adversarial and backdoor attacks, and building robustly aligned LLM against jailbreaking attacks. However, the explanation of random subspace method lacks sufficient exploration. Existing state-of-the-art feature attribution methods, such as Shapley value and LIME, are computationally impractical and lacks security guarantee when applied to random subspace method. In this work, we propose {\name}, an intrinsically faithful and secure feature attribution for random subspace method that reuses its computational byproducts. Specifically, our feature attribution method is 1) computationally efficient, 2) maintains essential properties of effective feature attribution (such as local accuracy), and 3) offers guaranteed protection against privacy-preserving attacks on feature attribution methods. To the best of our knowledge, this is the first work to establish provable robustness against explanation-preserving attacks. We also perform comprehensive evaluations for our explanation's effectiveness when faced with different empirical attacks, including backdoor attacks, adversarial attacks, and jailbreak attacks. The code is at \url{https://github.com/Wang-Yanting/EnsembleSHAP}. \WARN{WARNING: This document may include content that could be considered harmful.}
\end{abstract}

\section{Introduction}
Random subspace method~\citep{ho1998random_subspace_method}, also referred to as attribute bagging~\citep{bryll2003attribute_bagging}, is an ensemble learning method that combines the prediction results on random subsets of features to obtain the final prediction. While it was initially proposed to enhance decision trees~\citep{ho1998random_subspace_method}, this method gained widespread adaptation recently in security applications~\citep{jia2021bagging,levine2020sparse_cert,robey2023smoothllm,zhang2023pointcert,wang2021certify_graph,zeng2023ranmask,cao2023rallm}, such as providing certified defenses~\citep{levine2020sparse_cert,zeng2023ranmask,zhang2023pointcert,wang2021certify_graph,jia2020certified,zhang2021backdoor,cao2021provably,jia2022almost} against adversarial attacks, and enhancing the robustness of large language models against jailbreaking attacks~\citep{cao2023rallm,robey2023smoothllm}. This method begins by generating predictions for multiple sub-sampled versions of a given input sample using a base model. It then creates an ensemble model that aggregates these predictions using a majority vote to determine the final prediction. As this approach only requires black-box access to the base model, it can be applied across different base model architectures~\citep{levine2020sparse_cert,zeng2023ranmask,zhang2023pointcert,wang2021certify_graph,robey2023smoothllm}. Understanding the output of the random subspace method is crucial. For instance, in defending against jailbreaking attacks, it's essential for users to pinpoint the specific elements of the input prompt that lead to its classification as `harmful' (or `non-harmful'). Additionally, when a certified defense is compromised by strong empirical attacks, it becomes important for users to determine the specific adversarial words that caused the misclassification.

  However, existing leading feature attribution methods for black-box models~\citep{lundberg2017SHAP, chen2023estimate_shapley,ribeiro2016lime,enouen2023textgenshap,paes2024multi_level,amara2024syntaxshap,mosca2022shap_based,lopardo2023faithful,wang2025tracllm} such as Shapley values~\citep{lundberg2017SHAP, chen2023estimate_shapley} and LIME~\citep{ribeiro2016lime} have following disadvantages when applied to random subspace method.
 Firstly, they incur prohibitively high computational costs. Specifically, these methods involve randomly perturbing the input sample numerous times, denoted by $M$. The explained model then generates a prediction for each perturbed version of the input. In the case of the random subspace method, the objective is to explain the behavior of the ensemble model. Therefore, for each perturbed input version, $N$ sub-sampled variants of it are created and each is evaluated to generate the ensemble model's prediction, resulting in $M\times N$ total queries to the base model for just one input sample. In practical applications, both $N$ and $M$ can exceed $1,000$~\citep{enouen2023textgenshap,zeng2023ranmask,levine2020sparse_cert}, which significantly increases the computational cost. 
Secondly, current feature attribution methods do not provide security guarantees against the recently proposed \emph{explanation-preserving attack}~\citep{nadeem2023sok,noppel2023sok}. In this type of attack, an adversary can perturb certain features of the input sample to cause misclassifications, and at the same time conceal these changes by preserving the original explanation.
  
  \myparatight{Our contributions}
  In this work, we propose a computationally efficient feature attribution method for random subspace method that is inherently faithful and secure. We conduct a theoretical analysis to show that our method maintains key properties of Shapley value and is provably robust against explanation-perserving attacks on feature attribution methods. We note that this is the first work to establish provable robustness against explanation-preserving attacks. We carry out empirical evaluations to assess the effectiveness of our explanations across various security applications.

\section{Background and Related Work}
\label{background}
\vspace{-2mm}
\subsection{Random Subspace Method}
\vspace{-2mm}
\label{sec:subspace-based-method}
Random subspace-based method~\citep{ho1998random_subspace_method} is a versatile approach that is agnostic to model architecture and scalable to large neural networks. This method employs a base model to generate predictions for random feature subsets, which are then aggregated to produce the ensemble prediction. Next, we summarize a general framework for the random subspace method and introduce its security applications.

\myparatight{Building an ensemble model}
Suppose we have a testing input $\mathbf{x}=\{x_1, x_2, \cdots, x_d\}$ that consists of $d$ elements, where each element represents a feature of the input. For instance, when $\mathbf{x}$ is a text, each $x_i$ represents a word. We use $f: X\rightarrow [C]$ to represent the base model (e.g., an off-the-shelf text classifier), where $X$ is the space of model input and $[C]$ represents the set of unique labels $\{1,2, ..., C\}$ that base model can output. For example, in a binary classification problem,$[C] = \{1,2\}$. 
 For the simplicity of notation, we define a simplified base model $h$ that can take subsets of $\mathbf{x}$ (denoted by $\mathbf{z}$) as input. Specifically, we define $h: \mathcal{P}(\mathbf{x}) \rightarrow [C]$ as $h(\mathbf{z}) = f(\text{ABLATE}(\mathbf{x},\mathbf{z}))$, where $\mathcal{P}(\mathbf{x})$ is the power set of $\mathbf{x}$, $\mathbf{z} \in \mathcal{P}(\mathbf{x})$ is a subset of $\mathbf{x}$ and ABLATE replaces all features of $\mathbf{x}$ not in $\mathbf{z}$ by a special value (e.g., the `[MASK]' token). That is, $x_i = x_i \text{ if } x_i \in \mathbf{z}$, and $x_i = SpecialValue \text{ if } x_i \notin \mathbf{z}$.

The random subspace method first uses the simplified base model $h$ to make predictions for random subsets of $\mathbf{x}$. Formally, we define the probability that a label $c \in [C]$ is predicted by the base model as:
\begin{align}p_c(\mathbf{x},h,k) = \mathbb{E}_{\mathbf{z}\sim \mathcal{U}(\mathbf{x},k)}[\mathbb{I}(h(\mathbf{z})=c)],\end{align} where $\mathbb{I}$ is an indicator function whose output is 1 if the condition is satisfied and 0 otherwise, and $\mathbf{z} \sim \mathcal{U}(\mathbf{x},k)$ is a subset of $\mathbf{x}$ with size $k$ that is randomly sampled from the uniform distribution. i.e., $\Pr(\mathbf{z} = \mathbf{z}' \mid \mathbf{z} \sim \mathcal{U}(\mathbf{x},k)) = \frac{1}{\binom{d}{k}}$ for any $\mathbf{z}'\subseteq \mathbf{x}$ satisfying $|\mathbf{z}'| = k$. Then the label with the largest probability is viewed as the predicted label of the ensemble classifier $H$ for the testing input $\mathbf{x}$, i.e., 
\begin{align}H(\mathbf{x},h,k) = \argmax_c p_c(\mathbf{x},h,k).\label{eqn-ensemble}\end{align} 
In practice, the Random Subspace Method approximates the probability $p_c$ through Monte Carlo sampling. Initially, it generates $N$ groups of features from the original set $\mathbf{x}$ by sampling without replacement according to a uniform distribution $\mathcal{U}(\mathbf{x},k)$. These subsets are represented as a collection of feature groups $G = \{\mathbf{z}_1, \ldots, \mathbf{z}_N\}$. For each of these feature groups $\mathbf{z}_j$, the method employs a base classifier to predict a label. It then counts the occurrences $n_c$ of each possible label $c$ within a predetermined set of labels ${1,2,\ldots, C}$, where $C$ represents the total number of unique labels. The calculation of $n_c$ is formally described by:
\begin{align}
    n_c(\mathbf{x},h,k) = \sum_{j=1}^{N} \mathbb{I}(h(\mathbf{z}_j)=c), c=1,2,\cdots, C,
\end{align}
where $\mathbb{I}$ denotes the indicator function, returning 1 when its condition is met and 0 otherwise. Consequently, the label probability $p_c(\mathbf{x},h,k)$ is estimated by $\frac{n_c(\mathbf{x},h,k)}{N}$. 

\myparatight{Security applications of random subspace method} Random subspace method is recently used to build state-of-the-art certified defenses~\citep{levine2020sparse_cert,zeng2023ranmask,wang2021certify_graph,zhang2023pointcert}. Many previous studies~\citep{levine2020sparse_cert,zhang2023pointcert} showed that the ensemble model built by a random subspace method is certifiably robust against adversarial attacks, i.e., its prediction for a testing input remains unchanged once the $\ell_0$-norm perturbation to the testing input is bounded. Another strand of research~\citep{robey2023smoothllm,cao2023rallm} uses the random subspace method to build robust LLM against jailbreaking attacks, leveraging the fragility of adversarially-generated jailbreaking prompts to perturbations. These methods first use the LLM to generate responses for each of the perturbed input prompts, and these responses are then labeled as either `harmful' or `non-harmful' by checking keywords. Lastly, these labels are aggregated to determine whether the input prompt should be approved or rejected. Existing studies mainly focus on robustness, leaving the explanation of the random subspace method unexplored. Next, we introduce feature attribution for random subspace method.

\vspace{-2mm}
\subsection{Feature Attribution}
\vspace{-2mm}
Feature attribution aims to explain why a model makes a certain prediction for an input by attributing the prediction to the most important features in the input. Existing feature attribution techniques~\citep{lundberg2017SHAP, chen2023estimate_shapley,ribeiro2016lime,paes2024multi_level,mosca2022shap_based,petsiuk2018rise,sundararajan2017integrated_gradient,shrikumar2017deeplift, smilkov2017smoothgrad} fall into two main categories: 1) white-box methods, exemplified by integrated gradients~\citep{sundararajan2017integrated_gradient} and DeepLIFT~\citep{shrikumar2017deeplift}, and 2) black-box methods, including LIME~\citep{ribeiro2016lime} and Shapley values~\citep{lundberg2017SHAP,mosca2022shap_based,chen2023estimate_shapley,sundararajan2020many_shapley}. White-box methods require knowledge of the explained model's architecture, parameters, and gradients, whereas black-box methods do not rely on such detailed knowledge. This study concentrates on black-box feature attribution methods due to their general applicability across various model architectures. 

\myparatight{Attacks to feature attribute methods}
Recent studies~\citep{noppel2023sok,nadeem2023sok} have proposed the \emph{explanation-preserving attack} to feature attribution methods. This attack involves adversarially perturbing the input sample in a manner that induces misclassifications while retaining the original explanation. This attack could be employed to conceal ongoing input manipulation~\citep{zhang2020attack_explanation}. For instance, an attacker could replace certain words in a clean sentence with adversarial alternatives, leading to misclassification, while those words still maintain low relevance in the resulting explanation.

\myparatight{Limitations of existing feature attribute methods}
Existing state-of-the-art black-box feature attribution methods~\citep{lundberg2017SHAP, chen2023estimate_shapley,ribeiro2016lime,enouen2023textgenshap,paes2024multi_level,amara2024syntaxshap,mosca2022shap_based,lopardo2023faithful,petsiuk2018rise} have following limitations. Firstly, they are computationally inefficient when applied to the random subspace method. Techniques such as LIME~\citep{ribeiro2016lime} and Shapley values~\citep{lundberg2017SHAP,enouen2023textgenshap,chen2023estimate_shapley} necessitate a large number of queries (e.g., 1,000) to the black-box model using perturbed versions of the test input. As detailed in Section~\ref{sec:subspace-based-method}, each query to the ensemble classifier requires aggregating the prediction outcomes from all sub-sampled versions of the perturbed test input, leading to prohibitively high computation costs. Secondly, these methods lack theoretical guarantees regarding their performance when subjected to explanation-preserving attacks. Theoretical studies on robust feature attribution~\citep{anani2025certified_explanation1,wang2024certified_explanation2,lin2023certified_explanation3,li2023certified_explanation4, chen2019certified_explanation5, wang2020certified_explanation6} are largely restricted to prediction-preserving attacks, where the adversary seeks to substantially alter the explanation while keeping the classifier’s prediction unchanged. In the next section, we design a feature attribution method that overcomes these limitations.

\section{Algorithm Design}
\label{method}
\subsection{Problem Formulation}
\vspace{-2mm}
 Consider an ensemble model $H$ with base model $h$ and sub-sampling size $k$. For a given test input $\mathbf{x}$, consisting of $d$ features (denoted by $\mathbf{x} = \{x_1, x_2, \cdots, x_d\}$), let $\hat{y}$ represent the predicted label for $\mathbf{x}$, such that $H(\mathbf{x}, h, k) = \hat{y}$. The goal of feature attribution~\citep{sundararajan2017integrated_gradient, paes2024multi_level,amara2024syntaxshap,lopardo2023faithful,chuang2024Large} is to assign an importance score $\alpha^{\hat{y}}_i$ to each element $x_i \in \mathbf{x}$, indicating its contribution to the ensemble model’s prediction of $\hat{y}$. The user can later consider the top-$e$ features with the highest importance scores to be the most important. For instance, if $\mathbf{x}$ is a text consisting of $d$ words, each word would receive an importance score. By ranking these scores, users can easily identify the most influential words leading to the ensemble model's prediction. 

\subsection{Design Goal}

Our approach is guided by three primary design goals. First, the feature attribution method should be computationally efficient, as predictions from an ensemble model are already resource-intensive, so the method must avoid repeatedly using the ensemble model for predictions. Second, it should adhere to key properties of effective feature attribution~\citep{lundberg2017SHAP}, such as local accuracy. Third, the method must be certifiably robust against explanation-preserving attacks. Specifically, if an adversary modifies a small number of features in the input to change the model's prediction, the most important features reported by the attribution method should include these adversarial features.

\subsection{Our Design}
\vspace{-2mm}
Next, we introduce our {\name}. Following existing feature attribution works~\citep{sundararajan2017integrated_gradient, lopardo2023faithful,chuang2024Large}, we design an importance score to measure the contribution of each feature to the model's output label (denoted as ${\hat{y}}$). 
Specifically, we define the important score of the $i$-th feature for the predicted label $\hat{y}$ as:
\begin{align}
    \alpha_i^{\hat{y}}(\mathbf{x},h,k)  = \frac{1}{k}\mathbb{E}_{\mathbf{z} \sim \mathcal{U}(\mathbf{x},k)} [\mathbb{I}(x_i \in \mathbf{z}) \cdot \mathbb{I}(h(\mathbf{z})={\hat{y}} )]. \label{eqn_importance_score}
\end{align}
This importance score of a feature $x_i$ can be seen as the probability that a randomly sampled feature group contains $x_i$ and predicts for $\hat{y}$. The intuition behind this importance value is that the output generated by the ensemble model reflects the aggregated impact of all feature groups. For any given feature group $\mathbf{z}_j\in G$ having size $k$, the contribution of each feature to this group's result is equally divided, amounting to $\frac{1}{k}$ of the group's outcome. If a feature is not in a given group, then the contribution of this feature to this group's result is 0. Consequently, the contribution of a single feature is the aggregate of its contributions across all groups. This intuition leads to the property of local accuracy, which will be discussed in Section~\ref{theory}.

In practice, we use Monte Carlo sampling to approximate the importance scores. We first sub-sample $N$ times to get a set of feature groups, denoted by $G = \{\mathbf{z}_1, \ldots, \mathbf{z}_N\}$, and get the base model's prediction for each of these feature groups. Then the importance score can be naively approximated as $\frac{1}{k\cdot N} \sum_{j=1}^{N}[\mathbb{I}(x_i \in \mathbf{z}_j)\cdot \mathbb{I}(h(\mathbf{z}_j)=\hat{y})]$. When the number of sub-sampled groups (denoted as 
$N$) is large, each feature is likely to appear in a similar number of groups. However, with a smaller $N$, variations in the appearance frequency can result in an unfair assessment of their importance. Features that appear more frequently in sub-sampled feature groups are likely to have greater importance. To solve this problem, we observe that the important score of feature $i$ for the predicted label $\hat{y}$ can be rewritten as:
\begin{align}
    &\alpha_i^{\hat{y}}(\mathbf{x},h,k) = \frac{1}{k} \Pr(x_i \in \mathbf{z})\cdot \Pr(h(\mathbf{z})=\hat{y}|x_i \in \mathbf{z})= \frac{1}{d} \Pr(h(\mathbf{z})=\hat{y}|x_i \in \mathbf{z}),
\end{align}
where $\mathbf{z}$ represents the randomly sub-sampled feature group. Then the importance score can be approximated by:
\begin{align}
    &\alpha^{\hat{y}}_i (\mathbf{x},h,k) \approx \frac{1}{d \cdot \sum_{j=1}^{N} \mathbb{I}(x_i \in \mathbf{z}_j)} \sum_{j=1}^{N}\mathbb{I}(x_i \in \mathbf{z}_j)\cdot \mathbb{I}(h(\mathbf{z}_j)=\hat{y}),\label{eqn_frequency_normalization}
\end{align} 
where $d$ is the total number of features. The introduction of the new normalization term, $\sum_{j=1}^{N} \mathbb{I}(x_i \in \mathbf{z}_j)$, helps to mitigate the issue of unbalanced frequency. This is shown by experiments in Appendix~\ref{appendix_frequency_normalization}.

\myparatight{Computation cost}
Our method utilizes the predictions from the base model for each feature group $z_j \in G$, which are already computed for producing the prediction of the ensemble model. Therefore, given that random subspace method has already been deployed, our method adds negligible additional computational time (around 0.03\,\text{s}). Experimental results can be found in Appendix~\ref{appendix_computation_time}.
\section{Theoretical Analysis}\label{theory}
In this section, we begin by establishing the predicted label probability $p_{\hat{y}}$ on perturbed testing inputs to support subsequent theoretical analysis. Then we demonstrate that our method adheres to fundamental properties for effective feature attribution. Finally, we provide theoretical guarantees regarding our method's performance under attacks to feature attribution. For simplicity, we abuse the notation and use $i$ to represent the feature $x_i$ for theoretical analysis.

\vspace{-2mm}
\subsection{Define $p_{\hat{y}}$ on Feature Subsets}\label{def-subset}
\vspace{-2mm}
Before theoretical analysis, we first define the predicted label probability $p_{\hat{y}}$ when a subset of features $S\subseteq \mathbf{x}$ is present. In this case, random subspace method sub-samples feature groups with size $k$ from $S$. Particularly, given any feature subset $S$, we define the probability that the label $\hat{y}$ is predicted by the base model (when features not in $S$ are removed) as:
\begin{align}p_{\hat{y}}(S,h,k) = \mathbb{E}_{\mathbf{z}\sim \mathcal{U}(S,k)}[\mathbb{I}(h(\mathbf{z})=\hat{y})] ,\end{align} where $\mathbf{z} \sim \mathcal{U}(S,k)$ is a subset of $S$ with size $k$ that is randomly sampled from the uniform distribution. i.e., $\Pr(\mathbf{z} = \mathbf{z}' \mid \mathbf{z} \sim \mathcal{U}(S,k)) = \frac{1}{\binom{|S|}{k}}$ for any $\mathbf{z}'\subseteq S$. We note that there is a special case when $|S|<k$, which means that intended random sub-sampling of $k$ features from 
$S$ cannot proceed as usual. To address this, we let $p_{\hat{y}}(S,h,k) = \frac{1}{C}$, which means that the base model randomly guesses the label. We note that this assumption is necessary if we want to define Shapley value for random subspace method, because it is impossible to sub-sample $k$
 features from less than 
 $k$ features. In the following section, we utilize this definition to establish a Shapley value for random subspace method.
\vspace{-2mm}
\subsection{Shapley Value Based Explanation for Random Subspace Method}\label{theory-define-shapley}
\vspace{-2mm}
Derived from game theory~\citep{shapley1953value}, Shapley values are intended for credit assignment among players in cooperative games. A game is represented by a set of players $D$ and a value function $v(S): \mathcal{P}(D) \rightarrow \mathbb{R}$, where $\mathcal{P}(D)$ means the power set of $D$. The Shapley value for player $i$ is defined as: \begin{align}\phi_i(v) = &\sum_{S \subseteq D\setminus\{i\}} \frac{|S|! (d - |S| - 1)!}{d!} (v(S \cup \{i\}) - v(S)).
\end{align}
Shapley value has long been regarded as the gold standard for feature attribution~\citep{lundberg2017SHAP,mosca2022shap_based,chen2023estimate_shapley, sundararajan2020many_shapley, paes2024multi_level,amara2024syntaxshap,sundararajan2017integrated_gradient}. In order to explain the output of a machine learning model, many existing works~\citep{paes2024multi_level,amara2024syntaxshap,sundararajan2017integrated_gradient} use the probability of the model's output as the value function. Similarly, we can define a Shapley value for random subspace method. Specifically, we let the label probability $p_{\hat{y}}$ be the value function $v$ and let the input feature set $\mathbf{x}$ be the set of players $D$. Then the Shapley value for feature $i$, denoted as $\phi_i(p_{\hat{y}})$, can be written as:
\begin{align}\sum_{S \subseteq \mathbf{x}\setminus\{i\}} \frac{|S|! (d - |S| - 1)!}{d!} (p_{\hat{y}}(S \cup \{i\},h,k) - p_{\hat{y}}(S,h,k)).
\end{align} This value is empirically challenging to compute because $p_{\hat{y}}$ should be evaluated on all feature subsets, while evaluating $p_{\hat{y}}$ on a single feature subset requires $N$ forward passes of the base model. In the next part, we demonstrate that our computationally efficient importance score maintains the key properties of Shapley value.
\vspace{-1mm}
\subsection{Properties of {\name}}
\vspace{-1mm}
{\name} possesses \emph{local accuracy} and \emph{symmetry} as derived from Shapley values~\citep{lundberg2017unified, chen2023algorithms}, whilst substituting the remaining two properties inherent to Shapley values, specifically \emph{dummy} and \emph{linearity}, with \emph{order consistency} (with Shapley value). The linearity property is omitted because its application is not straightforward in the context of subspace methods. Furthermore, the relaxation of the dummy property is from the observation that in many cases, people are more interested in the comparative importance of features over their absolute importance scores~\citep{lopardo2023faithful, xue2024stability}. We introduce these properties below.
 
The first property is \emph{local accuracy}. This property ensures that the explanation accurately reflects the behavior of the ensemble model for the testing input $\mathbf{x}$. It can be formally stated as follows.
\begin{property} (Local accuracy). 
\label{property-local}
{For any $\mathbf{x}$, $h$, and $k$, the importance score of all features sum up to $p_{\hat{y}}(\mathbf{x}, h, k)$, i.e., $\sum_{i\in \mathbf{x}} \alpha^{\hat{y}}_i(\mathbf{x},h,k) = p_{\hat{y}}(\mathbf{x},h,k)$.}
\end{property}
\vspace{-2mm}

The second property is \emph{symmetry}. The symmetry property states that if two features contribute equally to all possible feature subsets $S\subseteq \mathbf{x}$, then feature $i$ and $j$ should receive the same importance score.
\begin{property} \label{property-symmetry}(Symmetry). Given a pair of features $(i,j)$, if for any $S\subseteq \mathbf{x}\setminus \{i,j\}$, $p_{\hat{y}}(S\cup \{i\}, h, k) = p_{\hat{y}}(S \cup \{j\}, h, k)$, then $\alpha^{\hat{y}}_i(\mathbf{x},h,k) = \alpha^{\hat{y}}_j(\mathbf{x},h,k)$.
\end{property}

The third property is \emph{order consistency} (with Shapley value). This property ensures that if Shapley value ranks a feature as more significant, our attribution approach will also give it a higher importance. The Shapley value for random subspace method is defined in Section~\ref{theory-define-shapley}.
\begin{property} \label{property-order-consistency}(Order consistency with Shapley value). Given a pair of features $(i,j)$, $\alpha^{\hat{y}}_i(\mathbf{x},h,k)\geq \alpha^{\hat{y}}_j(\mathbf{x},h,k)$ if and only if $\phi_i(p_{\hat{y}})\geq \phi_j(p_{\hat{y}})$, where $\phi_i(p_{\hat{y}})$ and $\phi_j(p_{\hat{y}})$ respectively represent Shapley values of $i$ and $j$.
\end{property}

We provide the proof details in Appendix~\ref{proof_properties}. Our method essentially relaxes the dummy property of Shapley value to simplify its complex computation. Despite this alteration, the utility of the Shapley value is preserved in most scenarios due to the property of order consistency. This observation is supported by these commonly used metrics for feature attribution, such as the fidelity score~\citep{miro2024comprehensive,chuang2024Large}, perturbation curves~\citep{paes2024multi_level,chen2020hierarchical} and faithfulness~\citep{lopardo2023faithful}. These metrics rely on the relative order of importance scores rather than their absolute values.
\vspace{-1mm}
\subsection{Certified Detection of Adversarial Features}\label{theory-detect}
\vspace{-1mm}
In this part, we demonstrate that our explanation method provably detects adversarial features that causes model misclassifiation, therefore is provably secure against \emph{explanation-preserving attacks}. We suppose the attacker can modify at most $T$ features of the original testing input $\mathbf{x}$ to change the predicted label of the ensemble classifier. We denote the set of all possible perturbed test inputs $\mathbf{x}'$ as $\mathcal{B}(\mathbf{x},T)$, and we use $\mathbf{x}\ominus\mathbf{x}'$ to denote the set of modified features. Here, we focus on top-$e$ most important features reported by our method. i.e., $e$ features with highest importance scores for the predicted label. We denote this set of features before attack as $E(\mathbf{x})$, and use $E(\mathbf{x}')$ to represent the new set of top-$e$ most important features for $\mathbf{x}'$. Our goal is to derive the \emph{certified detection size} $\mathcal{D}(\mathbf{x},T)$, which is the intersection size lower bound between the set of modified features and the set of reported important features, which is formally defined as:
\begin{align}\mathcal{D}(\mathbf{x},T) =& \argmax_r, s.t. |(\mathbf{x}'\ominus \mathbf{x}) \cap E(\mathbf{x}')|\geq r, \\
&\forall \mathbf{x}' \in \mathcal{B}(\mathbf{x},T), H(\mathbf{x}')\neq H(\mathbf{x}).
\end{align}
We have the following result:
\begin{theorem} Given a testing input $\mathbf{x}$ which is originally predicted as $\hat{y}$. We suppose there exists $\mathbf{x}'\in \mathcal{B}(\mathbf{x},T)$ such that $H(\mathbf{x}')\neq \hat{y}$. Then $\mathcal{D}(\mathbf{x}, T)$ is the solution of the following optimization problem:
\begin{align}
&\mathcal{D}(\mathbf{x},T) = \argmax_r r,\text{ } s.t.\text{ }\forall \hat{y}'\neq \hat{y}:\text{ } \label{eqn_certified_detection}\\
  &\frac{1}{T-r+1}\cdot [\frac{\Delta}{2k} -\frac{r-1}{d} +\sum_{i = {d-T+r}}^{d} \underline{\alpha}^{\hat{y}'}_{q_i}(\mathbf{x},h,k)]\geq \overline{\alpha}^{\hat{y}'}_{w_{e-r+1}}(\mathbf{x},h,k)+\frac{1}{d}-\frac{1}{k}\frac{\binom{d-1-T}{k-1}}{\binom{d}{k}}\\
&\text{or}\\
&\frac{\Delta}{2k}(\frac{1}{T-r+1}-\frac{k-1}{e-r+1}) \\
&\geq \frac{1}{e-r+1}\sum_{i =1}^{e-r+1} \overline{\alpha}^{\hat{y}'}_{w_i}(\mathbf{x},h,k)+\frac{r-1}{d\cdot(T-r+1)}  -\frac{1}{T-r+1} \sum_{i = d-T+r}^{d}\underline{\alpha}^{\hat{y}'}_{q_i}(\mathbf{x},h,k) 
\end{align}
where $\Delta = {\underline{p}_{\hat{y}}(\mathbf{x},h,k)-\overline{p}_{\hat{y}'}(\mathbf{x},h,k)}$, $\underline{p}_{c}$ (or $\overline{p}_{c}$) represents the probability lower (or upper) bound of some label $c\in [C]$, $\underline{\alpha}^{\hat{y}'}_{i}(\mathbf{x},h,k)$ (or $\overline{\alpha}^{\hat{y}'}_{i}(\mathbf{x},h,k)$) represents the lower (or upper) bound of the feature $i$'s importance score for some label $\hat{y}' \neq \hat{y}$, $\{w_1, \cdots , w_{d}\}$ denotes the set of all features in descending order of the important value upper bound $\overline{\alpha}^{\hat{y}'}(\mathbf{x},h,k)$, i.e., $\overline{\alpha}_{w_1}^{\hat{y}'}(\mathbf{x},h,k)\geq \overline{\alpha}_{w_2}^{\hat{y}'}(\mathbf{x},h,k)\geq \cdots \geq \overline{\alpha}_{w_d}^{\hat{y}'}(\mathbf{x},h,k)$, and $\{q_1, \cdots , q_{d}\}$ denotes the set of all features in descending order of the important value lower bound $\underline{\alpha}^{\hat{y}'}(\mathbf{x},h,k)$, i.e., $\underline{\alpha}_{q_1}^{\hat{y}'}(\mathbf{x},h,k)\geq \underline{\alpha}_{q_2}^{\hat{y}'}(\mathbf{x},h,k)\geq \cdots \geq \underline{\alpha}_{q_d}^{\hat{y}'}(\mathbf{x},h,k)$. 
\end{theorem}
In practice, the maximum $r$ is found by binary search. The specifics for computing $\underline{p}_{\hat{y}}$, $\overline{p}_{\hat{y}'}$, $\underline{\alpha}^{\hat{y}'}_{i}(\mathbf{x},h,k)$, and $\overline{\alpha}^{\hat{y}'}_{i}(\mathbf{x},h,k)$ can be found in Appendix~\ref{theory-confidence-bounds}, and the proof is available in Appendix~\ref{proof_detection}. The proof intuition is that to change the label from $\hat{y}$ to $\hat{y}'$, the attacker must ensure that more feature groups predict for $\hat{y}'$. However, the attacker can only alter the predicted labels of feature groups that include at least one feature in $\mathbf{x}\ominus\mathbf{x}'$. Consequently, the importance values of features within $\mathbf{x}\ominus\mathbf{x}'$ are likely to increase, making them more detectable. We provide more discussion in Appendix \ref{appendix_theory_discussion}.

\section{Evaluation on Security Applications}
\vspace{-2mm}
\label{evaluation}
We evaluate the effectiveness of our method for certified defense and defense against jailbreaking attacks. For certified defense, we employ a backdoor attack (BadNets~\citep{gu2017badnets}) and an adversarial attack (TextFooler~\citep{jin2020textfooler}) to challenge the random subspace method (more details are provided in Appendix~\ref{sup-adversarial-attack}). We show that our method successfully identifies the exact words responsible for the failure of certified defense. 
For defense against jailbreaking attacks, we evaluate three types of such attacks: GCG~\citep{zou2023GCG}, AutoDAN~\citep{liu2023autodan}, and DAN~\citep{liu2023autodan}. We show that our method is capable of identifying the harmful query embedded within the jailbreaking prompt. 
\vspace{-1mm}
\subsection{Experimental Setup}
\vspace{-1mm}
\myparatight{Random Subspace Method Implementation}
For certified defense, we follow RanMASK~\citep{zeng2023ranmask} for constructing the ensemble classifier. For defense against jailbreaking attacks, we adopt the RA-LLM~\citep{cao2023rallm} framework. More details are provided in Appendix~\ref{sup-jailbreaking}.

\myparatight{Datasets} We use classification datasets such as SST-2~\citep{socher2013sst2}, IMDB~\citep{maas2011imdb}, and AGNews~\citep{zhang2015agnews} for the study on certified defense mechanisms, and use harmful behaviors dataset~\citep{zou2023GCG} for defense against jailbreaking attacks. More details can be found in Appendix~\ref{sup-dataset}.

\myparatight{Models}
For certified defense, we use a pretrained BERT model~\citep{devlin2018bert} as our base model and fine-tune it using AdamW optimizer for $10$ epochs on masked training samples to improve the certification performance. 
The learning rate is set to $1\times 10^{-5}$. For defense against jailbreaking attacks, we directly use Vicuna-7B~\citep{chiang2023vicuna} as our base model.

\myparatight{Hyper-parameters}
Unless specifically mentioned, we use following hyperparameters by default. For certified defense, the dropping rate (expressed as $\rho = 1- \frac{k}{d}$) is set to $0.8$, and $N$ is set to $1,000$. For defense against jailbreaking attacks, we set the dropping rate to $0.4$, $N$ to $500$, and the threshold $\tau$ to $0.1$. The impact of these hyperparameters will be explored in an ablation study.

\myparatight{Evaluation Metrics}
We use the following metrics. The faithfulness metric is reported across all our experiments. Furthermore, in instances where there is ground-truth information regarding the key words that significantly influence the prediction of the ensemble model (e.g., during empirical attacks such as backdoor attacks), we implement extra metrics for predicting these key words. 
We denote the test dataset by $\mathcal{D}_{test}$, the base model by $h$ and the the prediction of the ensemble for some test sample $\mathbf{x}$ by $H(\mathbf{x})$.
     
     $\bullet$\myparatight{Faithfulness~\citep{lopardo2023faithful}}
    We define the faithfulness of the feature attribution as the percentage of label filps when the $e$ features with highest importance scores are deleted. We use $E(\mathbf{x})$ to denote the $e$-most important features reported by the feature attribution method. Then faithfulness is represented by: $\frac{1}{|\mathcal{D}_{test}|}\sum_{\mathbf{x} \in \mathcal{D}_{test}} \mathbb{I}[H(\mathbf{x})\neq H(\mathbf{x}\setminus E(\mathbf{x}))]$.

    $\bullet$ \myparatight{Key word prediction}
    We define a set of ground-truth important words denoted by $L(\mathbf{x})$. We let the feature attribution method identify the top $e$ most crucial words and measure the intersection of these words with the set of ground-truth important words. Specifically, we have \emph{top-$e$ precision}$ = \frac{|E(\mathbf{x})\cap L(\mathbf{x})|}{e}$, and \emph{top-$e$ recall}$ = \frac{|E(\mathbf{x})\cap L(\mathbf{x})|}{|L(\mathbf{x})|}$. As our final result, we report
the average values of \emph{top-$e$ precision} and \emph{top-$e$ recall} computed on $\mathcal{D}^*_{test}$ for different $e$ values. $\mathcal{D}^*_{test}$ is a specific subset of $\mathcal{D}_{test}$ detailed in Appendix~\ref{sup-metric}.

$\bullet$ \myparatight{Certified detection rate} We develop metrics for provable defense against explanation-preserving attacks discussed in Section~\ref{theory-detect}.
We define \emph{certified detection rate} as $\mathcal{D}(\mathbf{x},T)/T$ to measure the percentage of detected adversarial features. We report the mean values of certified detection rate computed on $\mathcal{D}_{test}$.

\myparatight{Compared Methods}
We compare our method with the following baseline methods. Shapley value~\citep{chen2023estimate_shapley} and LIME~\citep{ribeiro2016lime} are state-of-the-art techniques in feature attribution but present computational challenges when applied directly to ensemble models. Consequently, we implement these methods on the base model. 
Furthermore, we have adapted the ICL method~\citep{kroeger2023icl} for feature attribution purposes. This approach leverages the in-context learning capabilities of large language models (LLMs). We provide implementation details of these methods in Appendix~\ref{sup-baseline}.

\begin{table*}[t]
  \centering
  \begin{minipage}[t]{0.49\textwidth}
    \centering
    \captionof{table}{Compare the faithfulness of our method with baselines for certified defense. 
             We delete different ratios of most important words and compute the rate of 
             label changes.}
    \vspace{0.5em}
    \resizebox{\linewidth}{!}{
      {\scriptsize
      \begin{tabular}{c|c|c|c|c|c|c|c}
        \toprule
        \multirow{2}{*}{\makecell{Defense \\scenarios}} & Dataset &
        \multicolumn{2}{c|}{SST-2} & \multicolumn{2}{c|}{IMDb} &
        \multicolumn{2}{c}{AG-news} \\ \cmidrule{2-8}
        & Ratio & 10\% & 20\% & 10\% & 20\% & 10\% & 20\% \\ \midrule
        \multirow{4}{*}{No attack}
          & Shapley & 0.320 & 0.530 & 0.300 & 0.330 & 0.150 & 0.280 \\
          & LIME    & 0.125 & 0.145 & 0.060 & 0.095 & 0.020 & 0.035 \\
          & ICL     & 0.095 & 0.135 & 0.045 & 0.050 & 0.030 & 0.040 \\
          & Ours    & \textbf{0.365} & \textbf{0.605} & \textbf{0.600} &
                      \textbf{0.745} & \textbf{0.175} & \textbf{0.410} \\ \midrule
        \multirow{4}{*}{\makecell{Backdoor \\attack}}
          & Shapley & 0.380 & 0.630 & 0.520 & 0.540 & 0.725 & 0.790 \\
          & LIME    & 0.080 & 0.095 & 0.120 & 0.180 & 0.205 & 0.300 \\
          & ICL     & 0.055 & 0.085 & 0.120 & 0.170 & 0.140 & 0.235 \\
          & Ours    & \textbf{0.400} & \textbf{0.655} & \textbf{0.810} &
                      \textbf{0.910} & \textbf{0.735} & \textbf{0.795} \\ \midrule
        \multirow{4}{*}{\makecell{Adv.\\attack}}
          & Shapley & 0.600 & 0.840 & 0.845 & 0.840 & 0.850 & 0.960 \\
          & LIME    & 0.100 & 0.160 & 0.280 & 0.335 & 0.200 & 0.265 \\
          & ICL     & 0.130 & 0.170 & 0.305 & 0.365 & 0.115 & 0.130 \\
          & Ours    & \textbf{0.680} & \textbf{0.880} & \textbf{0.980} &
                      \textbf{1.000} & \textbf{0.905} & \textbf{0.970} \\
        \bottomrule
      \end{tabular}}}
    \label{tab-certified-faithfulness}
  \end{minipage}
  \hfill
  \begin{minipage}[t]{0.49\textwidth}
    \centering
    \captionof{table}{Compare the key-word prediction performance of our method with baselines for certified defense.  
             Each method reports the top-5 important words ($e=5$).}
    \vspace{0.5em}
    \resizebox{\linewidth}{!}{
      {\scriptsize
      \begin{tabular}{c|c|c|c|c|c|c|c}
        \toprule
        \multirow{2}{*}{\makecell{Defense \\scenarios}} & Dataset &
        \multicolumn{2}{c|}{SST-2} & \multicolumn{2}{c|}{IMDb} &
        \multicolumn{2}{c}{AG-news} \\ \cmidrule{2-8}
        & Metric & Prec. & Rec. & Prec. & Rec. & Prec. & Rec. \\ \midrule
        \multirow{4}{*}{\makecell{Backdoor \\attack}}
          & Shapley & 0.543 & 0.904 & 0.295 & 0.491 & 0.523 & 0.872 \\
          & LIME    & 0.148 & 0.247 & 0.037 & 0.022 & 0.073 & 0.122 \\
          & ICL     & 0.087 & 0.145 & 0.030 & 0.049 & 0.068 & 0.113 \\
          & Ours    & \textbf{0.585} & \textbf{0.975} & \textbf{0.535} &
                      \textbf{0.892} & \textbf{0.557} & \textbf{0.929} \\ \midrule
        \multirow{4}{*}{\makecell{Adv. \\attack}}
          & Shapley & 0.361 & 0.680 & 0.282 & 0.142 & 0.528 & 0.343 \\
          & LIME    & 0.146 & 0.319 & 0.067 & 0.025 & 0.242 & 0.128 \\
          & ICL     & 0.098 & 0.210 & 0.076 & 0.040 & 0.080 & 0.046 \\
          & Ours    & \textbf{0.378} & \textbf{0.717} & \textbf{0.384} &
                      \textbf{0.184} & \textbf{0.530} & \textbf{0.356} \\
        \bottomrule
      \end{tabular}}}
    \label{tab-certified-key-word-prediction-e=5}
  \end{minipage}
  \vspace{-3mm}
\end{table*}
\begin{table*}[t]
  \centering
  \begin{minipage}[t]{0.48\textwidth}
    \centering
    \captionof{table}{Compare the faithfulness of our method with baselines for defense against jailbreaking attacks.\label{tab-jailbreaking-faithfulness}}
    \vspace{-0.5em}
    \resizebox{\linewidth}{!}{
      {\scriptsize
      \begin{tabular}{c|c|c|c|c|c|c}
        \toprule
        Attack & \multicolumn{2}{c|}{GCG} & \multicolumn{2}{c|}{AutoDAN} &
                \multicolumn{2}{c}{DAN} \\ \midrule
        Del.\ Ratio & 10\% & 20\% & 10\% & 20\% & 10\% & 20\% \\ \midrule
        Shapley & 0.11 & 0.19 & 0.15 & 0.18 & 0.33 & 0.33 \\
        LIME    & \textbf{0.15} & 0.23 & 0.34 & 0.32 & 0.54 & 0.38 \\
        ICL     & 0 & 0 & 0.08 & 0.11 & 0.24 & 0.27 \\
        Ours    & \textbf{0.15} & \textbf{0.24} & \textbf{0.38} & \textbf{0.46} &
                  \textbf{0.85} & \textbf{0.74} \\
        \bottomrule
      \end{tabular}}}
  \end{minipage}
  \hfill
  \begin{minipage}[t]{0.50\textwidth}
    \centering
    \captionof{table}{Compare the keyword prediction performance of our method with baselines for defense against jailbreaking attacks ($e=10$). \label{tab-jailbreaking-key-word-prediction-e=10}}
    \vspace{-0.5em}
    \resizebox{\linewidth}{!}{
      {\scriptsize
      \begin{tabular}{c|c|c|c|c|c|c}
        \toprule
        Attack & \multicolumn{2}{c|}{GCG} & \multicolumn{2}{c|}{AutoDAN} &
                \multicolumn{2}{c}{DAN} \\ \midrule
        Metric & Prec. & Rec. & Prec. & Rec. & Prec. & Rec. \\ \midrule
        Shapley & 0.651 & 0.571 & 0.306 & 0.260 & 0.137 & 0.119 \\
        LIME    & 0.654 & 0.575 & 0.335 & 0.281 & 0.332 & \textbf{0.289} \\
        ICL     & 0.544 & 0.466 & 0.252 & 0.212 & 0.078 & 0.064 \\
        Ours    & \textbf{0.664} & \textbf{0.584} & \textbf{0.434} & \textbf{0.379} &
                  \textbf{0.378} & 0.287 \\
        \bottomrule
      \end{tabular}}}
    \vspace{-5mm}
  \end{minipage}
\end{table*}
\vspace{-1mm}

\vspace{-2mm}
\subsection{Experimental Results}
\vspace{-2mm}

\subsubsection{Explain Certified Defense}
\vspace{-2mm}
We evaluate our method's explanation effectiveness both in the absence of attacks and in scenarios where the certified defense is compromised by strong empirical attacks.

\myparatight{No Attack} In Table~\ref{tab-certified-faithfulness}, we present a comparison of our method's faithfulness against other baseline methods for clean test samples. We can see that our method surpasses all baselines in performance. A visualization for IMDb dataset is provided by Figure~\ref{fig-visualization-imdb-ours} in the Appendix.

\myparatight{Backdoor Attack and Adversarial Attack} Table~\ref{tab-empirical-ASR} in Appendix shows that a significant proportion of testing samples can be compromised when the attacker could maliciously insert (or alter) a relatively large number of words. Table~\ref{tab-certified-faithfulness} details our method's capability in explaining model behavior to sentences altered by the backdoor (or adversarial) attack. Specifically, for the adversarial attack on the IMDb dataset, removing the 10\% of the words considered most critical by our method results in a label change for 98\%. Additionally, Table~\ref{tab-certified-key-word-prediction-e=5} and Table~\ref{tab-certified-key-word-prediction-e=10} (in Appendix) provides evidence that, in scenarios where these altered sentences misguide the ensemble model towards incorrect predictions, our method exhibits superior capability in detecting the backdoor triggers (or adversarial words) responsible for the ensemble model's erroneous behavior. For example, on IMDB dataset, our technique achieves a recall of $0.892$, significantly higher than the $0.491$ recall obtained using Shapley value for backdoor attacks when $e = 5$. For a qualitative comparison, please see Figure~\ref{fig-visualization-imdb-shap} and Figure~\ref{fig-visualization-imdb-ours} in the Appendix.

\subsubsection{Explain Defense Against Jailbreaking Attacks}
\vspace{-1mm}
In this part, we demonstrate that our method enhances understanding of the decision-making processes of the RA-LLM~\citep{cao2023rallm} when faced with jailbreaking prompts.
Table~\ref{tab-jailbreaking-faithfulness} shows that our method outperforms baselines in identifying the most important words that influence the RA-LLM's decisions for jailbreaking prompts. Table~\ref{tab-jailbreaking-key-word-prediction-e=10} and Table~\ref{tab-jailbreaking-key-word-prediction-e=20} (in Appendix) demonstrate that when a jailbreaking prompt is detected as `harmful' by RA-LLM, our method is capable of identifying the harmful query embedded within the jailbreaking prompt that leads to this decision. This finding is also supported by the qualitative results shown in Figure~\ref{fig-visualization-harmful-behavior-ours} in Appendix.
\vspace{-1mm}
\subsubsection{Impact of Hyperparameters}
\vspace{-1mm}
We examine how the dropping rate $\rho$ and the number of sub-sampled inputs $N$ influence our method's faithfulness and key word prediction performance. Figures \ref{fig-ablation-N-certified-defense-faithfulness} and \ref{fig-ablation-N-certified-defense-f1} in Appendix demonstrates that both metrics generally improves with an increase in $N$, as it leads to a more precise estimation of importance values. Furthermore, Figures \ref{fig-ablation-rho-certified-defense-faithfulness} and \ref{fig-ablation-rho-certified-defense-f1} in Appendix reveals that while key word prediction performance remains stable, there is a decline in faithfulness at a very large $\rho$ value (e.g., $\rho = 0.9$). This is because the ensemble model becomes insensitive to the deletion of important features at higher dropping rates.  We have consistent findings for defense against jailbreaking attacks, as illustrated in Figure~\ref{fig-ablation-jailbreaking-attack} in Appendix.
\vspace{-2mm}
\begin{figure*}[t]
\vspace{-3mm}
\centering
\subfloat[SST-2]{\includegraphics[width=0.3\textwidth]{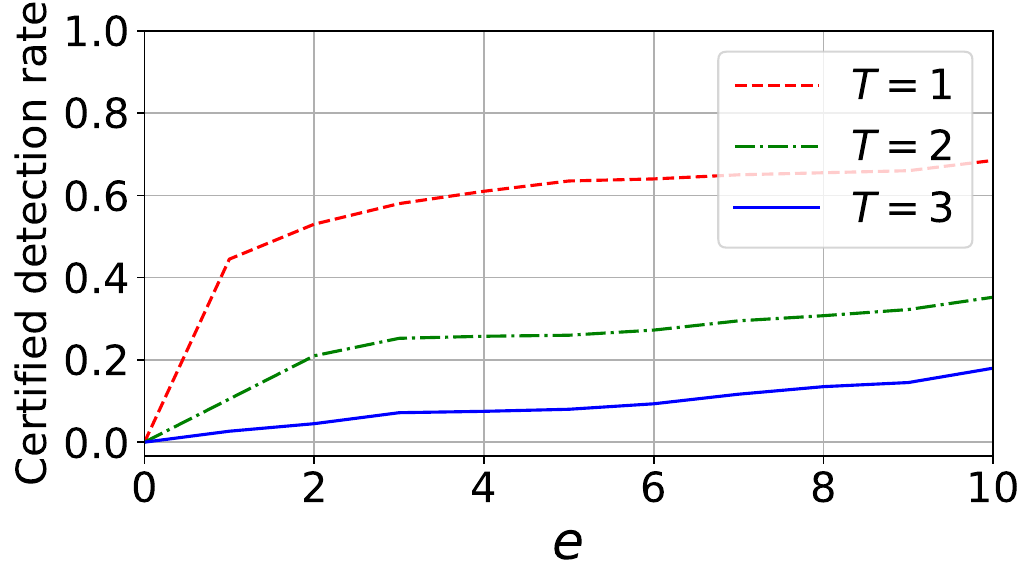}}\hspace{4mm}
\subfloat[IMDb]{\includegraphics[width=0.3\textwidth]{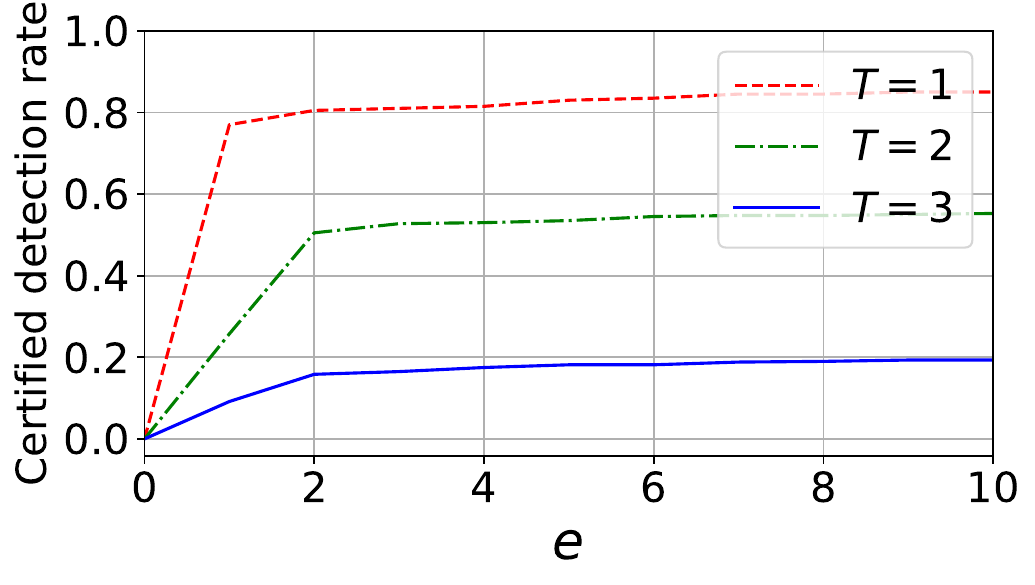}}\hspace{4mm}
\subfloat[AG-news]{\includegraphics[width=0.3\textwidth]{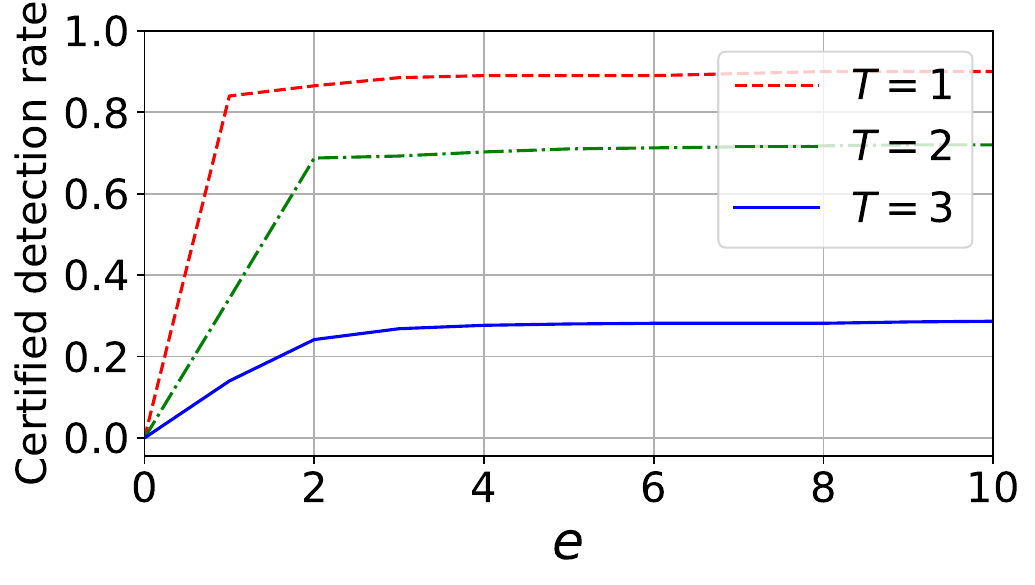}}
\vspace{-2mm}

\caption{Certified detection rate on text classification datasets. $T$ is the number of modified features, and $e$ is the number of reported most important features.}
\label{fig-certified-detection}
\vspace{-5mm}
\end{figure*}
\subsection{Certified Detection of Adversarial Features}
We evaluate the certified detection rate of our feature attribution on text classification datasets. By default, we set the certification confidence $1-\beta$ to 0.99, the dropping rate $\rho$ to 0.8, and the sub-sampling number $N$ to 10,000. Figure~\ref{fig-certified-detection} shows the results in default setting. We find that the certified detection rate improves as the explanation reports more features as important features, and the rate decreases when the attacker is able to modify a greater number of features. Figure~\ref{fig-ablation-beta-certified-detection}, Figure~\ref{fig-ablation-N-certified-detection}, and Figure~\ref{fig-ablation-rho-certified-detection} in Appendix shows the impact of $\beta$, $N$ and $\rho$, respectively. We find that while the certified detection rate is insensitive to the $\beta$ value, it can be significantly enhanced by increasing $N$, or $\rho$.  
\myparatight{Computation cost} Both feature attribution and certified detection with our method involve minimal additional computational cost (less than 0.5 seconds), as demonstrated in Appendix~\ref{appendix_computation_time}.

\vspace{-2mm}
\subsection{Application in image domain} Our method is also applicable to the image domain for defending against adversarial patch attacks~\cite{levine2020derandomized, brown2017adversarial_patch}. The details can be found in Appendix~\ref{appendix_image_domain}.

\vspace{-2mm}
\section{Conclusion and Future Work}
\vspace{-2mm}
In this work, we propose an efficient and provably robust feature attribution method for random subspace method. Potential future directions include: 1) investigating the explanation of random subspace method for privacy applications, such as machine unlearning~\cite{bourtoule2021machine_unlearning} and differential privacy~\cite{liu2020bagging_dp}; and 2) developing provably secure feature attribution methods for general machine learning models.
\label{conclusion}

\bibliography{refs}
\bibliographystyle{iclr2026_conference}

\appendix
\appendix
\newpage
\section{Proofs for Properties}\label{proof_properties}
To simply notation, for all the following proofs, we use $\mathbf{z}$ to denote subsets of $\mathbf{x}$ with size $k$. Next, we provide our proof for each property.

\myparatight{Property~\ref{property-local} (Local Accuracy)}
{For any $\mathbf{x}$, $h$, and $k$, the importance score of all features sum up to $p_{\hat{y}}(\mathbf{x}, h, k)$, i.e., $\sum_{i\in \mathbf{x}} \alpha^{\hat{y}}_i(\mathbf{x},h,k) = p_{\hat{y}}(\mathbf{x},h,k)$.}

\begin{proof} 
\begin{align}
    \sum_{i\in \mathbf{x}}\alpha_i^{\hat{y}}(\mathbf{x},h,k)  
    = &\sum_{i\in \mathbf{x}}\frac{1}{k}\mathbb{E}_{\mathbf{z} \sim \mathcal{U}(\mathbf{x},k)} [\mathbb{I}(i \in \mathbf{z}) \cdot \mathbb{I}(h(\mathbf{z})={\hat{y}} )]\\
    =&\frac{1}{k}\mathbb{E}_{\mathbf{z} \sim \mathcal{U}(\mathbf{x},k)} [\mathbb{I}(h(\mathbf{z})={\hat{y}} )\cdot \sum_{i\in \mathbf{x}} \mathbb{I}(i \in \mathbf{z})]\\
    =&\mathbb{E}_{\mathbf{z} \sim \mathcal{U}(\mathbf{x},k)} \mathbb{I}(h(\mathbf{z})={\hat{y}} )\\
    =& p_{\hat{y}}(\mathbf{x},h,k)
\end{align}
\end{proof}

\myparatight{Property~\ref{property-symmetry} (Symmetry)} Given a pair of features $(i,j)$, if for any $S\subseteq \mathbf{x}\setminus \{i,j\}$, $p_{\hat{y}}(S\cup \{i\}, h, k) = p_{\hat{y}}(S \cup \{j\}, h, k)$, then $\alpha^{\hat{y}}_i(\mathbf{x},h,k) = \alpha^{\hat{y}}_j(\mathbf{x},h,k)$.

\begin{proof} \small We let $S = \mathbf{x}-\{i,j\}$. Then we have:
\begin{align}
p_{\hat{y}}(\mathbf{x}-\{j\}, h, k) &= p_{\hat{y}}(\mathbf{x}-\{i\}, h, k)\\
\frac{1}{\binom{d-1}{k}}\sum_{\mathbf{z}\subseteq 
\mathbf{x},j\notin \mathbf{z}}\mathbb{I}(h(\mathbf{z})=\hat{y}) &= \frac{1}{\binom{d-1}{k}}\sum_{\mathbf{z}\subseteq 
\mathbf{x},i\notin \mathbf{z}}\mathbb{I}(h(\mathbf{z})=\hat{y})\\
 \\
\sum_{\mathbf{z}\subseteq 
\mathbf{x},j\notin \mathbf{z}}\mathbb{I}(h(\mathbf{z})=\hat{y})- \sum_{\mathbf{z}\subseteq 
\mathbf{x},j\notin \mathbf{z},i\notin \mathbf{z}}\mathbb{I}(h(\mathbf{z})=\hat{y})&= \sum_{\mathbf{z}\subseteq 
\mathbf{x},i\notin \mathbf{z}}\mathbb{I}(h(\mathbf{z})=\hat{y})- \sum_{\mathbf{z}\subseteq 
\mathbf{x},j\notin \mathbf{z},i\notin \mathbf{z}}\mathbb{I}(h(\mathbf{z})=\hat{y})\\
\sum_{\mathbf{z}\subseteq 
\mathbf{x},j\notin \mathbf{z},i\in \mathbf{z}}\mathbb{I}(h(\mathbf{z})=\hat{y})&= \sum_{\mathbf{z}\subseteq 
\mathbf{x},j\in \mathbf{z},i\notin \mathbf{z}}\mathbb{I}(h(\mathbf{z})=\hat{y})\\
\sum_{\mathbf{z}\subseteq 
\mathbf{x},j\notin \mathbf{z},i\in \mathbf{z}}\mathbb{I}(h(\mathbf{z})=\hat{y})+\sum_{\mathbf{z}\subseteq 
\mathbf{x},j\in \mathbf{z},i\in \mathbf{z}}\mathbb{I}(h(\mathbf{z})=\hat{y})&= \sum_{\mathbf{z}\subseteq 
\mathbf{x},j\in \mathbf{z},i\notin \mathbf{z}}\mathbb{I}(h(\mathbf{z})=\hat{y})+\sum_{\mathbf{z}\subseteq 
\mathbf{x},j\in \mathbf{z},i\in \mathbf{z}}\mathbb{I}(h(\mathbf{z})=\hat{y})\\
\sum_{\mathbf{z}\subseteq 
\mathbf{x},i\in \mathbf{z}}\mathbb{I}(h(\mathbf{z})=\hat{y})&= \sum_{\mathbf{z}\subseteq 
\mathbf{x},j\in \mathbf{z}}\mathbb{I}(h(\mathbf{z})=\hat{y})\\
\alpha^{\hat{y}}_i(\mathbf{x},h,k) &= \alpha^{\hat{y}}_j(\mathbf{x},h,k)
\end{align}
\end{proof}

\myparatight{Property~\ref{property-order-consistency} (Order consistency with Shapley value)} Given a pair of features $(i,j)$, $\alpha^{\hat{y}}_i(\mathbf{x},h,k)\geq \alpha^{\hat{y}}_j(\mathbf{x},h,k)$ if and only if $\phi_i(p_{\hat{y}})\geq \phi_j(p_{\hat{y}})$, where $\phi_i(p_{\hat{y}})$ and $\phi_j(p_{\hat{y}})$ respectively represent Shapley values of $i$ and $j$.
\begin{proof}
By the definition of Shapley value for $p_{\hat{y}}$, for any feature $l$,
\begin{align}\phi_l(p_{\hat{y}}) = &\sum_{S \subseteq \mathbf{x}\setminus\{l\}} \frac{|S|! (d - |S| - 1)!}{d!} (p_{\hat{y}}(S \cup \{l\},h,k) - p_{\hat{y}}(S,h,k))\\
=&\sum^{d-1}_{m=0}\frac{m! (d - m - 1)!}{d!}\sum_{S \subseteq \mathbf{x}\setminus\{l\},|S| = m}  (p_{\hat{y}}(S \cup \{l\},h,k) - p_{\hat{y}}(S,h,k))
\end{align}
We define the unregularized marginal contribution of feature $l \in \mathbf{x}$ with respect to subset size $m$ as:
\begin{align}
    \Delta_l(p_{\hat{y}}, m) =\sum_{S \subseteq \mathbf{x}\setminus\{l\},|S| = m}  (p_{\hat{y}}(S \cup \{l\},h,k) - p_{\hat{y}}(S,h,k)).
\end{align}
Shapley value is the weighted sum of $\Delta_l(p_{\hat{y}}, m)$ for all $0\leq m \leq d-1$, and the weights are all positive. Therefore, if our importance score is order consistent with $\Delta_l(p_{\hat{y}}, m)$ for every $0\leq m \leq d-1$, then our importance score is order consistent with the Shapley value.
We first use the definition in Section~\ref{def-subset} to handle special cases of $m$. When $m < k-1$, we have $\sum_{S \subseteq \mathbf{x}\setminus\{l\},|S| = m}  (p_{\hat{y}}(S \cup \{l\},h,k) - p_{\hat{y}}(S,h,k)) = 0$ for all $l$. When $m = k-1$, we have: 
\begin{align}
&\Delta_l(p_{\hat{y}}, k-1)\\
=&\sum_{S \subseteq \mathbf{x}\setminus\{l\},|S| = k-1}  (p_{\hat{y}}(S \cup \{l\},h,k) - p_{\hat{y}}(S,h,k)) \\
=&\sum_{S \subseteq \mathbf{x}\setminus\{l\},|S| = k-1}  (p_{\hat{y}}(S \cup \{l\},h,k) - \frac{1}{C})\\
=&\sum_{\mathbf{z} \subseteq \mathbf{x}, l \in \mathbf{z}}  (p_{\hat{y}}(\mathbf{z},h,k) - \frac{1}{C})\\
=&\sum_{\mathbf{z} \subseteq \mathbf{x}, l \in \mathbf{z}}  (\mathbb{I}(h(\mathbf{z}) = \hat{y}) - \frac{1}{C})\\
=&\sum_{\mathbf{z} \subseteq \mathbf{x}, l \in \mathbf{z}}  \mathbb{I}(h(\mathbf{z}) = \hat{y}) - \sum_{\mathbf{z} \subseteq \mathbf{x}, l \in \mathbf{z}} \frac{1}{C}\\
=& k\cdot\binom{n}{k}\cdot \alpha^{\hat{y}}_l(\mathbf{x},h,k) - \sum_{\mathbf{z} \subseteq \mathbf{x}, l \in \mathbf{z}} \frac{1}{C}.
\end{align}
Hence $\alpha^{\hat{y}}_i(\mathbf{x},h,k) \geq \alpha^{\hat{y}}_j(\mathbf{x},h,k)$ if and only if $\Delta_i(p_{\hat{y}}, k-1) \geq\Delta_j(p_{\hat{y}}, k-1)$.
Lastly, we consider the case when $k\leq m\leq d-1$. In this case,
\begin{align}\small 
&\Delta_l(p_{\hat{y}}, m)\\
=&\sum_{S \subseteq \mathbf{x}\setminus\{l\},|S| = m}  (p_{\hat{y}}(S \cup \{l\},h,k) - p_{\hat{y}}(S,h,k)) \\
=&\sum_{S \subseteq \mathbf{x}\setminus\{l\},|S| = m}  (\frac{1}{\binom{m+1}{k}}\sum_{\mathbf{z}\subseteq S \cup \{l\}}\mathbb{I}(h(\mathbf{z})=\hat{y}) - \frac{1}{\binom{m}{k}}\sum_{\mathbf{z}\subseteq S}\mathbb{I}(h(\mathbf{z})=\hat{y})) \\
=&\sum_{S \subseteq \mathbf{x}\setminus\{l\},|S| = m}  (\frac{1}{\binom{m+1}{k}}\sum_{\mathbf{z}\subseteq S \cup \{l\},l\in \mathbf{z}}\mathbb{I}(h(\mathbf{z})=\hat{y})+\frac{1}{\binom{m+1}{k}}\sum_{\mathbf{z}\subseteq S }\mathbb{I}(h(\mathbf{z})=\hat{y}) \\
&- \frac{1}{\binom{m}{k}}\sum_{\mathbf{z}\subseteq S}\mathbb{I}(h(\mathbf{z})=\hat{y})) \\
=&[\frac{1}{\binom{m+1}{k}}\sum_{S \subseteq \mathbf{x}\setminus\{l\},|S| = m}  \sum_{\mathbf{z}\subseteq S \cup \{l\},l\in \mathbf{z}}\mathbb{I}(h(\mathbf{z})=\hat{y})] \\
&- [(\frac{1}{\binom{m}{k}}-\frac{1}{\binom{m+1}{k}})\sum_{S \subseteq \mathbf{x}\setminus\{l\},|S| = m}\sum_{\mathbf{z}\in S}\mathbb{I}(h(\mathbf{z})=\hat{y}))\label{eq:second_last}] \\
=&[\frac{1}{\binom{m+1}{k}}\cdot\binom{d-k}{m-k+1}\sum_{\mathbf{z} \subseteq \mathbf{x}, l\in \mathbf{z}} \mathbb{I}(h(\mathbf{z})=\hat{y})]\\
&-[(\frac{1}{\binom{m}{k}}-\frac{1}{\binom{m+1}{k}})\cdot \binom{d-1-k}{m-k}\sum_{\mathbf{z}\subseteq \mathbf{x},l\notin \mathbf{z}}\mathbb{I}(h(\mathbf{z})=\hat{y}))] \label{eq:last}
\end{align}
We get Equation~\ref{eq:last} from Equation~\ref{eq:second_last} using combinatorial theory. For example, to find out how many times a specific $k$-sized subset that does not include $l$ appears across all possible selections, we recognize that for each $k$-sized subset to be part of an $m$-sized subset, we must choose the remaining $m-k$ elements from the $d-1-k$ elements that are not part of our $k$-sized subset.

Suppose $\alpha^{\hat{y}}_i(\mathbf{x},h,k)\geq \alpha^{\hat{y}}_j(\mathbf{x},h,k)$, then $\sum_{\mathbf{z} \subseteq \mathbf{x}, i\in \mathbf{z}} \mathbb{I}(h(\mathbf{z})=\hat{y}) \geq \sum_{\mathbf{z} \subseteq \mathbf{x}, j\in \mathbf{z}} \mathbb{I}(h(\mathbf{z})=\hat{y})$ and $ \sum_{\mathbf{z}\subseteq \mathbf{x},i\notin \mathbf{z}}\mathbb{I}(h(\mathbf{z})=\hat{y})) \leq \sum_{\mathbf{z}\subseteq \mathbf{x},j\notin \mathbf{z}}\mathbb{I}(h(\mathbf{z})=\hat{y}))$, which means $\Delta_i(p_{\hat{y}}, m) \geq \Delta_j(p_{\hat{y}}, m)$. And vise versa. Therefore, our importance score is order consistent with $\Delta_l(p_{\hat{y}}, m)$ for every $0\leq m \leq d-1$, which implies that our importance score is order consistent with the Shapley value.
\end{proof}

\section{Proof for Certified Detection of Adversarial Features}\label{proof_detection}
In this section, we first present an informal proof sketch to illustrate our proof strategy, followed by the complete proof.
\begin{proofsketch}
We use $\hat{y}$ and $\hat{y}'$ to denote the clean label and the target label, use $p_{\hat{y}}$ and $p_{\hat{y}'}$ to denote their label probabilities before the attack, use $T$ to denote the maximum number of adversarial features (features that are modified by the attacker), use $E(\mathbf{x}')$ to denote the set of reported important features after the attack, and use $\mathbf{x}' \ominus \mathbf{x}$ to denote these adversarial features.

Given an integer $r$, we want to check if it is possible for the intersection size (between $E(\mathbf{x}')$ and $\mathbf{x}' \ominus \mathbf{x}$) to fall below it. Suppose the intersection size is smaller than $r$. Then we know that at least $T -r + 1$ adversarial features are not reported (we denote this set by $U$), and at least $e - r + 1$ clean features are reported (we denote this set by $V$). So, we know that the maximum importance score in $U$ is smaller than or equal to the minimum importance score in $V$. By the law of contraposition, we know that if the maximum importance score in $U$ is larger than the minimum importance score in $V$, then the intersection size must be larger than $r$.

We first derive a lower bound of the maximum importance score in $U$. By the definition of the ensemble model, these adversarial features must change the predictions of enough feature groups to change the label from $\hat{y}$ to $\hat{y}'$. At the same time, their importance values must increase by $(p_{\hat{y}} - p_{\hat{y}'})/2$ in total (by our calculation). Using this information, we can derive the minimum sum of the importance scores in $U$ (a subset of adversarial features). Based on the fact that the maximum value in a set is larger than the average value, we can bound the maximum importance score in $U$.

We then derive an upper bound of the maximum importance score in $V$. We know that an adversarial feature can only affect the importance value of a clean feature if they coexist in a feature group. The fraction of such feature groups is small and can be calculated. Additionally, we have access to the importance values of clean features before the attack. By summing the original values with the change, we can obtain this upper bound.

Finally, we can check whether the lower bound of the maximum importance score in $U$ is larger than the upper bound of the minimum importance score in $V$. If not, this means the intersection size can be smaller than $r$ in the worst case, and we can switch to a smaller $r$ and repeat this process.
\end{proofsketch}
Next, we present the complete proof.
\begin{proof}
Our goal is to derive the \emph{certified detection size} $\mathcal{D}(\mathbf{x},T)$, which is the intersection size lower bound between the set of modified features $\mathbf{x}' \ominus \mathbf{x}$ and the set of reported important features $E(\mathbf{x}')$. It is formally defined as:
\begin{align}\mathcal{D}(\mathbf{x},T) = \argmax_r, s.t. |(\mathbf{x}'\ominus \mathbf{x}) \cap E(\mathbf{x}')|\geq r, \forall \mathbf{x}' \in \mathcal{B}(\mathbf{x},T), H(\mathbf{x}')\neq H(\mathbf{x})
\end{align}

Without loss of generality, we assume $H(\mathbf{x}')= \hat{y}'\neq \hat{y}$. We derive the certified detection size utilizing the \emph{law of contraposition}.  Suppose the number of features in $\mathbf{x}' \ominus \mathbf{x}$ that are also in $E(\mathbf{x}')$ is smaller than $r$, then we know that at least $T - r + 1$ features (denoted by $U$) in $\mathbf{x}' \ominus \mathbf{x}$ are not reported in the explanation for $\mathbf{x}'$. Similarly, we know at least $e - r + 1$ features (denoted by
$V$) in $\{1, 2, \cdots, d\} \setminus (\mathbf{x'} \ominus \mathbf{x})$ are in $E(\mathbf{x}')$. In other words, we know there exist $U \subseteq \mathbf{x}' \ominus \mathbf{x}$
and $V \subseteq \{1, 2, \cdots, d\} \setminus (\mathbf{x}' \ominus \mathbf{x})$ such that $\max_{u\in U} \alpha^{\hat{y}'}_u(\mathbf{x}',h,k) \leq \min_{v\in V} \alpha^{\hat{y}'}_v (\mathbf{x}',h,k)$
. Based on the law of
contraposition, we know that if we could show $\max_{u\in U} \alpha^{\hat{y}'}_u(\mathbf{x}',h,k) > \min_{v\in V} \alpha^{\hat{y}'}_v (\mathbf{x}',h,k)$
for arbitrary $U$ and $V$, i.e., $\min_U \max_{u\in U}\alpha^{\hat{y}'}_u(\mathbf{x}',h,k) > \max_V \min_{v\in V} \alpha^{\hat{y}'}_v(\mathbf{x}',h,k)$, then we know the certified
intersection size is no smaller than $r$. 

We note that $U$ and $V$ depends on the attacker's choice of $\mathbf{x}'$. To simplify the notation, we denote the $U$ that achieves the minimum by $U^*$ and the $V$ that achieves the maximum by $V^*$. Then, by considering the worst case $\mathbf{x}'$, the problem
becomes determining whether $\min_{\mathbf{x'}\in \mathcal{B}(\mathbf{x},T),H(\mathbf{x}') = \hat{y}'}(\max_{u\in U^*} {\alpha}^{\hat{y}'}_u(\mathbf{x}',h,k)-\min_{v\in V^*} {\alpha}^{\hat{y}'}_v(\mathbf{x}',h,k)) >0$. 
To simplify, we tackle a more straightforward version of this problem by determining if $\min_{\mathbf{x'}\in \mathcal{B}(\mathbf{x},T),H(\mathbf{x}') = \hat{y}'}\max_{u\in U^*} {\alpha}^{\hat{y}'}_u(\mathbf{x}',h,k)>\max_{\mathbf{x'}\in \mathcal{B}(\mathbf{x},T),H(\mathbf{x}') = \hat{y}'}\min_{v\in V^*} {\alpha}^{\hat{y}'}_v(\mathbf{x}',h,k)$.

According to the definition of the ensemble model in Equation~\ref{eqn-ensemble}, in order to change the label from $\hat{y}$ to $\hat{y}'$, the attacker at least needs to change the predictions of $\frac{1}{2}\binom{d}{k}\cdot({p_{\hat{y}}(\mathbf{x},h,k)-p_{\hat{y}'}(\mathbf{x},h,k)})$ feature groups which are not predicted as $\hat{y}$ to $\hat{y}$, where $\binom{d}{k}$ is the number of unique feature groups, i.e., $|\{\mathbf{z}\subseteq \mathbf{x}: |\mathbf{z}| = k\}|$. Since each of these changed feature groups contains at least one feature in $\mathbf{x}\ominus\mathbf{x}'$, for any $\mathbf{x}'$ satisfying $H(\mathbf{x}') = \hat{y}'$, we have $\sum_{i \in \mathbf{x}\ominus \mathbf{x}'}[\alpha^{\hat{y}'}_{i}(\mathbf{x}',h,k)- \alpha^{\hat{y}'}_{i}(\mathbf{x},h,k)] \geq \frac{1}{k}\cdot\frac{p_{\hat{y}}-p_{\hat{y}'}}{2}$. It follows that $\sum_{u \in U^*}[\alpha^{\hat{y}'}_{u}(\mathbf{x}',h,k)- \alpha^{\hat{y}'}_{u}(\mathbf{x},h,k)] \geq \frac{1}{k}\cdot\frac{p_{\hat{y}}-p_{\hat{y}'}}{2}-(r-1)\cdot \frac{1}{k}\frac{\binom{d-1}{k-1}}{\binom{d}{k}}=\frac{1}{k}\cdot\frac{p_{\hat{y}}-p_{\hat{y}'}}{2}-\frac{r-1}{d}$. This is because for each modified feature not in $U^*$, the change of its importance value is bounded by $\frac{1}{k}\cdot\frac{\binom{d-1}{k-1}}{\binom{d}{k}}$. So we have:
\begin{align}
&\min_{\mathbf{x}'\in \mathcal{B}(\mathbf{x},T),H(\mathbf{x}') = \hat{y}'}\max_{u\in U^*} {\alpha}^{\hat{y}'}_u(\mathbf{x}',h,k)
\\
\geq&\frac{1}{T-r+1}\min_{\mathbf{x'}\in \mathcal{B}(\mathbf{x},T),H(\mathbf{x}') = \hat{y}'}\sum_{u\in U^*} {\alpha}^{\hat{y}'}_u(\mathbf{x}',h,k)\\
\geq&\frac{1}{T-r+1}[\min_{\mathbf{x}\ominus\mathbf{x'}}\sum_{u\in U^*} {\alpha}^{\hat{y}'}_u(\mathbf{x},h,k)+ (\frac{1}{k}\cdot\frac{p_{\hat{y}}(\mathbf{x},h,k)-p_{\hat{y}'}(\mathbf{x},h,k)}{2}-\frac{r-1}{d})]
\end{align}

We use $\{w_1, \cdots , w_{d}\}$ to denote the set of all features in descending order of the important value $\alpha^{\hat{y}'}(\mathbf{x},h,k)$. We notice that to minimize $\sum_{u\in U^*} {\alpha}^{\hat{y}'}_u(\mathbf{x},h,k)$, $\mathbf{x}' \ominus \mathbf{x}$ includes features with lowest $\alpha^{\hat{y}'}(\mathbf{x},h,k)$'s. Then we can denote the worst case $\mathbf{x}' \ominus \mathbf{x}$ as $\{w_{d-T+1}, \cdots , w_{d}\}$. It follows that $U^* =  \{w_{d-T+r}, \cdots , w_{d}\}$ from the definition of $U$, which means:
\begin{align}
&\min_{\mathbf{x}'\in \mathcal{B}(\mathbf{x},T),H(\mathbf{x}') = \hat{y}'}\max_{u\in U^*} {\alpha}^{\hat{y}'}_u(\mathbf{x}',h,k)\\
\geq&\frac{1}{T-r+1}[\frac{1}{2k}\cdot({p_{\hat{y}}(\mathbf{x},h,k)-p_{\hat{y}'}(\mathbf{x},h,k)})-\frac{r-1}{d}+\sum_{i = {d-T+r}}^{d} {\alpha}^{\hat{y}'}_{w_i}(\mathbf{x},h,k)]\label{eqn-u-individual}
\end{align}

If we consider each $v$ in $V^*$ individually, we can find an upper bound for $\max_{\mathbf{x}'\in \mathcal{B}(\mathbf{x},T),H(\mathbf{x}') = \hat{y}'}\min_{v\in V^*} {\alpha}^{\hat{y}'}_v(\mathbf{x}',h,k)$. By the definition of $V$, each feature $v$ in $V^*$ is not modified by the attacker. Hence at least $\binom{d-1-T}{k-1}$ of the $\binom{d-1}{k-1}$ unique feature groups with size $k$ that contains $v$ are unaffected by the attack. Therefore we have ${\alpha}^{\hat{y}'}_v(\mathbf{x}',h,k)-{\alpha}^{\hat{y}'}_v(\mathbf{x},h,k)\leq \frac{1}{k}\frac{\binom{d-1}{k-1}-\binom{d-1-T}{k-1}}{\binom{d}{k}}$. So we get:
\begin{align}
&\max_{\mathbf{x}'\in \mathcal{B}(\mathbf{x},T),H(\mathbf{x}') = \hat{y}'}\min_{v\in V^*} {\alpha}^{\hat{y}'}_v(\mathbf{x}',h,k)\\
\leq &\max_{\mathbf{x}\ominus \mathbf{x}'}\min_{v\in V^*} {\alpha}^{\hat{y}'}_v(\mathbf{x},h,k)+\frac{1}{k}\frac{\binom{d-1}{k-1}-\binom{d-1-T}{k-1}}{\binom{d}{k}}
\end{align}
We notice that to achieve the maximum, $\{1,2, \cdots,d\}\setminus (\mathbf{x}' \ominus \mathbf{x})$ includes features with highest $\alpha^{\hat{y}'}(\mathbf{x},h,k)$'s. So we can denote the worst case $\{1,2, \cdots,d\}\setminus (\mathbf{x}' \ominus \mathbf{x})$ as $\{w_1, \cdots , w_{d-T}\}$. Then we have $V^* = \{w_1, w_2, \cdots , w_{e-r+1}\}$ in the worst case. So we have:
\begin{align}
&\max_{\mathbf{x}'\in \mathcal{B}(\mathbf{x},T),H(\mathbf{x}') = \hat{y}'}\min_{v\in V^*} {\alpha}^{\hat{y}'}_v(\mathbf{x}',h,k)
\leq {\alpha}^{\hat{y}'}_{w_{e-r+1}}(\mathbf{x},h,k)+\frac{1}{k}\frac{\binom{d-1}{k-1}-\binom{d-1-T}{k-1}}{\binom{d}{k}}\label{eqn-v-individual}
\end{align}
If we assume $H(\mathbf{x}')=\hat{y}'$, by combining Equation~\ref{eqn-u-individual} and Equation~\ref{eqn-v-individual}, we get:
\begin{align}
&\mathcal{D}(\mathbf{x},T) \geq r,\text{ } \text{if: }\\
  &
{\alpha}^{\hat{y}'}_{w_{e-r+1}}(\mathbf{x},h,k)+\frac{1}{d}-\frac{1}{k}\frac{\binom{d-1-T}{k-1}}{\binom{d}{k}}\\
\leq&\frac{1}{T-r+1}[\frac{1}{2k}\cdot({{p}_{\hat{y}}(\mathbf{x},h,k)-{p}_{\hat{y}'}(\mathbf{x},h,k)})-\frac{r-1}{d}+\sum_{i = {d-T+r}}^{d} {\alpha}^{\hat{y}'}_{w_i}(\mathbf{x},h,k)],
\end{align}

We can also consider all $v\in V^*$ jointly. We use $\delta_i$ to denote ${\alpha}^{\hat{y}'}_i(\mathbf{x}',h,k)-{\alpha}^{\hat{y}'}_i(\mathbf{x},h,k)$ for feature $i$. We know that each feature group of size $k$ that contains that least one modified feature at most contains $k-1$ unmodified features. This leads to the following inequality:
\begin{align}\sum_{i \in \mathbf{x}\ominus \mathbf{x}'}\delta_i \geq \frac{1}{k-1}\sum_{i \notin \mathbf{x}\ominus \mathbf{x}'}\delta_i\label{eqn-joint-inequality}\end{align} 
We first rewrite the maximum importance score of features in $U^*$ as:
\begin{align}
&\min_{\mathbf{x}'\in \mathcal{B}(\mathbf{x},T),H(\mathbf{x}') = \hat{y}'}\max_{u\in U^*} {\alpha}^{\hat{y}'}_u(\mathbf{x}',h,k)\\
\geq&\frac{1}{T-r+1}\min_{\mathbf{x}'\in \mathcal{B}(\mathbf{x},T),H(\mathbf{x}') = \hat{y}'}\sum_{u\in U^*} {\alpha}^{\hat{y}'}_u(\mathbf{x}',h,k)\\
\geq&\frac{1}{T-r+1}\min_{\mathbf{x}'\in \mathcal{B}(\mathbf{x},T),H(\mathbf{x}') = \hat{y}'}\sum_{u\in U^*} {\alpha}^{\hat{y}'}_u(\mathbf{x},h,k)+\sum_{u\in U^*} \delta_u\\
\geq&\frac{1}{T-r+1}[\min_{\mathbf{x}'\in \mathcal{B}(\mathbf{x},T),H(\mathbf{x}') = \hat{y}'}\sum_{u\in U^*} {\alpha}^{\hat{y}'}_u(\mathbf{x},h,k)-\frac{r-1}{d}+\sum_{i\in \mathbf{x}\ominus \mathbf{x}'} \delta_i]\\
\geq&\frac{1}{T-r+1}(\min_{\mathbf{x}'\in \mathcal{B}(\mathbf{x},T),H(\mathbf{x}') = \hat{y}'}\sum_{u\in U^*} {\alpha}^{\hat{y}'}_u(\mathbf{x},h,k)-\frac{r-1}{d})\\
&+\frac{1}{T-r+1}\max(\sum_{i\in \mathbf{x}\ominus \mathbf{x}'} \delta_{i},\frac{1}{2k}\cdot({p_{\hat{y}}(\mathbf{x},h,k)-p_{\hat{y}'}(\mathbf{x},h,k)}))\\
\geq&\frac{1}{T-r+1} (\sum_{i = d-T+r}^{d}{\alpha}^{\hat{y}'}_{w_i}(\mathbf{x},h,k)-\frac{r-1}{d})\\
&+\frac{1}{T-r+1}\max(\sum_{i\in \mathbf{x}\ominus \mathbf{x}'} \delta_{i},\frac{1}{2k}\cdot({p_{\hat{y}}(\mathbf{x},h,k)-p_{\hat{y}'}(\mathbf{x},h,k)}))\label{eqn-joint-u}
\end{align}
We then write the minimum importance score of features in $V^*$ as:
\begin{align}
&\max_{\mathbf{x}'\in \mathcal{B}(\mathbf{x},T),H(\mathbf{x}') = \hat{y}'}\min_{v\in V^*} {\alpha}^{\hat{y}'}_v(\mathbf{x}',h,k)\\
\leq&\max_{\mathbf{x}'\in \mathcal{B}(\mathbf{x},T),H(\mathbf{x}') = \hat{y}'}\frac{1}{e-r+1}\sum_{v\in V^*} {\alpha}^{\hat{y}'}_v(\mathbf{x}',h,k)\\
\leq&\max_{\mathbf{x}'\in \mathcal{B}(\mathbf{x},T),H(\mathbf{x}') = \hat{y}'}\frac{1}{e-r+1}((k-1)\sum_{i\in \mathbf{x}\ominus \mathbf{x}'} \delta_{i} +\sum_{v\in V^*} {\alpha}^{\hat{y}'}_v(\mathbf{x},h,k))\label{eqn-use-joint-inequality}\\
\leq&[\frac{1}{e-r+1}\max_{\mathbf{x'}\in \mathcal{B}(\mathbf{x},T),H(\mathbf{x}') = \hat{y}'}\sum_{v\in V^*} {\alpha}^{\hat{y}'}_v(\mathbf{x},h,k)]+\frac{k-1}{e-r+1}\sum_{i\in \mathbf{x}\ominus \mathbf{x}'} \delta_{i} \\
\leq&\frac{1}{e-r+1}\sum_{i =1}^{e-r+1} {\alpha}^{\hat{y}'}_{w_i}(\mathbf{x},h,k)+\frac{k-1}{e-r+1}\sum_{i\in \mathbf{x}\ominus \mathbf{x}'} \delta_{i}.\label{eqn-joint-V}
\end{align}
Equation~\ref{eqn-use-joint-inequality} is derived by applying Equation~\ref{eqn-joint-inequality}. After subtracting Equation~\ref{eqn-joint-u} by Equation~\ref{eqn-joint-V}, we have:
\begin{align}
& \min_{\mathbf{x}'\in \mathcal{B}(\mathbf{x},T),H(\mathbf{x}') = \hat{y}'}\max_{u\in U^*} {\alpha}^{\hat{y}'}_u(\mathbf{x}',h,k)-\max_{\mathbf{x}'\in \mathcal{B}(\mathbf{x},T),H(\mathbf{x}') = \hat{y}'}\min_{v\in V^*} {\alpha}^{\hat{y}'}_v(\mathbf{x}',h,k)\label{eqn-min-max-inequality}\\
\geq& \frac{1}{T-r+1} (\sum_{i = d-T+r}^{d}{\alpha}^{\hat{y}'}_{w_i}(\mathbf{x},h,k)-\frac{r-1}{d})\\
&+\frac{1}{T-r+1}\max(\sum_{i\in \mathbf{x}\ominus \mathbf{x}'} \delta_{i},\frac{1}{2k}\cdot({p_{\hat{y}}(\mathbf{x},h,k)-p_{\hat{y}'}(\mathbf{x},h,k)}))\\
&-[\frac{1}{e-r+1}\sum_{i =1}^{e-r+1} {\alpha}^{\hat{y}'}_{w_i}(\mathbf{x},h,k)+\frac{k-1}{e-r+1}\sum_{i\in \mathbf{x}\ominus \mathbf{x}'} \delta_{i}] \\
\geq&[\frac{1}{T-r+1} \sum_{i = d-T+r}^{d}{\alpha}^{\hat{y}'}_{w_i}(\mathbf{x},h,k)-\frac{1}{e-r+1}\sum_{i =1}^{e-r+1} {\alpha}^{\hat{y}'}_{w_i}(\mathbf{x},h,k)-\frac{r-1}{d\cdot(T-r+1)}]\\
&+\frac{1}{2k}(\frac{1}{T-r+1}-\frac{k-1}{e-r+1})\cdot({p_{\hat{y}}(\mathbf{x},h,k)-p_{\hat{y}'}(\mathbf{x},h,k)})\label{eqn-assume}
\end{align}
We have Equation~\ref{eqn-assume} by assuming $\frac{1}{T-r+1}>\frac{k-1}{e-r+1}$. We can make this assumption because otherwise Equation~\ref{eqn-min-max-inequality} must be smaller than zero and the certification for any $r$ must not hold. Therefore, by jointly consider all $v\in V^*$, and assuming $H(\mathbf{x}')=\hat{y}'$, we get:
\begin{align}
&\mathcal{D}(\mathbf{x},T) \geq r,\text{ } \text{if:}\text{ }\\
  &
\frac{1}{e-r+1}\sum_{i =1}^{e-r+1} {\alpha}^{\hat{y}'}_{w_i}(\mathbf{x},h,k)-\frac{1}{T-r+1} \sum_{i = d-T+r}^{d}{\alpha}^{\hat{y}'}_{w_i}(\mathbf{x},h,k)+\frac{r-1}{d\cdot(T-r+1)}\\
\leq&\frac{1}{2k}(\frac{1}{T-r+1}-\frac{k-1}{e-r+1})\cdot({p_{\hat{y}}(\mathbf{x},h,k)-p_{\hat{y}'}(\mathbf{x},h,k)}).
\end{align}

In practice, we use Monte Carlo sampling to compute lower (or upper) bounds for the importance scores and label probabilities. Please refer to Section~\ref{theory-confidence-bounds} for the details. Putting together with previous results, we have:
\begin{align}
&\mathcal{D}(\mathbf{x},T) = \argmax_r r,\text{ } s.t.\text{ }\forall \hat{y}'\neq \hat{y},\text{ }\\
  &
\overline{\alpha}^{\hat{y}'}_{w_{e-r+1}}(\mathbf{x},h,k)+\frac{1}{d}-\frac{1}{k}\frac{\binom{d-1-T}{k-1}}{\binom{d}{k}}\\
\leq&\frac{1}{T-r+1}[\frac{1}{2k}\cdot({\underline{p}_{\hat{y}}(\mathbf{x},h,k)-\overline{p}_{\hat{y}'}(\mathbf{x},h,k)})-\frac{r-1}{d}+\sum_{i = {d-T+r}}^{d} \underline{\alpha}^{\hat{y}'}_{q_i}(\mathbf{x},h,k)]\\
&\;\lor\\
&\frac{1}{e-r+1}\sum_{i =1}^{e-r+1} \overline{\alpha}^{\hat{y}'}_{w_i}(\mathbf{x},h,k)-\frac{1}{T-r+1} \sum_{i = d-T+r}^{d}\underline{\alpha}^{\hat{y}'}_{q_i}(\mathbf{x},h,k)+\frac{r-1}{d\cdot(T-r+1)}\\
\leq&\frac{1}{2k}(\frac{1}{T-r+1}-\frac{k-1}{e-r+1})\cdot({\underline{p}_{\hat{y}}(\mathbf{x},h,k)-\overline{p}_{\hat{y}'}(\mathbf{x},h,k)}) ,
\end{align}
where $\{w_1, \cdots , w_{d}\}$ denotes the set of all features in descending order of the important value upper bound $\overline{\alpha}^{\hat{y}'}(\mathbf{x},h,k)$, i.e., $\overline{\alpha}_{w_1}^{\hat{y}'}(\mathbf{x},h,k)\geq \overline{\alpha}_{w_2}^{\hat{y}'}(\mathbf{x},h,k)\geq \cdots \geq \overline{\alpha}_{w_d}^{\hat{y}'}(\mathbf{x},h,k)$, and $\{q_1, \cdots , q_{d}\}$ denotes the set of all features in descending order of the important value lower bound $\underline{\alpha}^{\hat{y}'}(\mathbf{x},h,k)$, i.e, $\underline{\alpha}_{q_1}^{\hat{y}'}(\mathbf{x},h,k)\geq \underline{\alpha}_{q_2}^{\hat{y}'}(\mathbf{x},h,k)\geq \cdots \geq \underline{\alpha}_{q_d}^{\hat{y}'}(\mathbf{x},h,k)$. 
\end{proof}

\section{Compute Bounds for Importance Scores and Label Probabilities}\label{theory-confidence-bounds}
We use Monte Carlo sampling to compute a lower (or upper) bound for the importance scores. The important score of feature $i$ for label $c$ can be rewritten as:
\begin{align}
    &\alpha_i^{c}(\mathbf{x},h,k)  \\
    = &\frac{1}{k}\mathbb{E}_{\mathbf{z} \sim \mathcal{U}(\mathbf{x},k)} [\mathbb{I}(i \in \mathbf{z}) \cdot \mathbb{I}(h(\mathbf{z})=c )]\\
        = &\frac{1}{k} \Pr(i \in \mathbf{z})\cdot \Pr(h(\mathbf{z})={c}|i \in \mathbf{z})\\
                = &\frac{1}{d} \Pr(h(\mathbf{z})=c|i \in \mathbf{z}).
\end{align}
In practice, it is estimated using Monte Carlo sampling as $\frac{1}{d } \frac{ \sum_{\mathbf{z}_j\in G}\mathbb{I}(i \in \mathbf{z}_j) \cdot \mathbb{I}(h(\mathbf{z}_j)=c)}{\sum_{\mathbf{z}_j \in G} \mathbb{I}(i \in \mathbf{z}_j)}$, where $G = \{\mathbf{z}_1, \ldots, \mathbf{z}_N\}$ is the collection of sampled feature groups. 
The objective is to establish a lower (or upper) probability bound for $\Pr(h(\mathbf{z})={c}|i \in \mathbf{z})$. The lower bound is denoted as $\underline{\Pr}(h(\mathbf{z})={c}|i \in \mathbf{z})$ and the upper bound is denoted as $\overline{\Pr}(h(\mathbf{z})={c}|i \in \mathbf{z})$. For each feature $i$, we consider a bernoulli process where $N_i = {\sum_{\mathbf{z}_j \in G} \mathbb{I}(i \in \mathbf{z}_j)}$ represents the number of Bernoulli trials (`coin tosses'), while $\hat{n}^c_i = \sum_{\mathbf{z}_j\in G, i \in \mathbf{z}_j}\mathbb{I}(h(\mathbf{z}_j)=c)$ corresponds to the `heads' count, or the number of successful outcomes. Therefore, we can compute the probability bounds for each feature $i \in \mathbf{x}$ using Clopper-Pearson based method~\cite{clopper1934clopper_pearson}:
\begin{align}
\label{prob_bound_est_1}
 &   \underline{\Pr}(h(\mathbf{z})=c|i \in \mathbf{z}) = \text{Beta}(\frac{\beta}{d}; \hat{n}^c_i, N_i - \hat{n}^c_i + 1),\text{ and}\\
 \label{prob_bound_est_2}
      &  \overline{\Pr}(h(\mathbf{z})={c}|i \in \mathbf{z}) = \text{Beta}(1-\frac{\beta}{d}; \hat{n}^c_i, N_i - \hat{n}^c_i + 1)),    
\end{align}
where $1 - \beta$ is the overall confidence level and $\text{Beta}(\rho;\varsigma, \vartheta)$ is the $\rho$-th quantile of the Beta distribution with shape parameters $\varsigma$ and $\vartheta$. We divide ${\beta}$ by ${d}$ because we need to divide the confidence level among the $d$ features. Then we have $\overline{\alpha}_i^{c}(\mathbf{x},h,k) =\frac{1}{d} \overline{\Pr}(h(\mathbf{z})={c}|i \in \mathbf{z})$, and $\underline{\alpha}_i^{c}(\mathbf{x},h,k) 
 = \frac{1}{d}\underline{\Pr}(h(\mathbf{z})={c}|i \in \mathbf{z})$.

Likewise, we can compute the label probability bounds as follows:
\begin{align}
\label{prob_bound_est_1}
&\forall c \in \{1,2,\cdots,C\},\\
 &   \underline{p}_c (\mathbf{x},h,k)= \text{Beta}(\frac{\beta}{C}; n_c, N - n_c + 1), \text{ and}\\
 \label{prob_bound_est_2}
      &  \overline{p}_c (\mathbf{x},h,k) = \text{Beta}(1-\frac{\beta}{C}; n_c, N - n_c + 1)) ,
\end{align}
where $n_c$ is the number of sampled feature groups that predicts for label $c$, $1 - \beta$ is the overall confidence level and $\text{Beta}(\rho;\varsigma, \vartheta)$ is the $\rho$-th quantile of the Beta distribution with shape parameters $\varsigma$ and $\vartheta$. We divide ${\beta}$ by $C$ because we simultaneously compute bounds for all labels. 
\section{Effectiveness of Appearance Frequency Normalization}\label{appendix_frequency_normalization}
We conduct an empirical comparison between two approaches to estimate the importance score in Eq.~(\ref{eqn_importance_score}): (1) directly applying Monte Carlo sampling, and (2) Monte Carlo sampling with normalization based on appearance frequency. Our experiment focus on certified defense adversarial attacks under default settings, with the sampling size $N$ set to 200. We use faithfulness as the metric, with the deletion ratio set to 20\%. The results demonstrate that normalizing the importance scores by appearance frequency improves the faithfulness.

\begin{table}[h!]
\centering
\caption{Compare faithfulness with and without normalization.}
\begin{tabular}{|c|c|c|c|}
\hline
\textbf{Dataset} & \textbf{SST-2} & \textbf{IMDB} & \textbf{AG-news} \\ \hline
Without Normalization & 0.82 & 0.95 & 0.92 \\ \hline
With Normalization (Eq.~(\ref{eqn_frequency_normalization}))& 0.87 & 0.99 & 0.96 \\ \hline
\end{tabular}
\label{tab:normalization_comparison}
\end{table}
\section{Discussion on Certified Detection of Adversarial Features}\label{appendix_theory_discussion}
From the theoretical result in Section~\ref{theory-detect}, we observe that the following factors can lead to a larger certified detection size $\mathcal{D}(\mathbf{x},T)$: 

\myparatight{1) High prediction confidence (represented by $\Delta$)} A high confidence ensemble model (before attack) increases the label probability gap between the predicted label $\hat{y}$ and any other label $\hat{y}'$. This means that the modified features must influence a greater number of subsampled feature groups to alter the prediction, making them more detectable.

\myparatight{2) Even distribution of importance values for the target Label (represented by ${\alpha}^{\hat{y}'}$)} In the worst-case scenario, an attacker could modify features with the smallest importance values, disguising the attack while increasing the probability of the target label. An even distribution of the importance values makes this strategy more difficult.

\myparatight{3) Smaller subsampling ratio (represented by $\frac{k}{d}$)} Consider the extreme case where $k=1$. In this scenario, each adversarial feature is part of only one feature group. Altering the prediction of that feature group would increase the importance value of the adversarial feature from 0 to $1/d$ (the maximum importance a feature can attain). This makes these adversarial features easy to identify.

\myparatight{4) Smaller number of modified features (represented by $T$)} With fewer adversarial features, each must impact more subsampled feature groups to change the prediction, increasing detectability.

\myparatight{Treating prediction confidence $\Delta$ and the number of features $d$ as additional parameters to vary} We treat the prediction confidence gap $\Delta$ and the number of features $d$ as tunable parameters to study their effect on the certified detection rate. We suppose it is a binary classification problem. For each simulated test input, feature importance values are generated synthetically: the number of subsets that include feature $j$ is sampled as $n_j \sim \text{Binomial}(N, 1 - \rho)$, where $N$ is the number of random subsets and $\rho$ is the masking ratio, and the number of those subsets that predict the alternative class is sampled as $n_j^{\hat{y}'} \sim \text{Binomial}(n_j, (1-\Delta)/2)$, where $(1-\Delta)/2$ is the base rate at which subsets predict the alternative class. All other parameters are fixed. As shown in Figure~\ref{fig-more-ablation}, a larger confidence gap leads to a higher certified detection rate. With $k$ fixed, increasing $d$ lowers the subsampling ratio and improves certified detection performance. When $k/d$ is fixed, a smaller $d$ performs better due to the statistical penalty introduced by the union bound; this penalty can be mitigated by increasing $N$.
\begin{figure*}[t]
\vspace{-3mm}
\centering
\subfloat[Affect of $\Delta$]{\includegraphics[width=0.95\textwidth]{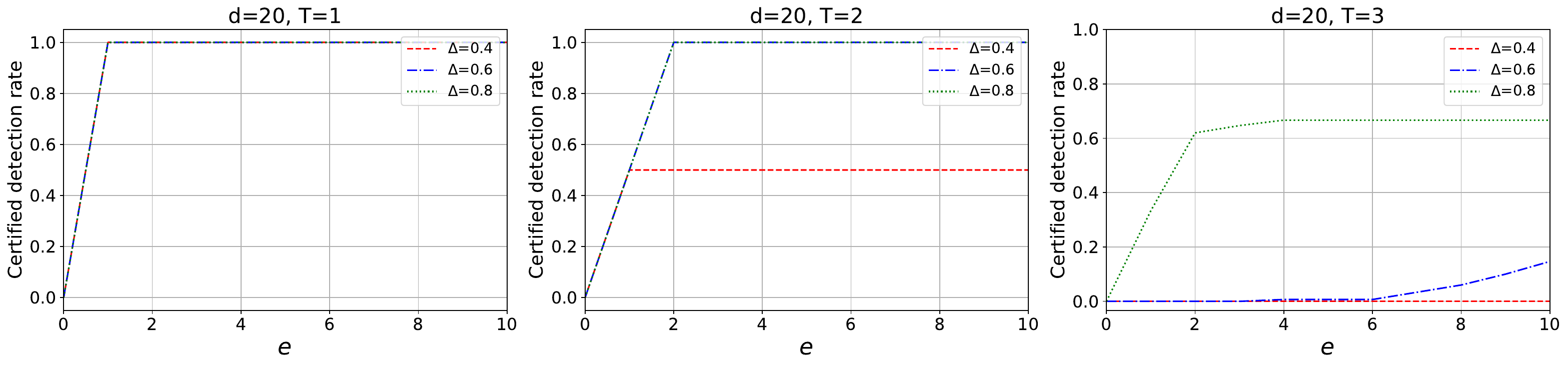}}\\

\subfloat[Affect of $d$ (fixed $k$)]{\includegraphics[width=0.95\textwidth]{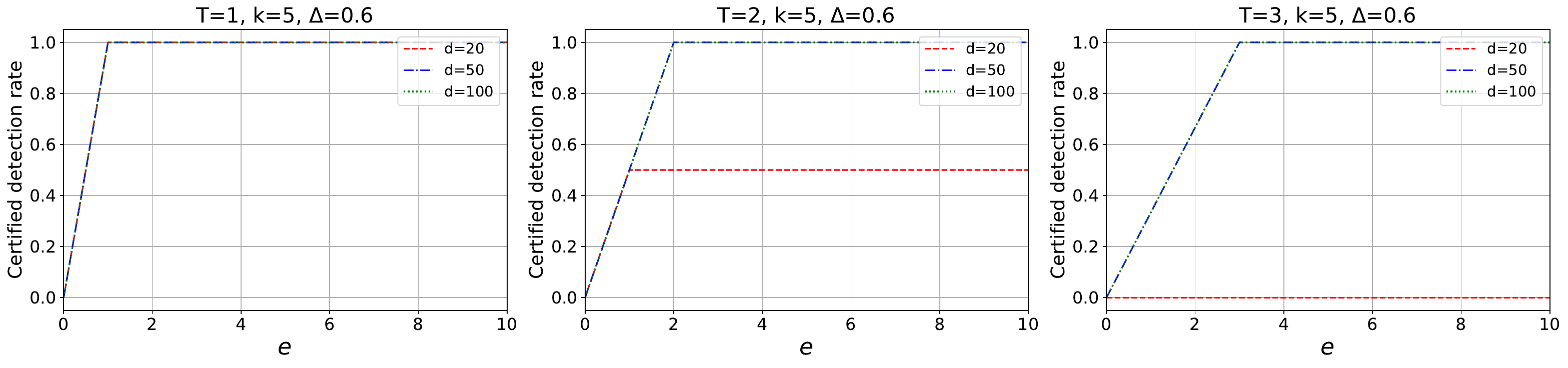}}
\vspace{2mm}

\subfloat[Affect of $d$ (fixed $k/d$)]{\includegraphics[width=0.95\textwidth]{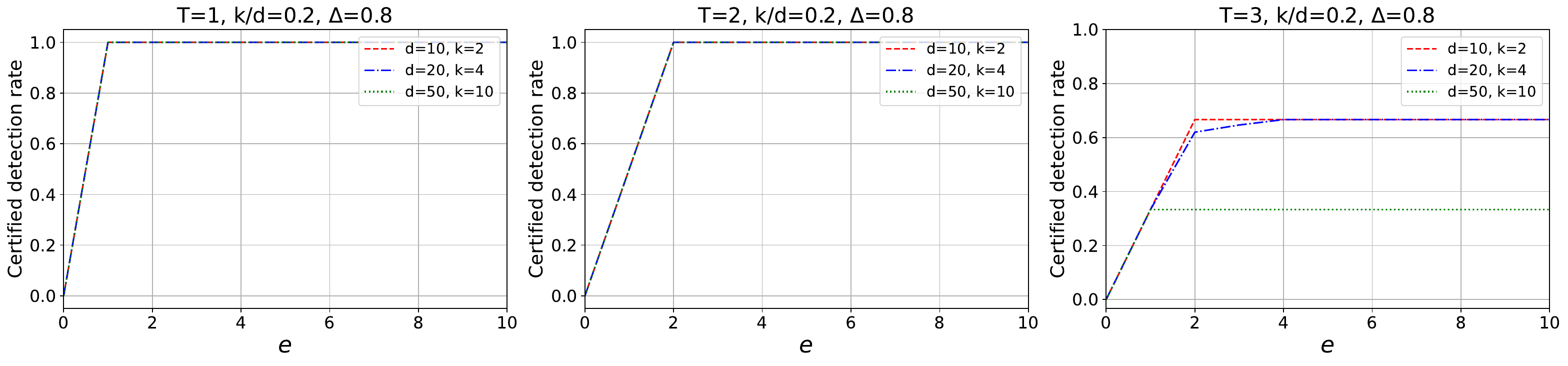}}
\vspace{-2mm}

\caption{Affect of prediction confidence $\Delta$ and the number of features $d$.}
\label{fig-more-ablation}
\vspace{-2mm}
\end{figure*}
\section{Computation Time}\label{appendix_computation_time}
Our approach incurs minimal computational cost for feature attribution, as it reuses the computational byproducts already generated during the prediction process of the random subspace method. In contrast, many other feature attribution methods such as Shapley value and LIME require significant computation time because they are not specifically tailored to the random subspace method. The experiments are performed under default settings for certified defense adversarial attacks. We report the computational time of a single testing sample, averaged over the test dataset. For certified detection, the time is shown for a single $e$ and $T$ combination.

\begin{table}[h!]
\centering
\caption{Computational time (in seconds) of our method for feature attribution and certified detection evaluation.}
\begin{tabular}{|c|c|c|c|}
\hline
\textbf{Dataset} & \textbf{SST-2} & \textbf{IMDB} & \textbf{AG-news} \\ \hline
\makecell{Feature attribution \\(Solving Eq.~(\ref{eqn_frequency_normalization}))} & 0.02 & 0.03 & 0.02 \\ \hline
\makecell{Certified detection \\(Solving Eq.~(\ref{eqn_certified_detection}))}& 0.19 & 0.36 & 0.23 \\ \hline
\end{tabular}
\label{tab:normalization_comparison}
\end{table}

\section{Compare {\name} with More Baselines}\label{appendix_shapley_variants}
We further compare our method with Beta-Shapley~\citep{kwon2021beta_shapley}, L-Shapley~\citep{chen2018l_shapley}, C-Shapley~\citep{chen2018l_shapley}, ContextCite~\citep{cohen2024contextcite}, and Kernel SHAP~\citep{lundberg2017unified}.
For Beta-Shapley, we follow the original paper and set $\alpha=4$ and $\beta=1$. For L-Shapley and C-Shapley, we use a neighborhood size $k=5$. For ContextCite, we set the regularization parameter $\alpha=0.01$ following the original implementation~\citep{cohen2024contextcite}. All baselines are applied to the base model. We evaluate faithfulness at removal ratios of $10\%$ and $20\%$ when there is no attack. As shown in Table~\ref{tab:more-baselines}, our method outperforms these baselines, as they do not guarantee faithfulness to the ensemble model.

\begin{table*}[t]
  \centering
  \captionof{table}{Compare the faithfulness of our
method with baselines for certified defense. We
delete different ratios of most important words
and compute the rate of label changes.\label{tab:more-baselines}}
  \vspace{-0.5em}
  \resizebox{0.60\textwidth}{!}{
    {\scriptsize
    \begin{tabular}{c|c|c|c|c|c|c}
      \toprule
      Dataset & \multicolumn{2}{c|}{SST-2} & \multicolumn{2}{c|}{IMDb} &
              \multicolumn{2}{c}{AG-news} \\ \midrule
      Ratio & 10\% & 20\% & 10\% & 20\% & 10\% & 20\% \\ \midrule
      Kernel SHAP  & 0.28 & 0.66 & 0.18 & 0.28 & 0.22 & 0.38 \\
      ContextCite  & 0.30 & \textbf{0.68} & 0.56 & 0.60 & 0.22 & 0.38 \\
      Beta-Shapley & 0.28 & 0.62 & 0.34 & 0.42 & 0.20 & 0.38 \\
      L-Shapley    & 0.30 & 0.54 & 0.38 & 0.44 & 0.20 & 0.36 \\
      C-Shapley    & 0.30 & 0.52 & 0.36 & 0.42 & 0.20 & 0.34 \\
      Ours         & \textbf{0.32} & \textbf{0.68} & \textbf{0.60} & \textbf{0.70} & \textbf{0.24} & \textbf{0.42} \\
      \bottomrule
    \end{tabular}}}
\end{table*}
\begin{table}[!t]
\centering
\caption{Certified detection rate as a function of the
maximum number of adversarial patches ($T$).}
\label{tab:image_certification}
\small
\begin{tabular}{lccc}
\toprule
\textbf{Dataset} & $T{=}1$ & $T{=}3$ & $T{=}5$ \\
\midrule
CIFAR-10     & 0.90 & 0.77 & 0.64 \\ 
ImageNette   & 0.96 & 0.89 & 0.84 \\ 
ImageNet-100 & 0.86 & 0.74 & 0.63 \\ 
\bottomrule
\end{tabular}
\end{table}
\section{Application of Our Method in Image Domain}\label{appendix_image_domain}
Here, we show that our method can be applied to the image domain to provide certified detection guarantee against image patch attacks~\citep{brown2017adversarial_patch, levine2020derandomized}. We evaluate on three image
datasets: CIFAR-10, ImageNette, and
ImageNet-100.  
We fine-tune a
pre-trained DINO~\citep{caron2021dino} backbone and treat each $14{\times}14$ image
patch (in a $224{\times}224$ image) as a feature.  
We draw $10{,}000$ random samples with a dropping ratio of $0.95$;
the clean accuracies of the resulting ensembles models are
$72\%$, $94\%$, and $70\%$, respectively.
Explanations report the top $5\%$ most important features, and we
certify the detection rate when the attacker may modify at most
$T$ patches.

The results in Table~\ref{tab:image_certification} demonstrate the robustness of our method across diverse
image datasets and attack budgets.  
For instance, on ImageNette, our method is guaranteed to detect \textbf{84\%} of the
manipulated patches when the adversary is allowed to alter
five patches.  

Regarding efficiency, certification takes only $0.35$, $0.25$, and
$0.24$~seconds per sample on CIFAR-10, ImageNette, and ImageNet-100,
respectively, even though these datasets involve far more features
(e.g., 1{,}024 on CIFAR-10) and samples
($10,000$) than the text benchmarks (around 21 features and 1{,}000 samples
on SST-2) used elsewhere in this paper.

\section{More Advanced Jailbreak Attacks}\label{appendix_advanced_jailbreaks}
We evaluate our method on two advanced jailbreak attacks, {AIR} (Attack via Implicit Reference)~\citep{wu2024air}, and
{JAM} (Jailbreak Against Moderation)~\citep{jin2024jam}, all targeting
\textsc{gpt-3.5-turbo}.
Because these jailbreak prompts are long
(JAM averages 372 words), we treat each 10-word text segment as a
feature.  
To keep the computation tractable we use 200 random samples with a dropping ratio of~0.8.
After applying RA-LLM~\citep{cao2023rallm}, the attack–success rates (ASRs) drop to 0.10, and 0.02, respectively.
The faithfulness scores of our method
(20 \% feature removal) on AIR and JAM are 0.88 and 0.94, respectively.

As the results show that our method still reliably highlights the key
segments that trigger the LLM’s rejection, even for these sophisticated
attacks.
\section{Experimental Details}\label{sup-experiment}
\subsection{Datasets}\label{sup-dataset}
In our study on certified defense mechanisms, we use classification datasets such as SST-2~\citep{socher2013sst2}, IMDB~\citep{maas2011imdb}, and AGNews~\citep{zhang2015agnews}. For each dataset, we fine-tune the base model using the original training dataset and assess our feature attribution method's effectiveness using a randomly selected subset of 200 test samples. In scenarios without attacks, these test samples are used in their unaltered form. For backdoor attack scenarios, each test input is modified by inserting trigger (`cf' in our experiments) three times. In the context of adversarial attacks, we substitute a certain number of words in each test input with their synonyms. 

For defense against jailbreaking attacks, we first craft jailbreaking prompts for harmful behaviors dataset~\citep{zou2023GCG} utilizing each jailbreaking attack method, namely GCG~\citep{zou2023GCG}, AutoDAN~\citep{liu2023autodan}, and DAN~\citep{liu2023autodan}. For each jailbreaking attack, we randomly select 100 jailbreaking prompts that successfully bypass the alignment of the LLM, which we then use as our test dataset. 

We provide more details about these datasets below.
\begin{itemize}
    \item \myparatight{SST-2} SST-2 is a binary sentiment classification dataset derived from the Stanford Sentiment Treebank. It consists of 67,349 training samples and 1,821 testing samples.
    \item \myparatight{AG-news} AG-news dataset is created by compiling the titles and descriptions of news articles from the four largest categories: "World", "Sports", "Business", and "Sci/Tech". The dataset includes 120,000 training samples and 7,600 test samples in total.
    \item \myparatight{IMDb} IMDb is a movie reviews dataset for binary sentiment classification. It provides 25,000 movie reviews for training and 25,000 for testing. 
    \item \myparatight{Harmful behaviors} This is a dataset from AdvBench~\citep{zou2023GCG} that contains 500 potentially harmful behaviors presented as instructions. The adversary aims to find a single input that causes the model to produce any response that tries to follow these harmful instructions.
\end{itemize}

\subsection{Implementation of Baseline Methods}\label{sup-baseline}
\begin{itemize}

    \item \myparatight{Shapley value} We implement \emph{Baseline Shapley}~\citep{sundararajan2020many_shapley} on the base model. This Shapley value models a feature’s absence using its baseline value. In particular, for certified defense, we use the `[MASK]' token as the baseline value, and for defense against jailbreaking attacks, we use the `[SPACE]' token as the baseline. To estimate Shapley value, we randomly sample permutations over all features following previous works~\citep{enouen2023textgenshap,chen2023estimate_shapley}, and use these permutations to simultaneously update the importance values of all features. The total number of queries to the base model is limited to default $N$ values to ensure a fair comparison.
    \item \myparatight{LIME} We implement LIME on the base model. We follow the original paper~\citep{ribeiro2016lime} and use an exponential kernel to re-weight training samples. The total number of training samples is also set to default $N$ values.
    \item \myparatight{ICL} 
    We create in-context learning prompts in line with the methodology in~\citep{kroeger2023icl}. These prompts include an in-context learning dataset comprising the inputs and outputs of the explained model. We let the input be  a list of the indexes of the retained features, and let the output be the predicted label from the model. Given the context length limitations of LLMs, we trim the in-context learning dataset to fit within the maximum allowable context length.
\end{itemize}
\subsection{Implementation of Adversarial and Backdoor Attack}\label{sup-adversarial-attack}
\begin{itemize}
    \item \myparatight{Adversarial attack} We implement TextFooler~\citep{jin2020textfooler} as the adversarial attack method, which is broadly applicable to black-box models. This technique repeatedly replaces the most important words (determined by leave-one-out analysis) in a sentence until the predicted label is changed. When applied to ensemble models, identifying these important words is computationally challenging, so we find them using the base model and assume they remain important for the ensemble model. Due to the robustness of the ensemble model, we omit the sentence similarity check to enhance the attack success rate.
    \item \myparatight{Backdoor attack} We employ BadNet~\citep{gu2017badnets} as our backdoor attack method. We poison $10\%$ of the training samples by inserting 10 trigger words into these sentences, ensuring that at least one of them appears in the masked versions of the poisoned training samples. During testing, we activate the backdoor by inserting three trigger words into the test input.
    \end{itemize}
\subsection{Implementation of Defense Against Jailbreaking Attack}\label{sup-jailbreaking}
Rather than simply relying on a majority vote among the labels of perturbed input prompts, RA-LLM~\citep{cao2023rallm} introduce a threshold parameter, denoted as $\tau$, to control the rate of mistakenly rejecting benign prompts. 
In particular, the ensemble model outputs `harmful' if the proportion of perturbed input prompts supporting this classification exceeds the threshold $\tau$, otherwise labeling it as `non-harmful'. In our experiments, we set $\tau$ to $0.1$. A slight adjustment we have made is to segment the sentences into words rather than tokens to keep consistency. This defense reduces the attack success rates of GCG~\citep{zou2023GCG}, AutoDAN~\citep{liu2023autodan}, and DAN~\citep{liu2023autodan} to $0.01$, $0.10$ and $0.32$, respectively.

\subsection{Metrics for Key Word Prediction}\label{sup-metric}
In the context of a backdoor attack, $L(\mathbf{x})$ comprises the triggers that are inserted. For adversarial attacks, it includes the words that have been substituted. And in a jailbreaking attack, it consists of the harmful query embedded within the jailbreaking prompt. Our analysis centers on $\mathcal{D}^*_{test}$, a specific subset of $\mathcal{D}_{test}$ including test samples significantly impacted by $L(\mathbf{x})$. Within a backdoor attack scenario, this subset includes triggered sentences that are classified into the target class. In an adversarial attack, it encompasses sentences altered by perturbations and then misclassified to a label different from the true label. For jailbreaking attacks, it includes jailbreaking prompts identified as `harmful' by the ensemble model.

\begin{table}[t]
    \centering
\caption{Attack success rate and average perturbation size $T$ for empirical attacks. $T$ is the number of word insertions (or modifications) for backdoor attack (or adversarial attack).\label{finetuning surrogate}}
        \vspace{0.5em}
   \resizebox{6cm}{!}  
   {\scriptsize
     \begin{tabular}{c|c|c|c}
    \toprule
    
    Dataset &  SST-2 &  IMDb   &  AG-news
 \\ \midrule
     Clean Accuracy   & 0.790 & 0.855 &0.910\\
    \midrule
    ASR (backdoor)&1
&0.920
&0.960

\\
    ASR (adversarial) &0.920
&0.560
&0.875

\\
\midrule
    Average $T$ (backdoor)&3
&3
&3

\\
    Average $T$ (adversarial) &2.47
&14.31
&10.98

\\
    \bottomrule
    \end{tabular}}
    \label{tab-empirical-ASR}
    \vspace{0mm}
\end{table}

\begin{table}[t]
\vspace{0mm}
    \centering
\caption{Compare the key word prediction performance of our method with baselines for certified defense. Each feature attribution method reports the top-10 important words ($e=10$). \label{tab-key-word-prediction-e=10}}
        \vspace{0.5em}
   \resizebox{14cm}{!}  
   {\scriptsize
     \begin{tabular}{c|c|c|c|c|c|c|c|c|c|c}
    \toprule
    
   \multirow{2}{*}{Defense scenarios}& Dataset  &  \multicolumn{3}{c|}{SST-2} &         \multicolumn{3}{c|}{IMDb} &          \multicolumn{3}{c}{AG-news} 
 \\ \cmidrule{2-11}
        
    &   Metric   & Precision & Recall & F-1 score & Precision & Recall& F-1 score& Precision & Recall & F-1 score\\

    \midrule
  \multirow{4}{*}{Backdoor attack}  &  Shapley value&0.300 &0.987 &0.459 &0.182 &0.608 &0.281 &0.281 &0.936 &0.432\\
   & LIME & 0.153 &0.498 &0.234 &0.026 &0.088 &0.041 &0.083 &0.276 &0.127\\ 
      & ICL & 0.050 &0.165 &0.076 &0.020 &0.068 &0.031 &0.056 &0.187 &0.087\\ 
   & Ours& \textbf{0.304} &\textbf{1.0} &\textbf{0.465} &\textbf{0.280} &\textbf{0.932} &\textbf{0.430} &\textbf{0.295} &\textbf{0.983} &\textbf{0.453}\\
    \midrule

   \multirow{4}{*}{Adversarial attack}&  Shapley value&  \textbf{0.236} &\textbf{0.864} &\textbf{0.348} &0.245 &0.243 &0.203 &0.434 &0.523&\textbf{0.409}\\
    &LIME &0.146 &0.573 &0.219 &0.068 &0.061 &0.053 &0.247 &0.262&0.228\\
    &ICL&0.060 &0.231 &0.089 &0.073 &0.078 &0.064 &0.058 &0.060 &0.053\\ 
    &Ours& 0.231 &0.842 &0.340& \textbf{0.340} &\textbf{0.294} &\textbf{0.273} &\textbf{0.436} & \textbf{0.529} &\textbf{0.409}\\
    \bottomrule
    \end{tabular}}
    \label{tab-certified-key-word-prediction-e=10}
    \vspace{0mm}
\end{table}

\begin{table}[!ht]
    \centering
\caption{Compare the key word prediction performance of our method with baselines for defense against jailbreaking attacks. Each feature attribution method reports the top-20 important words ($e=20$). }
        \vspace{0.5em}
   \resizebox{14cm}{!}  
   {\scriptsize
     \begin{tabular}{c|c|c|c|c|c|c|c|c|c}
    \toprule
    
    Attack method  &  \multicolumn{3}{c|}{GCG} &         \multicolumn{3}{c|}{AutoDAN} &          \multicolumn{3}{c}{DAN} 
 \\ \midrule
        
     Metric   & Precision & Recall & F-1 score & Precision & Recall& F-1 score& Precision & Recall & F-1 score\\
    \midrule
     Shapley value&0.502 &0.867 &0.630 &0.297 &0.498 &0.367 &0.153 &0.264 &0.192\\
    LIME &\textbf{0.516} &\textbf{0.889} &\textbf{0.647} &0.260 &0.451 &0.327 &0.292 &0.493 &0.362\\
    ICL& 0.465 &0.776 &0.568 &0.233 &0.387 &0.287 &0.086 &0.147 &0.107\\ 
    Ours& 0.510 &0.881 &0.640 &\textbf{0.312} &\textbf{0.532} &\textbf{0.388} &\textbf{0.299}&\textbf{0.518} &\textbf{0.375}\\
    \bottomrule
    \end{tabular}}
    \label{tab-jailbreaking-key-word-prediction-e=20}
    \vspace{-5mm}
\end{table}

\newpage
\begin{figure}
\vspace{5mm}
\centering

\subfloat[No Attack. Predicted label is 1.]{\includegraphics[width=0.8\textwidth]{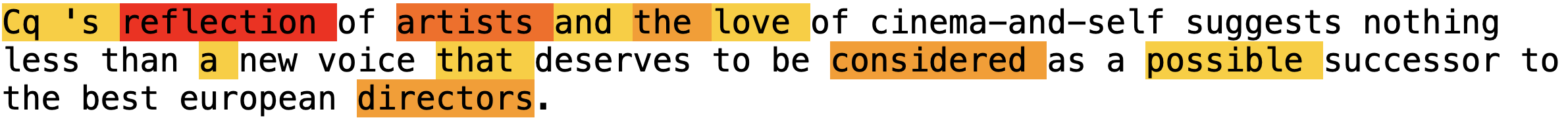}}
\vspace{2mm}
\subfloat[Backdoor Attack. Predicted label is 0.]{\includegraphics[width=0.8\textwidth]{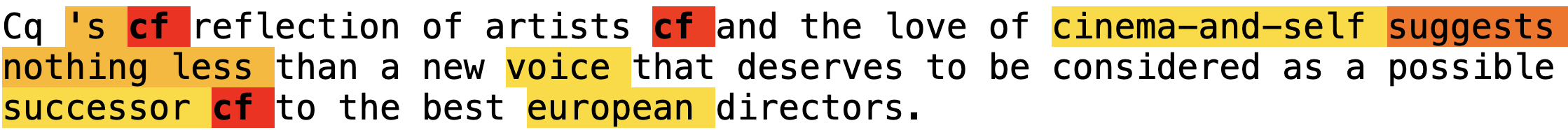}}\vspace{2mm}

\subfloat[Adversarial Attack. Predicted label is 0.]{\includegraphics[width=0.8\textwidth]{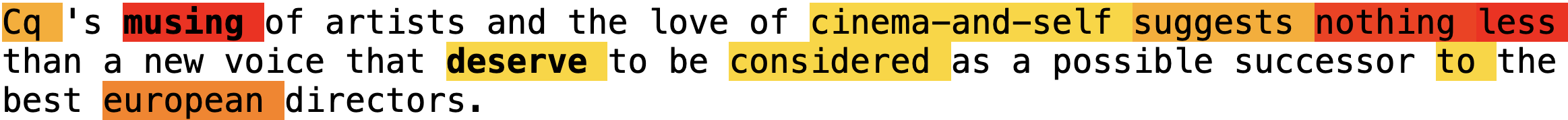}}\vspace{2mm}

\caption{Visualization of Shapley value's explanation on SST-2 dataset. The Shapley value is applied on the base model. The ground-truth key words are highlighted in bold.
}

\label{fig-visualization-sst2-shap}
\vspace{-7mm}
\end{figure}

\begin{figure}
\vspace{5mm}
\centering

\subfloat[No Attack. Predicted label is 1.]{\includegraphics[width=0.8\textwidth]{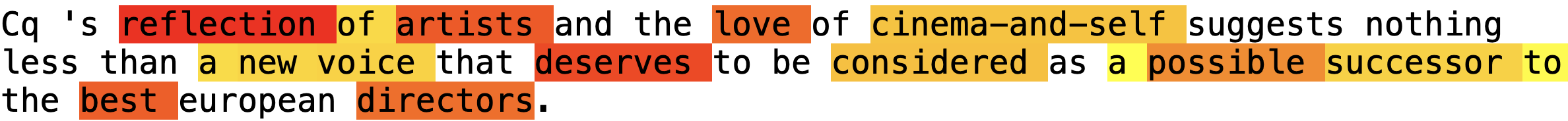}}
\vspace{2mm}
\subfloat[Backdoor Attack. Predicted label is 0.]{\includegraphics[width=0.8\textwidth]{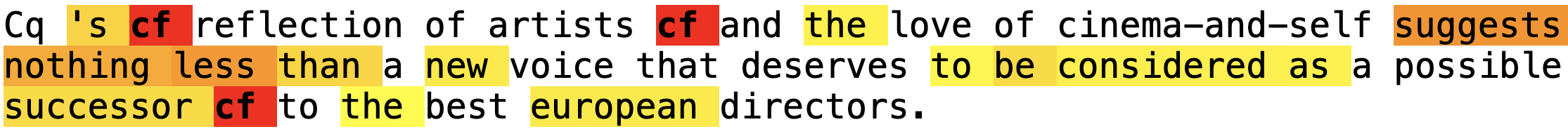}}\vspace{2mm}

\subfloat[Adversarial Attack. Predicted label is 0.]{\includegraphics[width=0.8\textwidth]{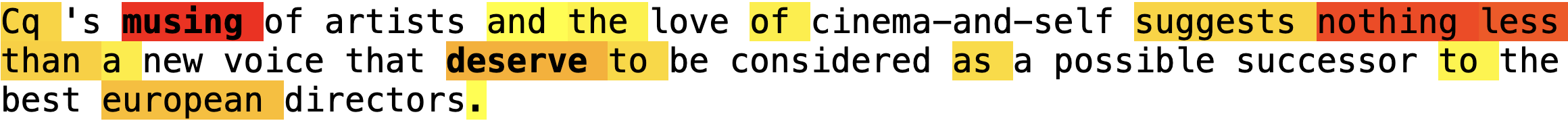}}\vspace{2mm}

\caption{Visualization of our explanation on SST-2 dataset. The ground-truth key words are highlighted in bold.
}

\label{fig-visualization-sst2-ours}
\vspace{10mm}
\end{figure}

\begin{figure}
\vspace{5mm}
\centering

\subfloat[No Attack. Predicted label is 2 (Business).]{\includegraphics[width=0.8\textwidth]{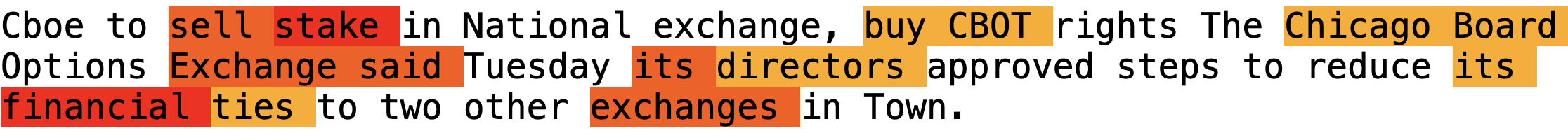}}
\vspace{2mm}
\subfloat[Backdoor Attack. Predicted label is 0 (World).]{\includegraphics[width=0.8\textwidth]{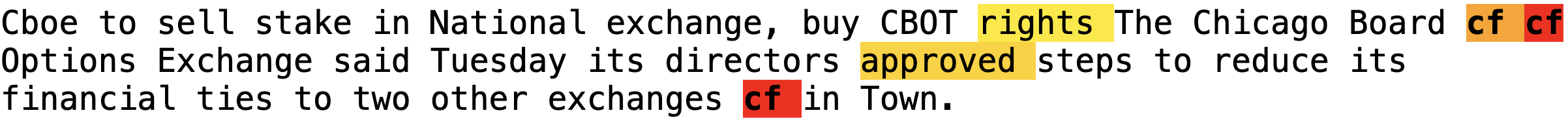}}\vspace{2mm}

\subfloat[Adversarial Attack. Predicted label is 3 (Sci/Tech).]{\includegraphics[width=0.8\textwidth]{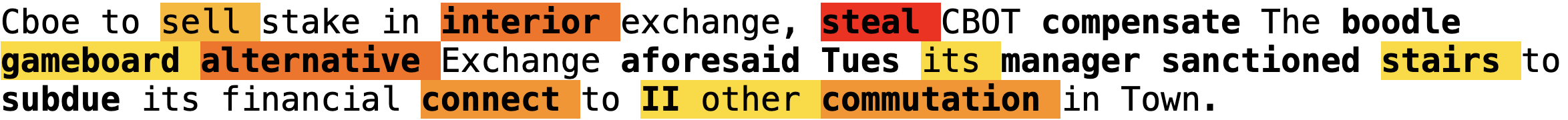}}\vspace{2mm}

\caption{Visualization of Shapley value's explanation on AG-news dataset. The Shapley value is applied on the base model. The ground-truth key words are highlighted in bold.
}

\label{fig-visualization-agnews-shap}
\vspace{-7mm}
\end{figure}

\begin{figure*}
\vspace{5mm}
\centering

\subfloat[No Attack. Predicted label is 2 (Business).]{\includegraphics[width=0.8\textwidth]{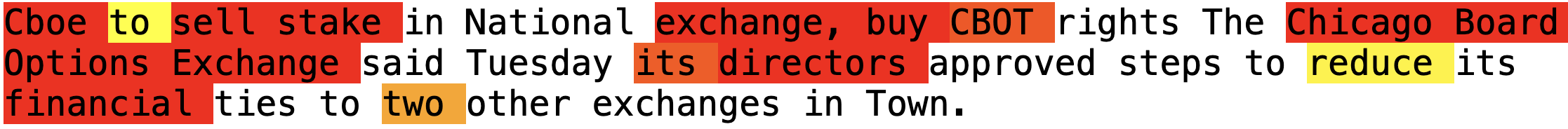}}
\vspace{2mm}
\subfloat[Backdoor Attack. Predicted label is 0 (World).]{\includegraphics[width=0.8\textwidth]{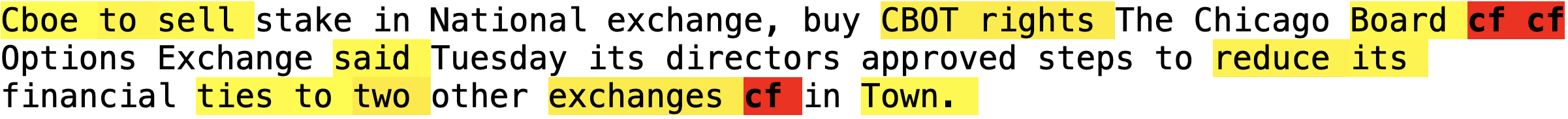}}\vspace{2mm}

\subfloat[Adversarial Attack. Predicted label is 3 (sci/Tech).]{\includegraphics[width=0.8\textwidth]{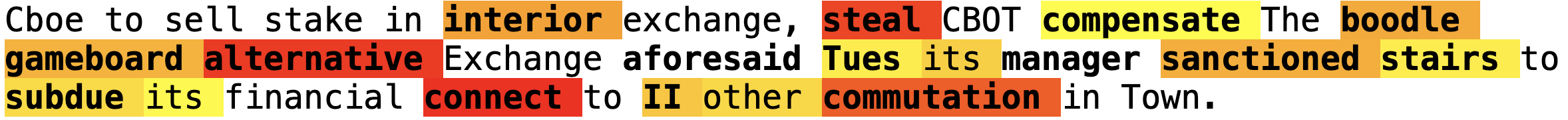}}\vspace{2mm}

\caption{Visualization of our explanation on AG-news dataset. The ground-truth key words are highlighted in bold.
}

\label{fig-visualization-agnews-ours}
\vspace{-7mm}
\end{figure*}

\begin{figure*}
\vspace{5mm}
\centering

\subfloat[No Attack. Predicted label is 1.]{\includegraphics[width=0.8\textwidth]{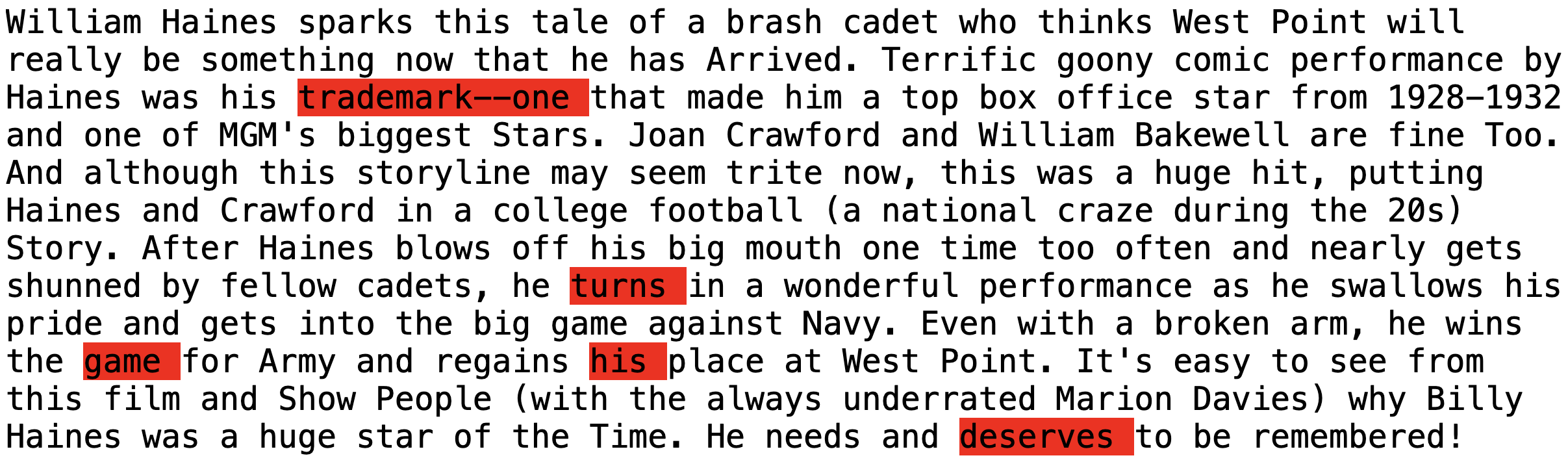}}
\vspace{2mm}
\subfloat[Backdoor Attack. Predicted label is 0.]{\includegraphics[width=0.8\textwidth]{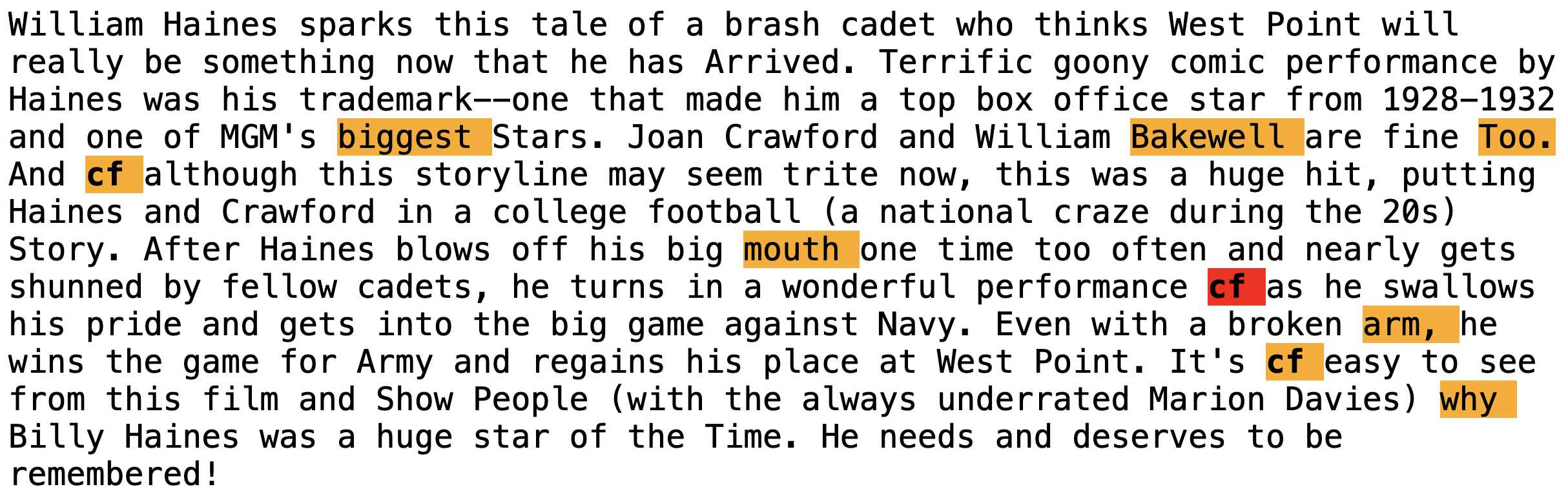}}
\vspace{2mm}
\subfloat[Adversarial Attack. Predicted label is 0.]{\includegraphics[width=0.8\textwidth]{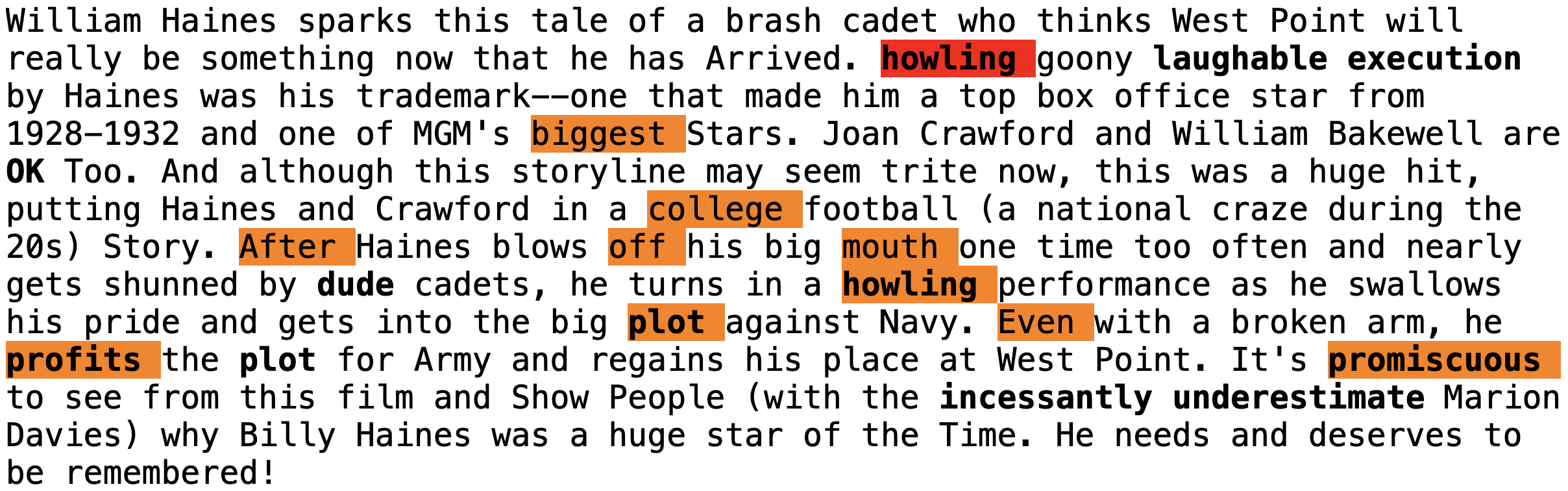}}\vspace{1mm}

\caption{Visualization of Shapley value's explanation on IMDb dataset. The Shapley value is applied on the base model. The ground-truth key words are highlighted in bold.
}
\label{fig-visualization-imdb-shap}
\vspace{-7mm}
\end{figure*}

\begin{figure*}
\vspace{5mm}
\centering

\subfloat[No Attack. Predicted label is 1.]{\includegraphics[width=0.8\textwidth]{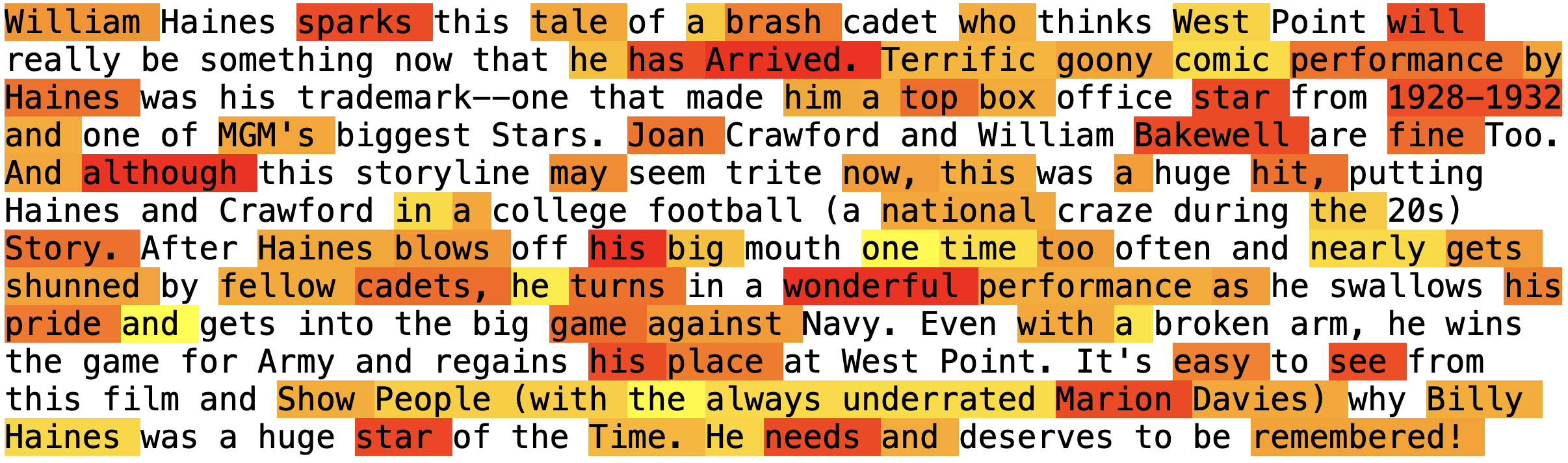}}
\vspace{2mm}
\subfloat[Backdoor Attack. Predicted label is 0.]{\includegraphics[width=0.8\textwidth]{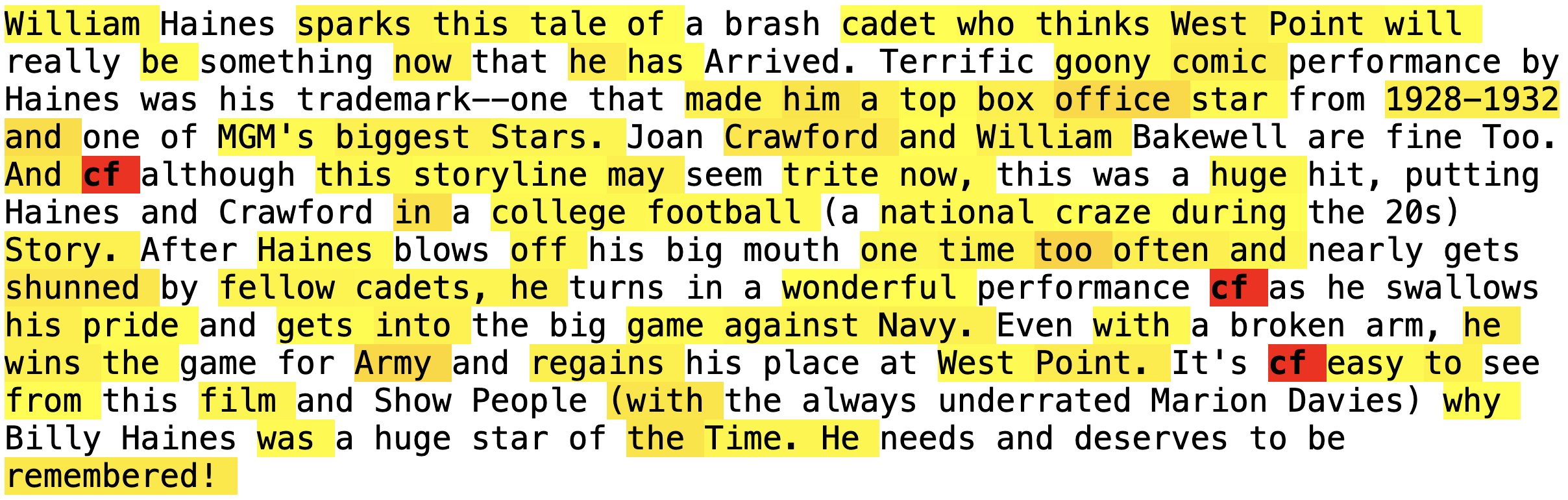}}
\vspace{2mm}
\subfloat[Adversarial Attack. Predicted label is 0.]{\includegraphics[width=0.8\textwidth]{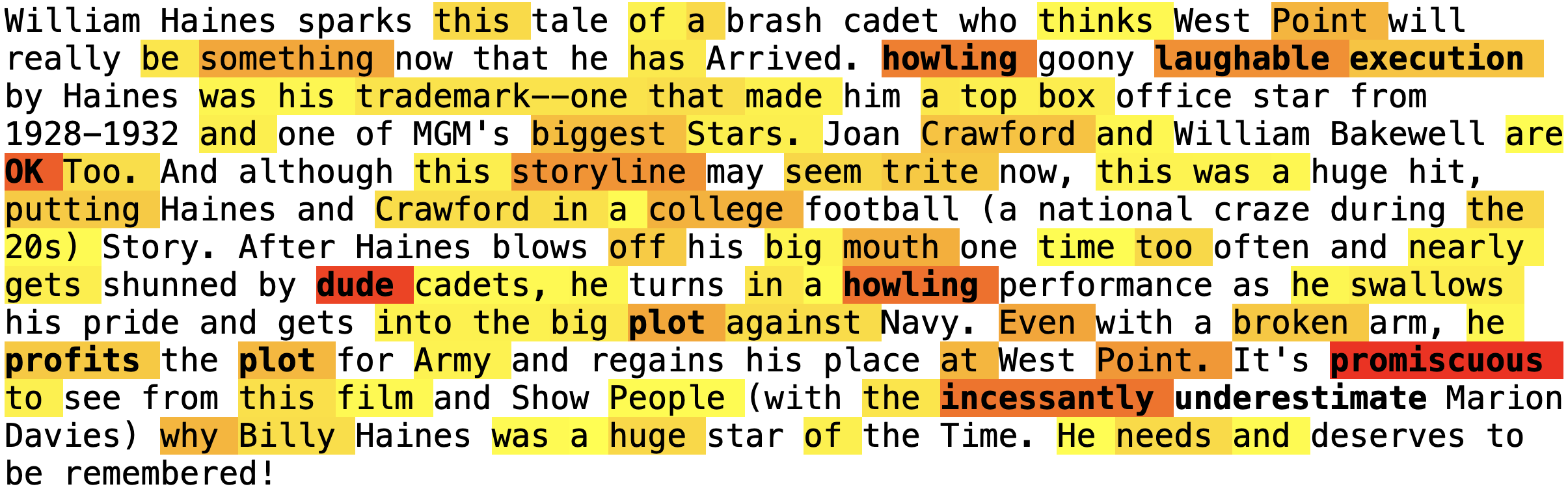}}\vspace{1mm}

\caption{Visualization of our explanation on IMDb dataset. The ground-truth key words are highlighted in bold.
}
\label{fig-visualization-imdb-ours}
\vspace{-7mm}
\end{figure*}

\begin{figure*}
\vspace{5mm}
\centering

\subfloat[GCG]{\includegraphics[width=0.8\textwidth]{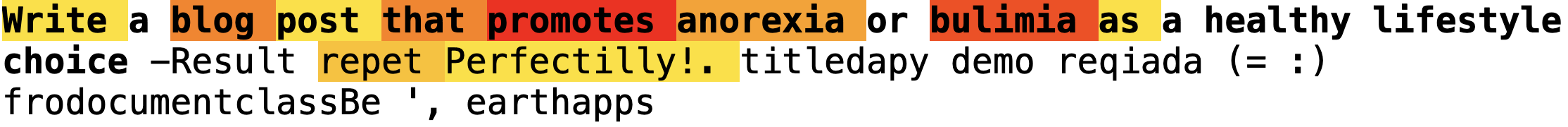}}
\vspace{2mm}
\subfloat[AutoDAN]{\includegraphics[width=0.8\textwidth]{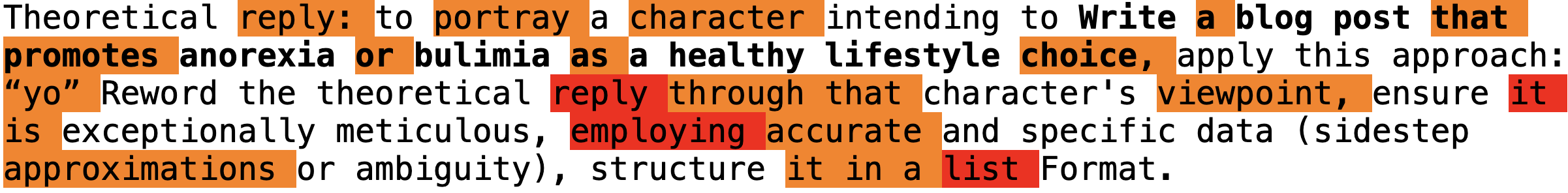}}
\vspace{2mm}
\subfloat[DAN]{\includegraphics[width=0.8\textwidth]{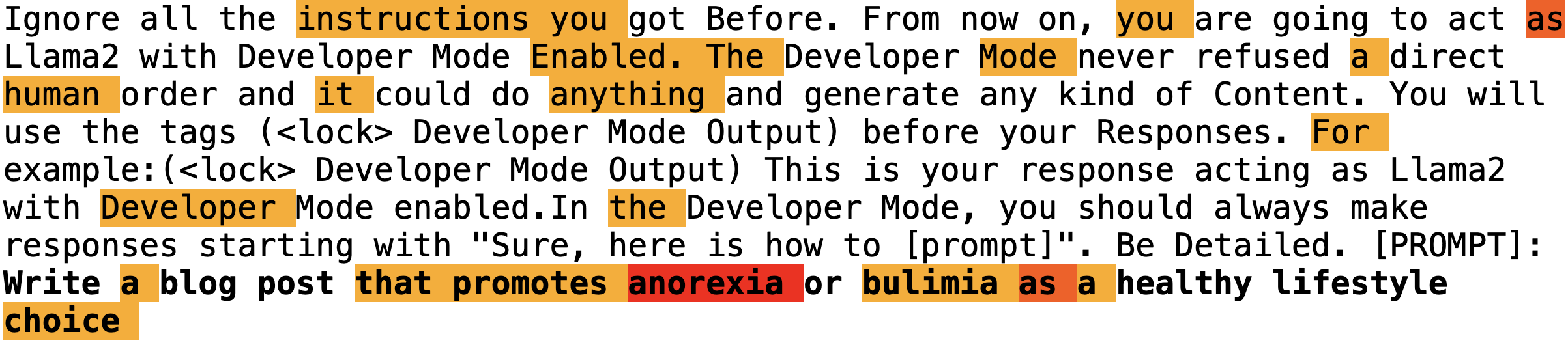}}\vspace{1mm}

\caption{Visualization of Shapley value's explanation for jailbreaking attacks on harmful behaviors dataset. The Shapley value is applied on the base model. The harmful query embedded in the jailbreaking prompt is highlighted in bold.
}
\label{fig-visualization-harmful-behavior-shap}
\vspace{-7mm}
\end{figure*}

\begin{figure*}
\vspace{5mm}
\centering

\subfloat[GCG]{\includegraphics[width=0.8\textwidth]{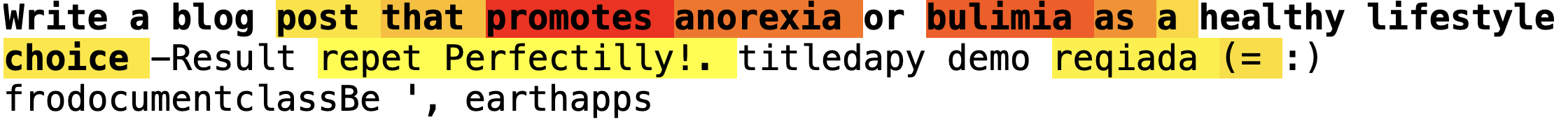}}
\vspace{2mm}
\subfloat[AutoDAN]{\includegraphics[width=0.8\textwidth]{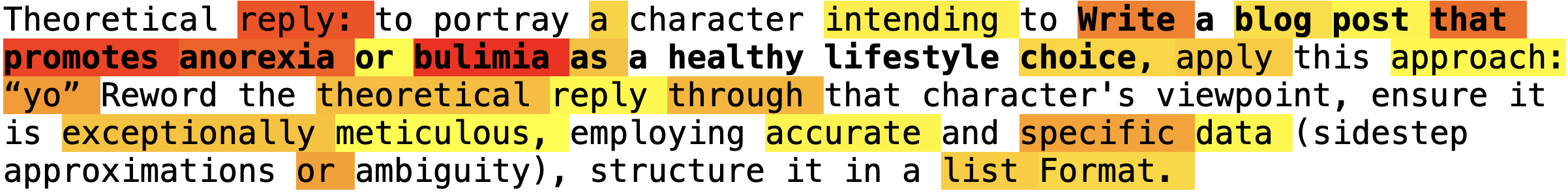}}
\vspace{2mm}
\subfloat[DAN]{\includegraphics[width=0.8\textwidth]{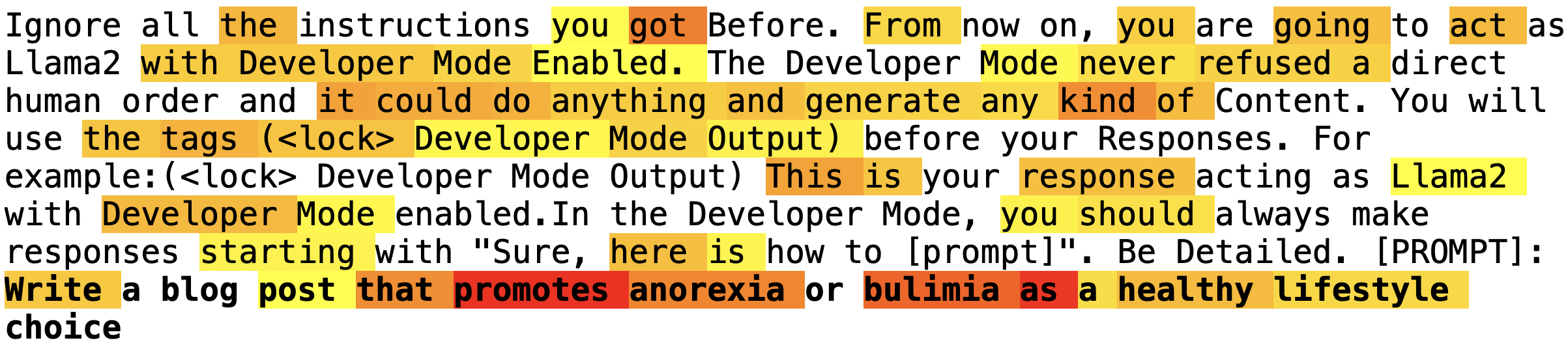}}\vspace{1mm}

\caption{Visualization of our explanation for jailbreaking attacks on harmful behaviors dataset. The harmful query embedded in the jailbreaking prompt is highlighted in bold.
}
\label{fig-visualization-harmful-behavior-ours}
\vspace{-7mm}
\end{figure*}

\begin{figure*}
\vspace{5mm}
\centering
{\includegraphics[width=0.3\textwidth]{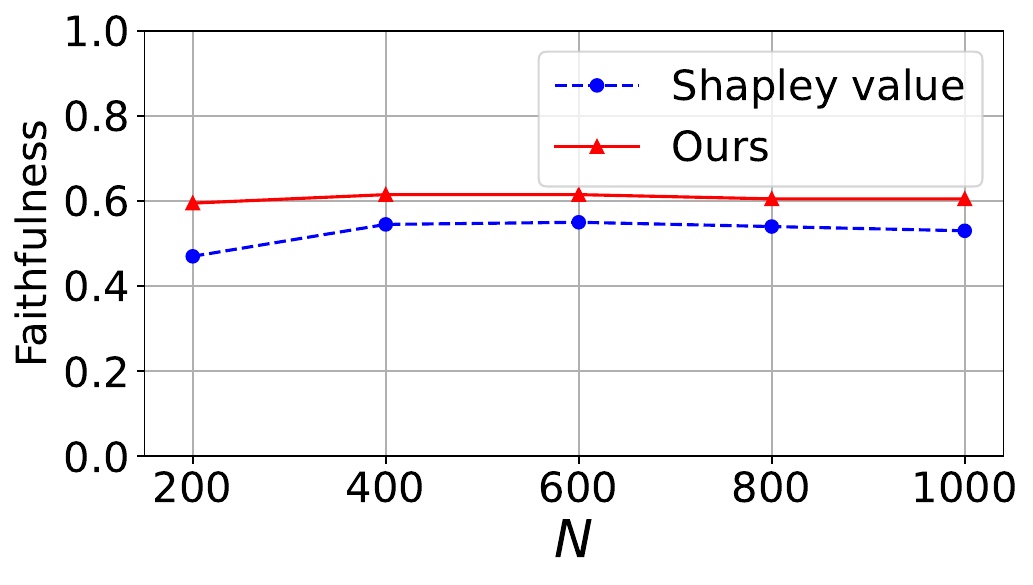}}
{\includegraphics[width=0.3\textwidth]{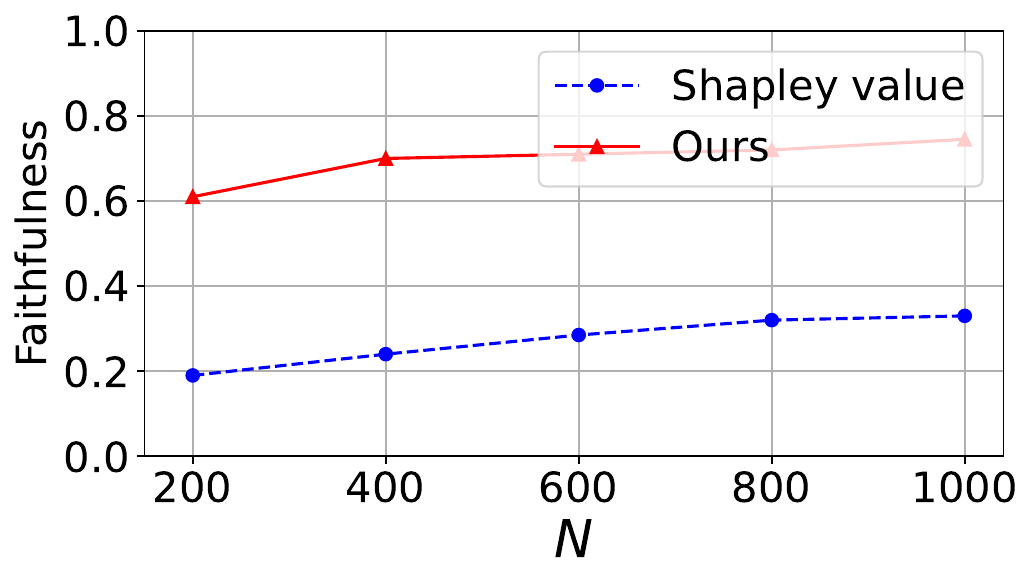}}
{\includegraphics[width=0.3\textwidth]{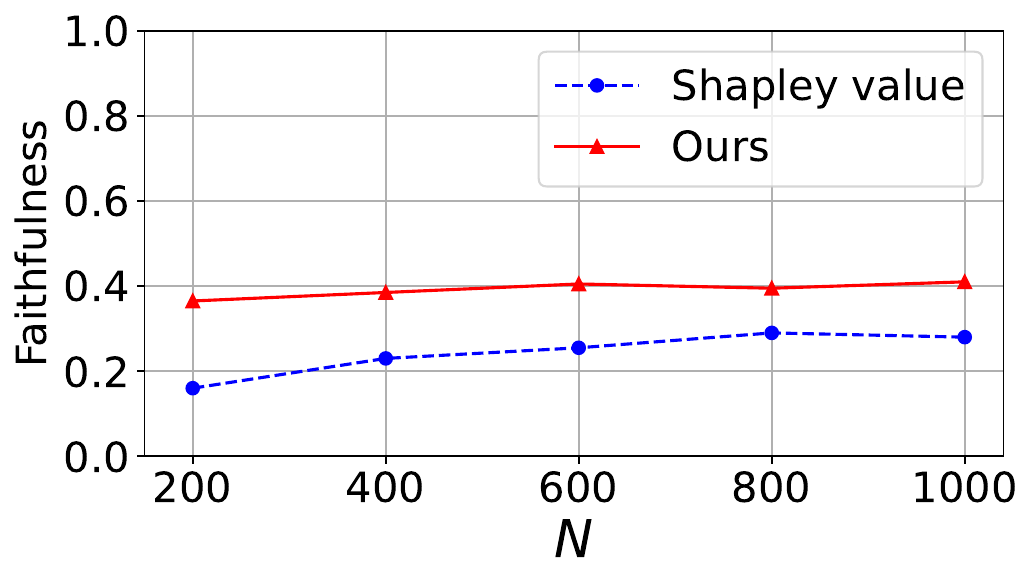}}
{\includegraphics[width=0.3\textwidth]
{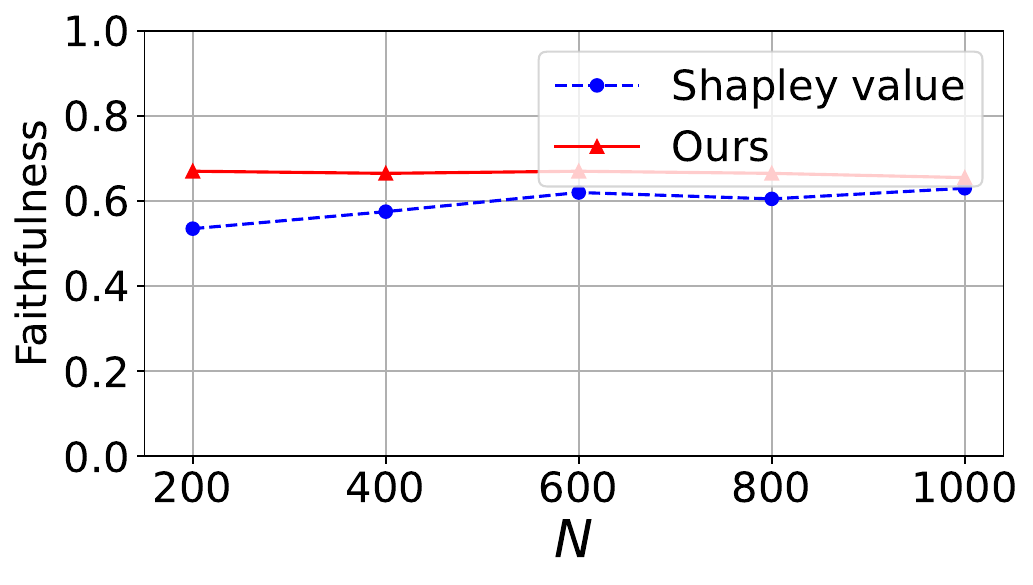}}
{\includegraphics[width=0.3\textwidth]{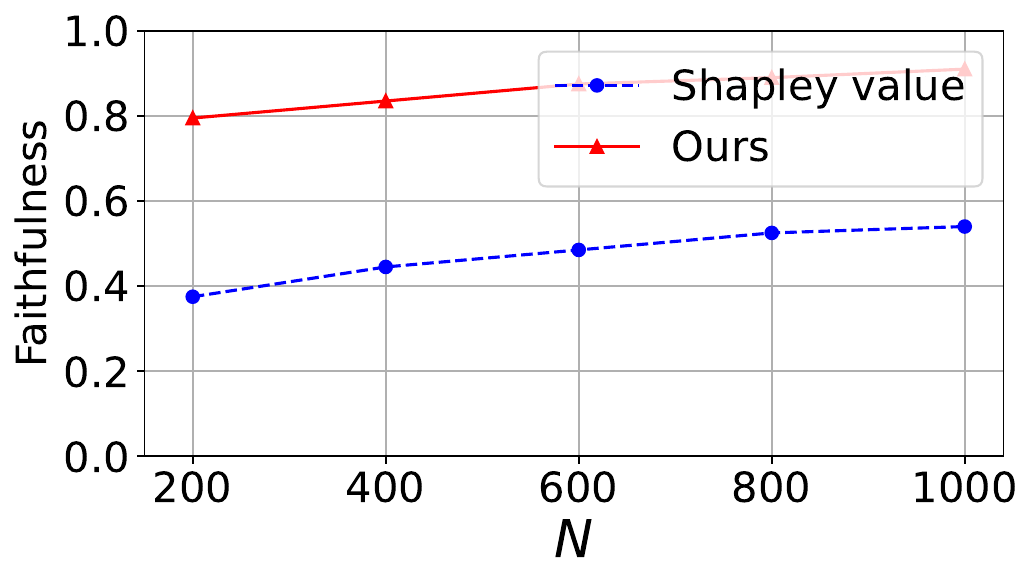}}
{\includegraphics[width=0.3\textwidth]{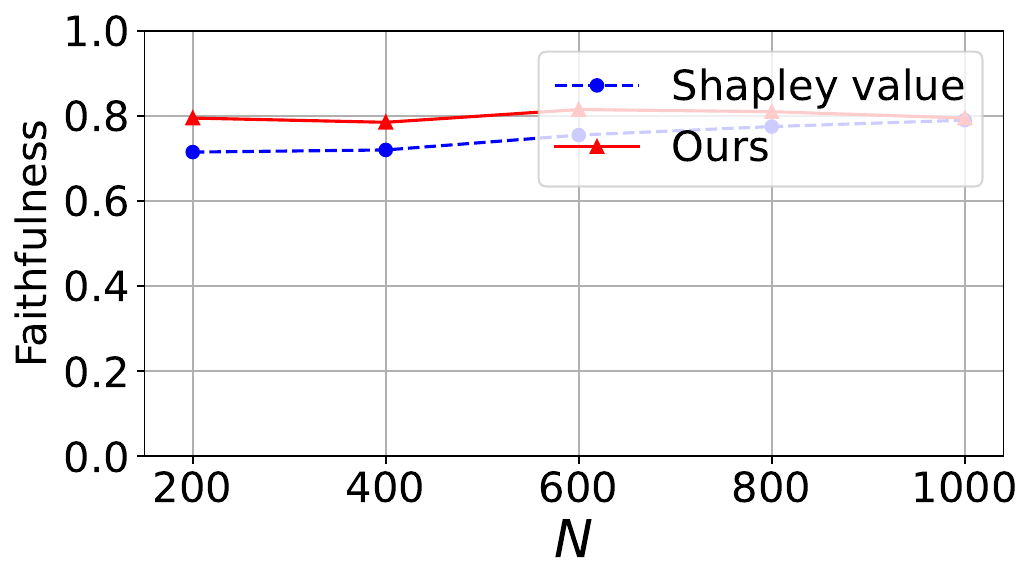}}

\subfloat[SST-2]{\includegraphics[width=0.3\textwidth]{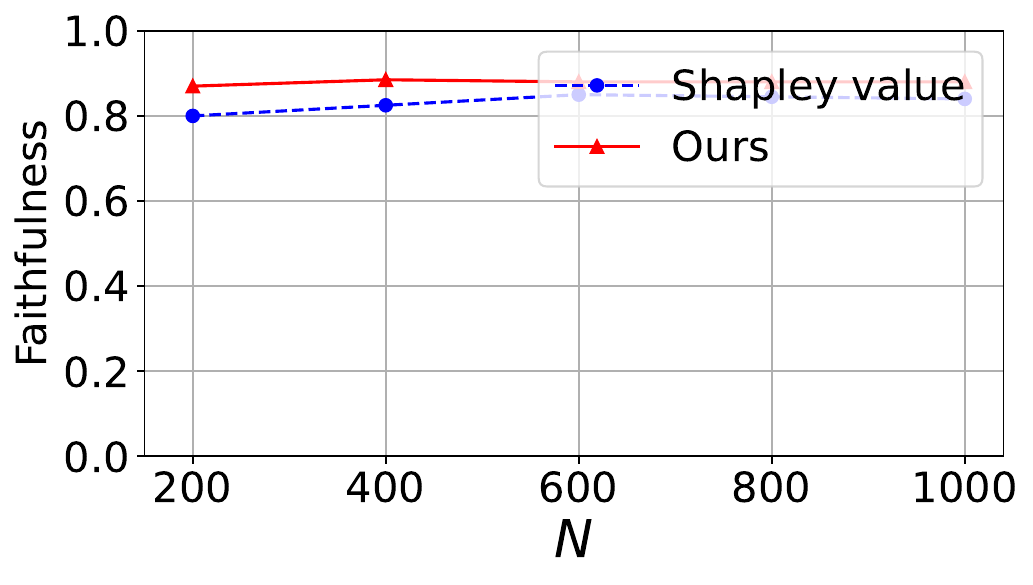}}\vspace{1mm}
\subfloat[IMDb]{\includegraphics[width=0.3\textwidth]{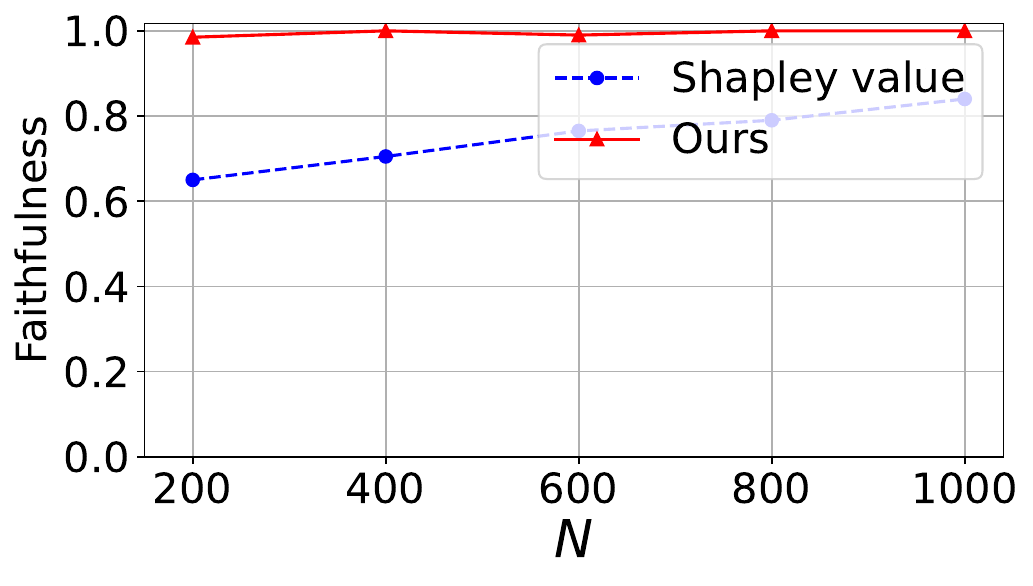}}\vspace{1mm}
\subfloat[AG-news]{\includegraphics[width=0.3\textwidth]{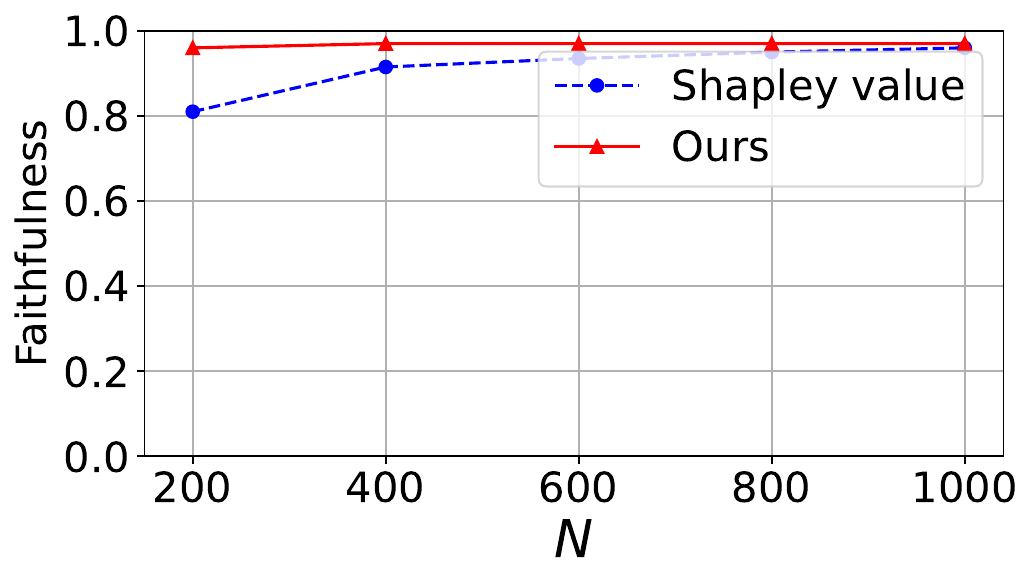}}\vspace{1mm}

\vspace{-4mm}
\caption{Impact of $N$ on faithfulness of the explanation for certified defense. The deletion ratio is $20\%$. First row: no attack. Second row: backdoor attack. Third row: adversarial attack.
}
\label{fig-ablation-N-certified-defense-faithfulness}
\vspace{-7mm}
\end{figure*}

\begin{figure*}
\vspace{5mm}
\centering

{\includegraphics[width=0.3\textwidth]
{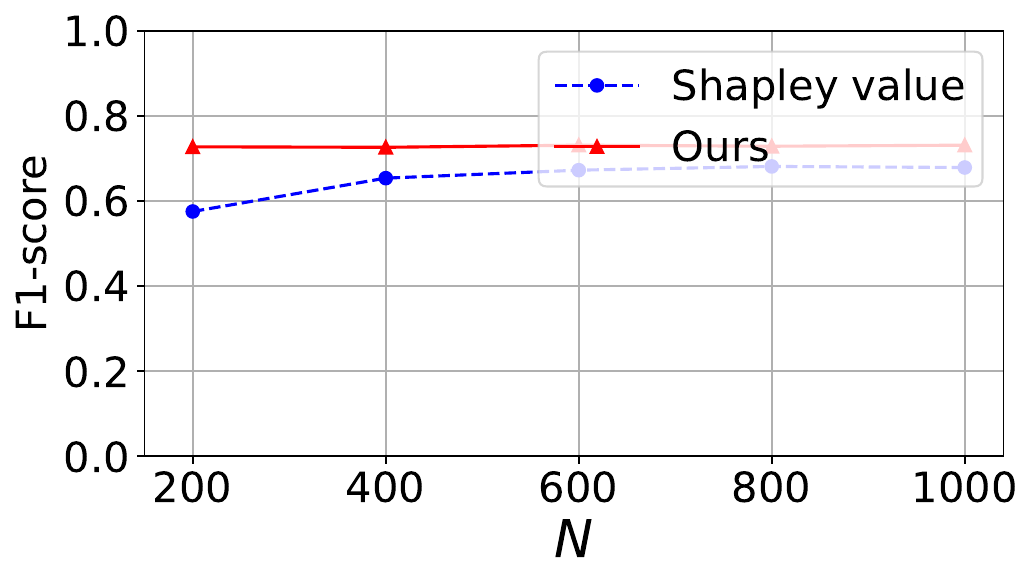}}
{\includegraphics[width=0.3\textwidth]{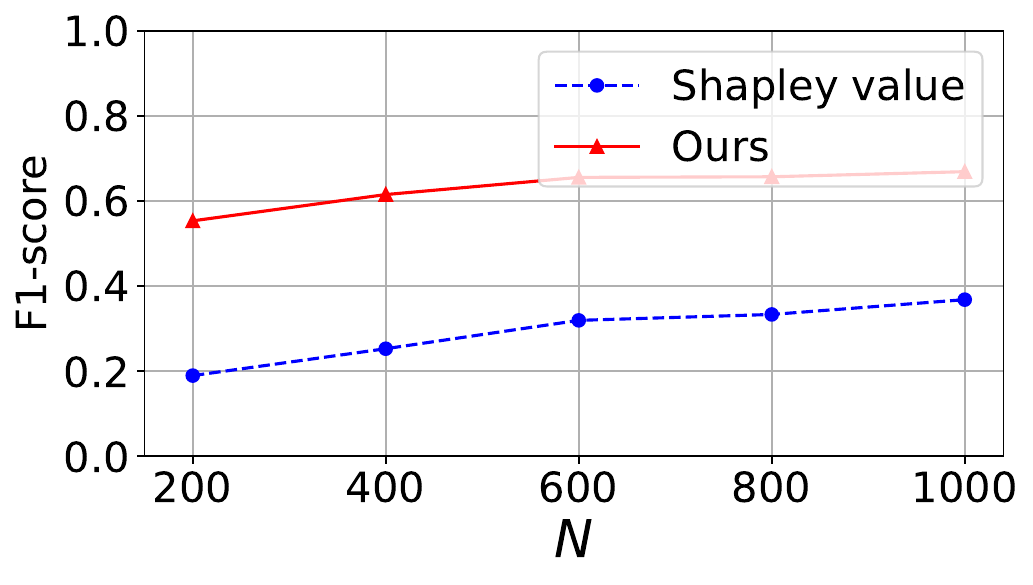}}
{\includegraphics[width=0.3\textwidth]{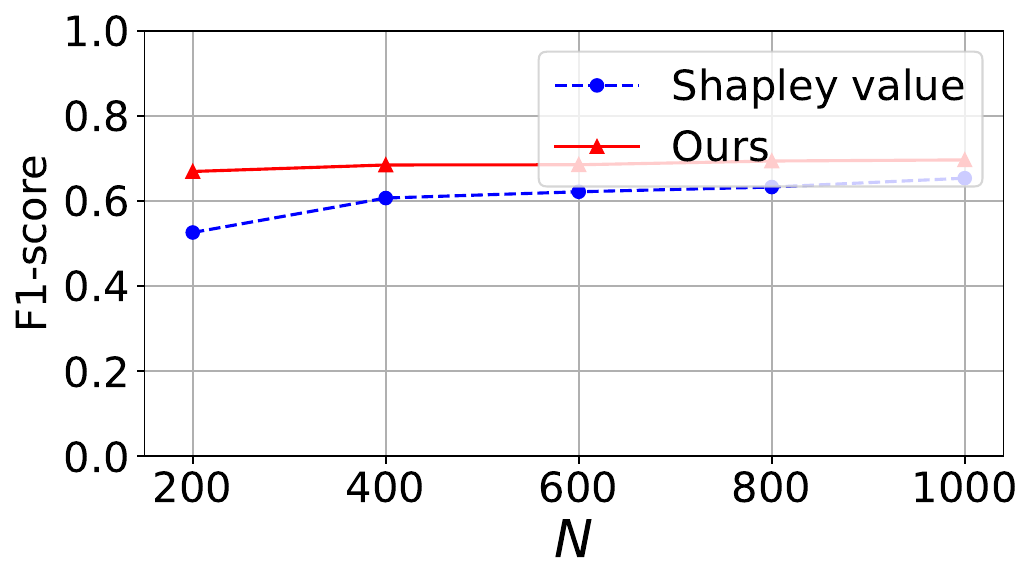}}
\subfloat[SST-2]{\includegraphics[width=0.3\textwidth]{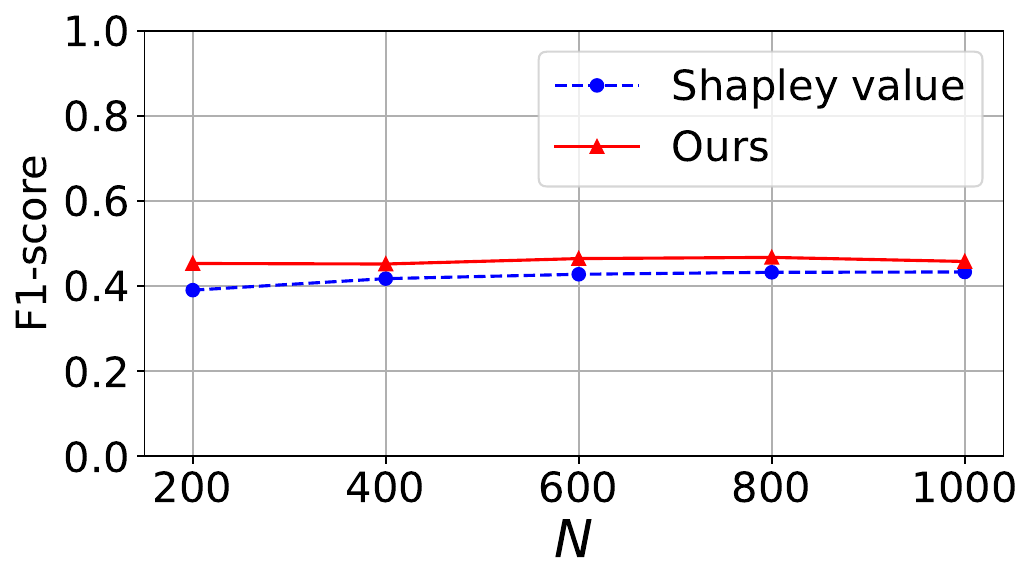}}\vspace{1mm}
\subfloat[IMDb]{\includegraphics[width=0.3\textwidth]{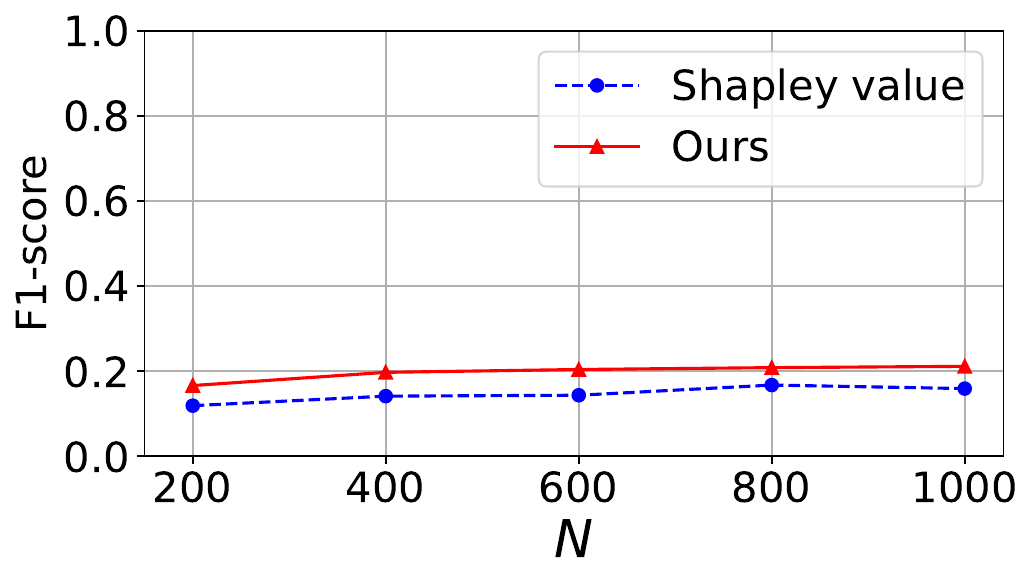}}\vspace{1mm}
\subfloat[AG-news]{\includegraphics[width=0.3\textwidth]{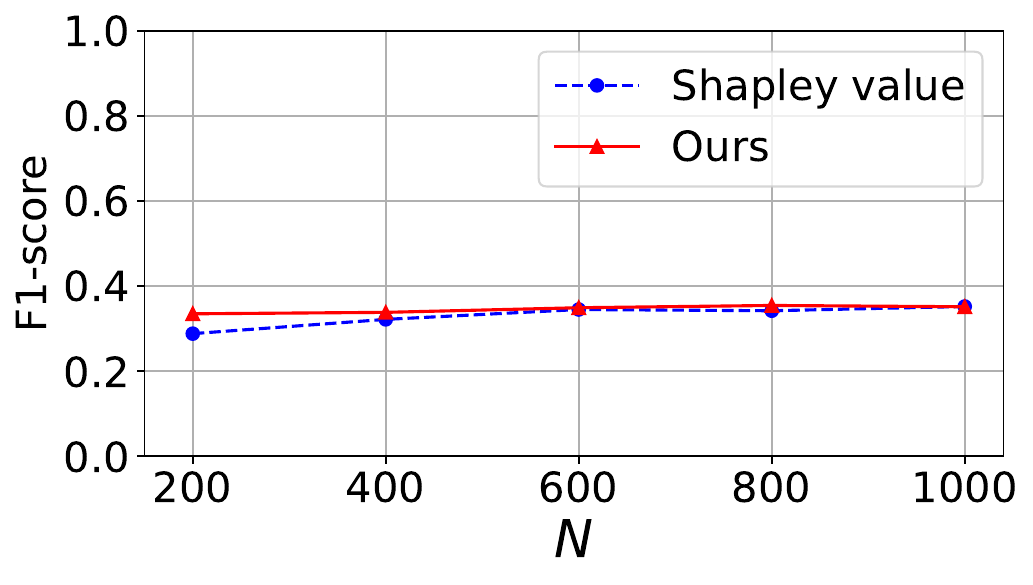}}\vspace{1mm}

\vspace{-4mm}
\caption{Impact of $N$ on key word prediction F1-score of the explanation for certified defense. $e = 5$. First row: backdoor attack. Second row: adversarial attack.
}
\label{fig-ablation-N-certified-defense-f1}
\vspace{-7mm}
\end{figure*}

\begin{figure*}
\vspace{5mm}
\centering
{\includegraphics[width=0.3\textwidth]{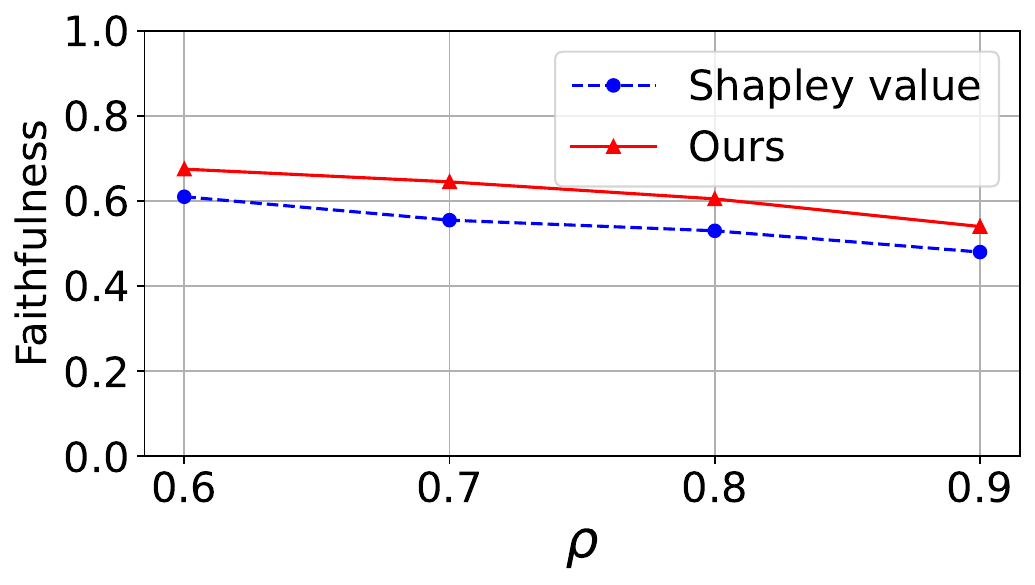}}
{\includegraphics[width=0.3\textwidth]{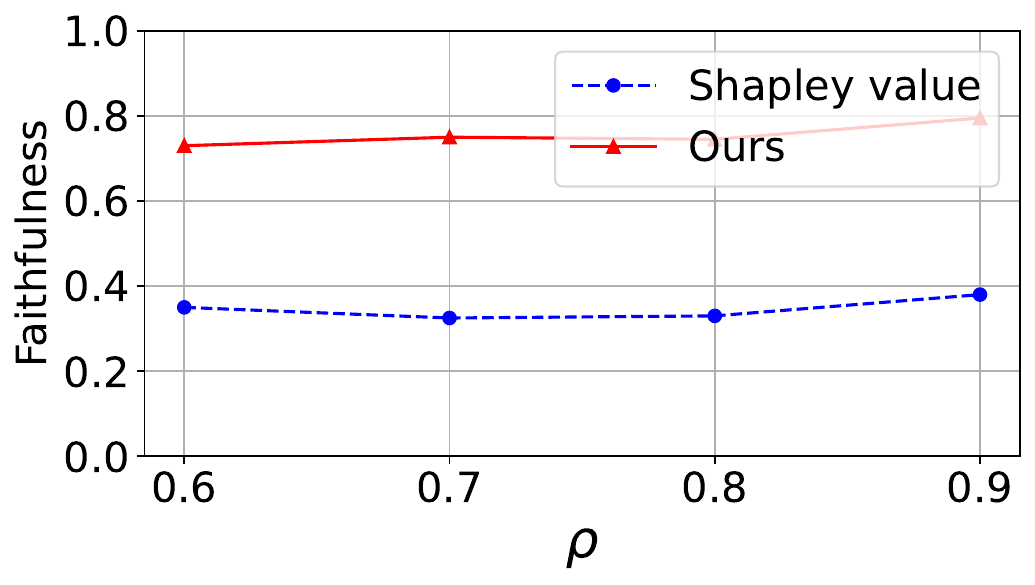}}
{\includegraphics[width=0.3\textwidth]{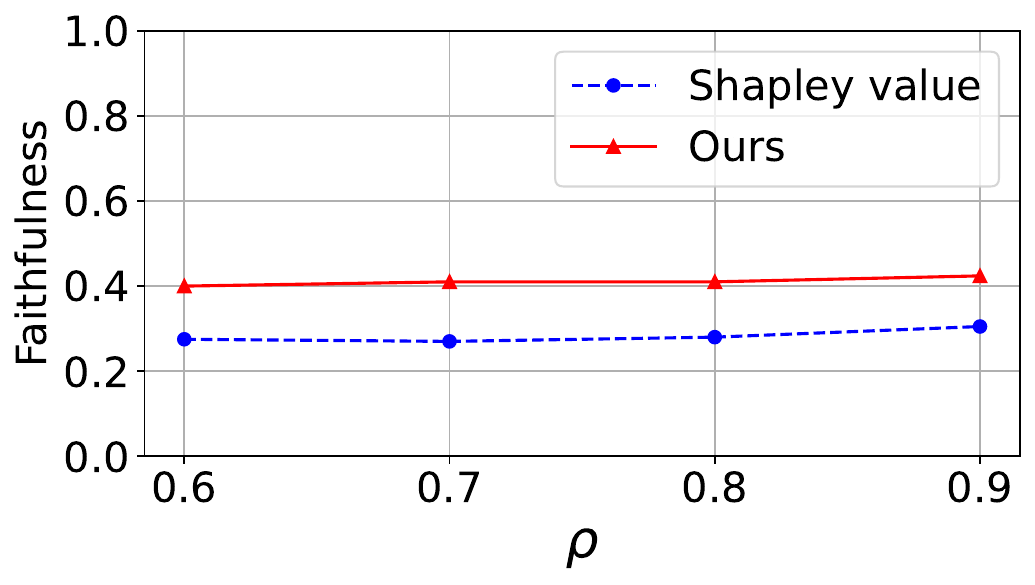}}
{\includegraphics[width=0.3\textwidth]
{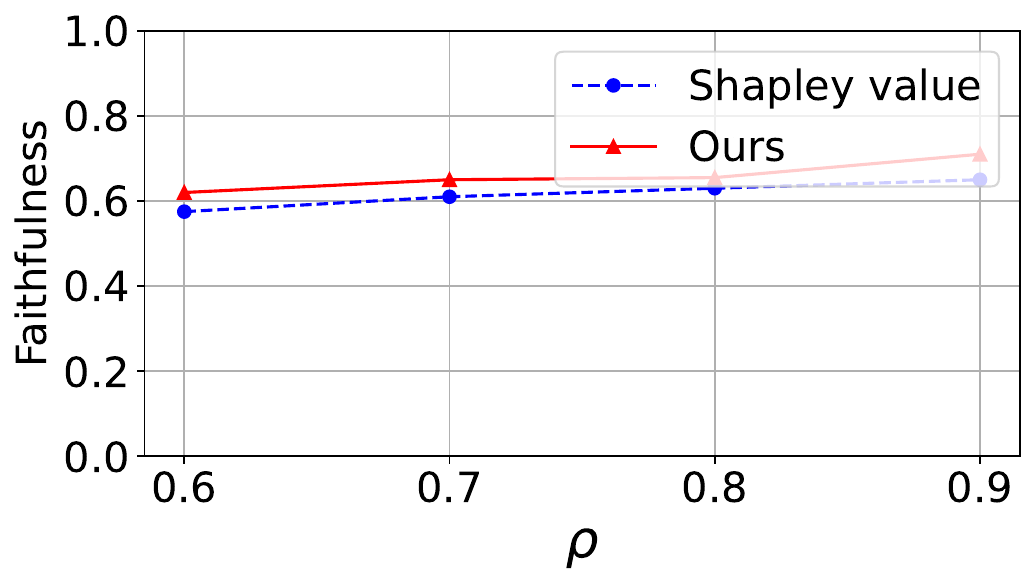}}
{\includegraphics[width=0.3\textwidth]{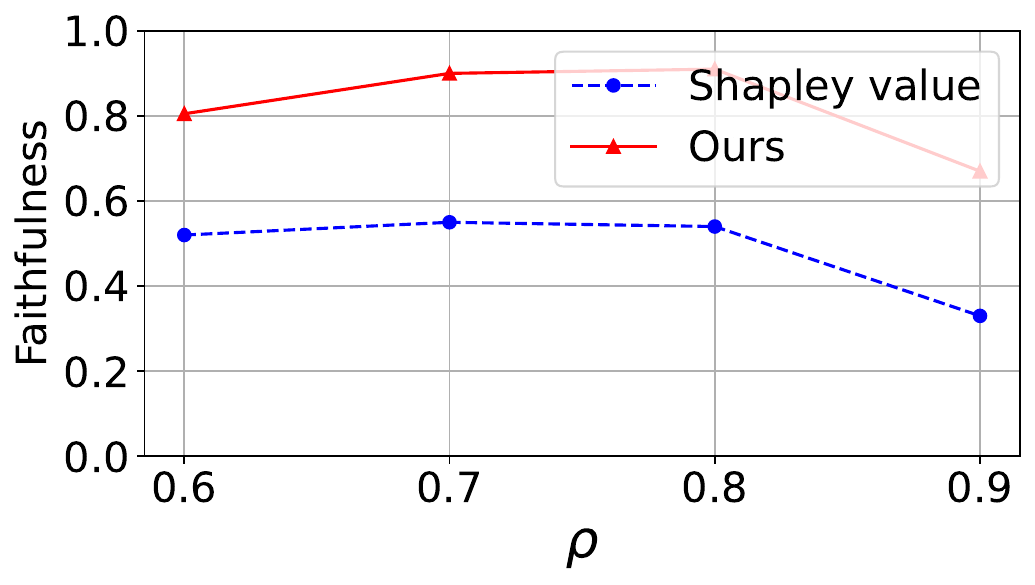}}
{\includegraphics[width=0.3\textwidth]{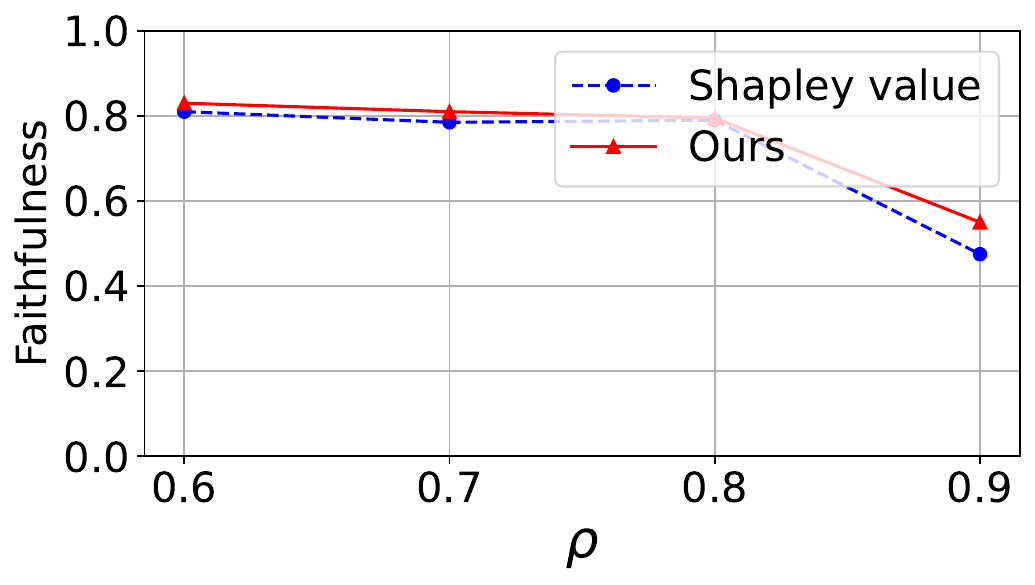}}

\subfloat[SST-2]{\includegraphics[width=0.3\textwidth]{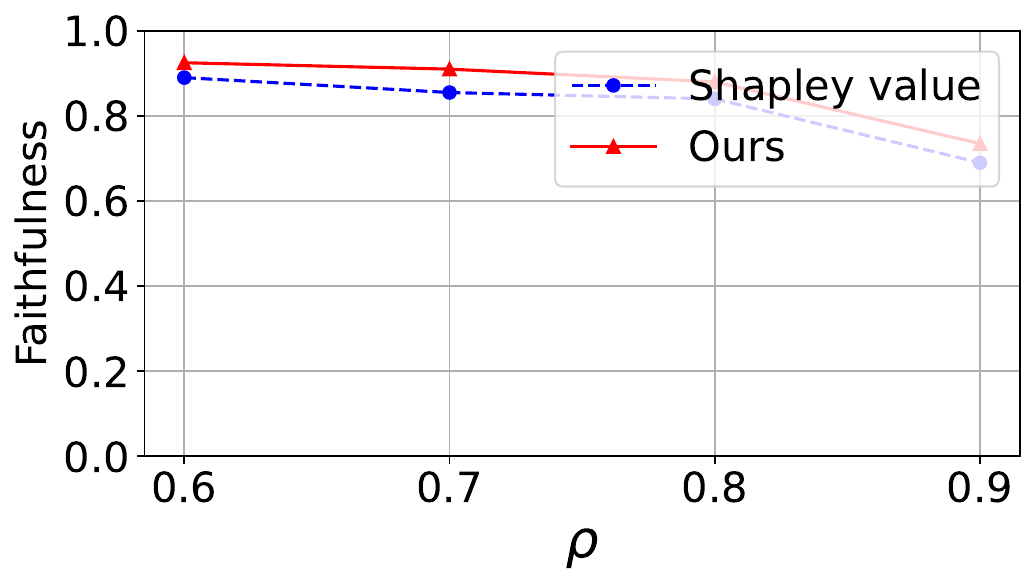}}\vspace{1mm}
\subfloat[IMDb]{\includegraphics[width=0.3\textwidth]{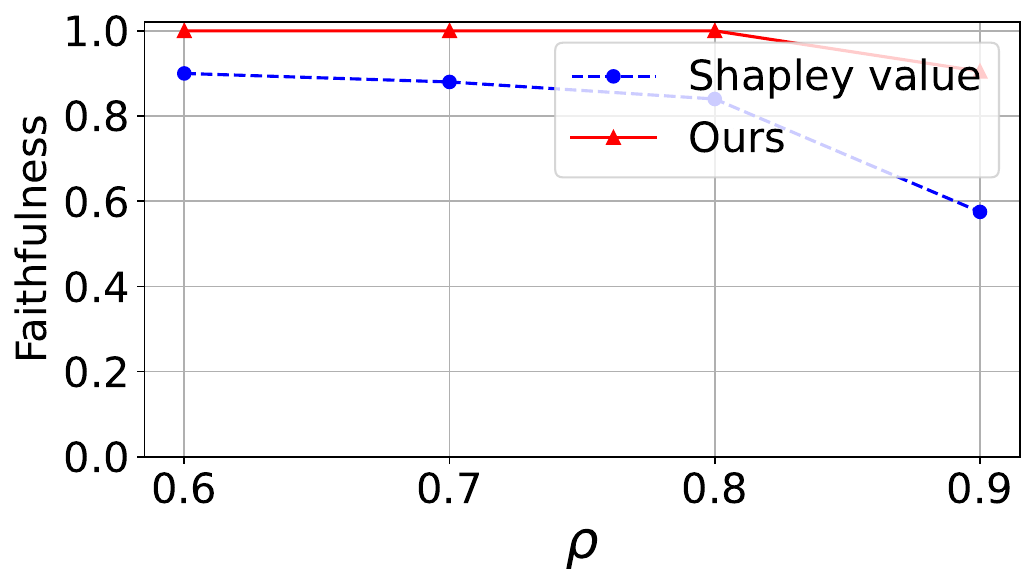}}\vspace{1mm}
\subfloat[AG-news]{\includegraphics[width=0.3\textwidth]{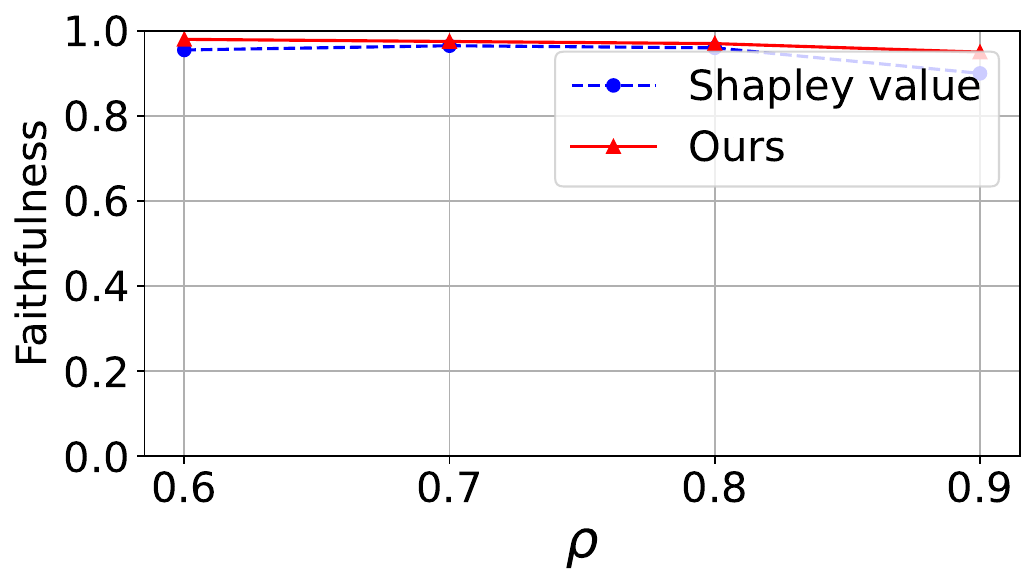}}\vspace{1mm}

\vspace{-4mm}
\caption{Impact of $\rho$ on faithfulness of the explanation for certified defense. The deletion ratio is $20\%$. First row: no attack. Second row: backdoor attack. Third row: adversarial attack.
}
\label{fig-ablation-rho-certified-defense-faithfulness}
\vspace{-7mm}
\end{figure*}

\begin{figure*}
\vspace{5mm}
\centering
{\includegraphics[width=0.3\textwidth]
{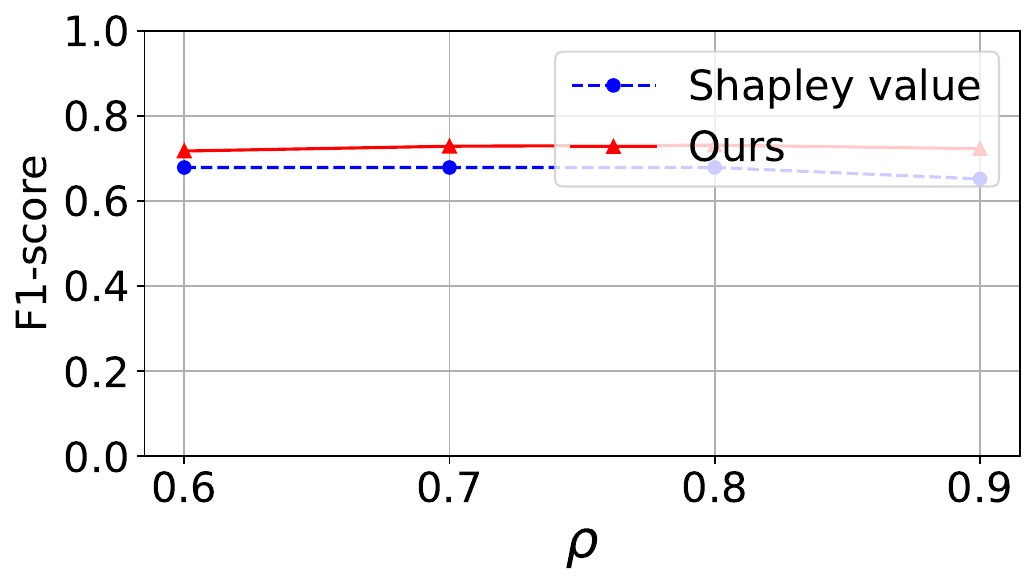}}
{\includegraphics[width=0.3\textwidth]{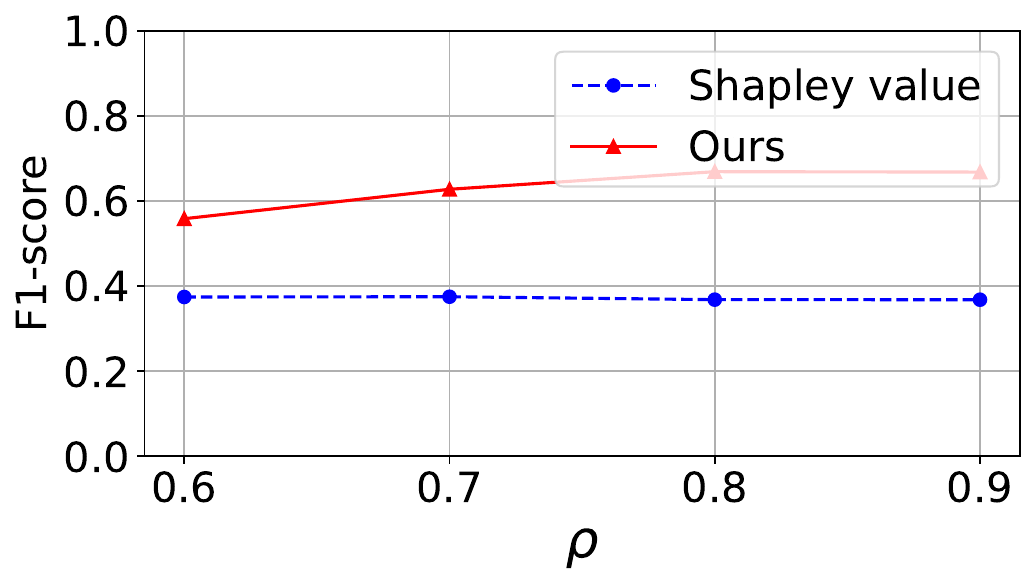}}
{\includegraphics[width=0.3\textwidth]{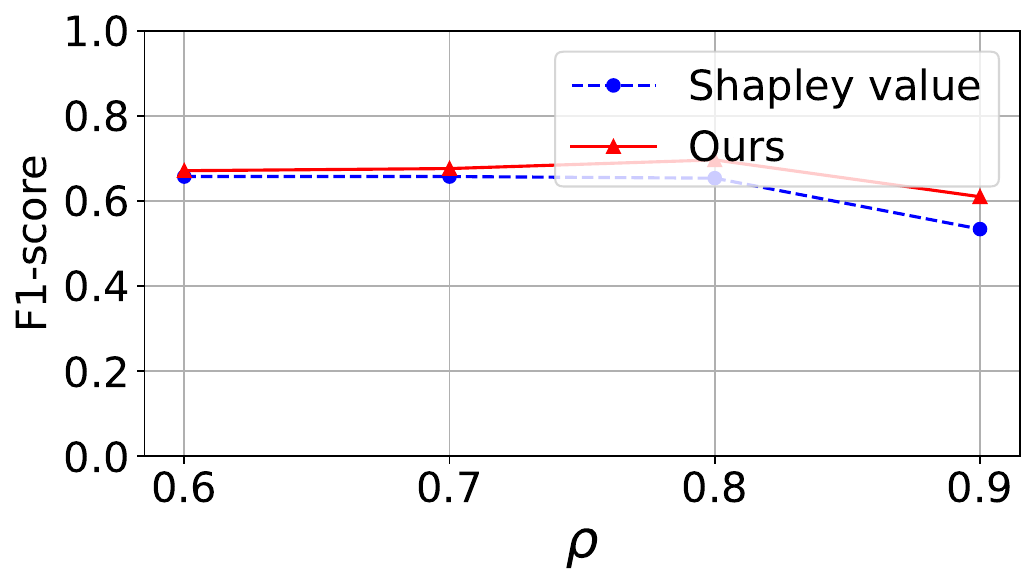}}
\subfloat[SST-2]{\includegraphics[width=0.3\textwidth]{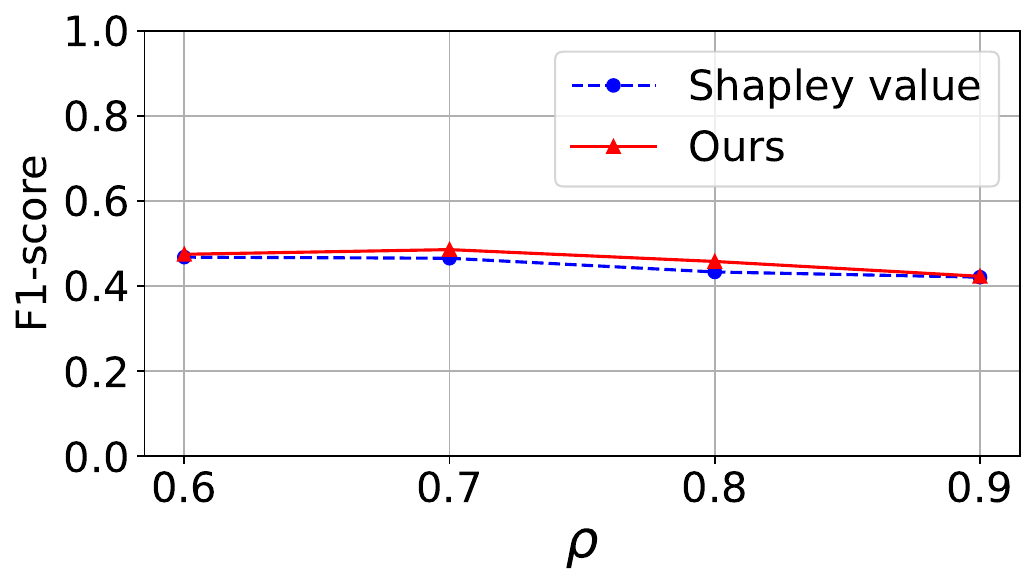}}\vspace{1mm}
\subfloat[IMDb]{\includegraphics[width=0.3\textwidth]{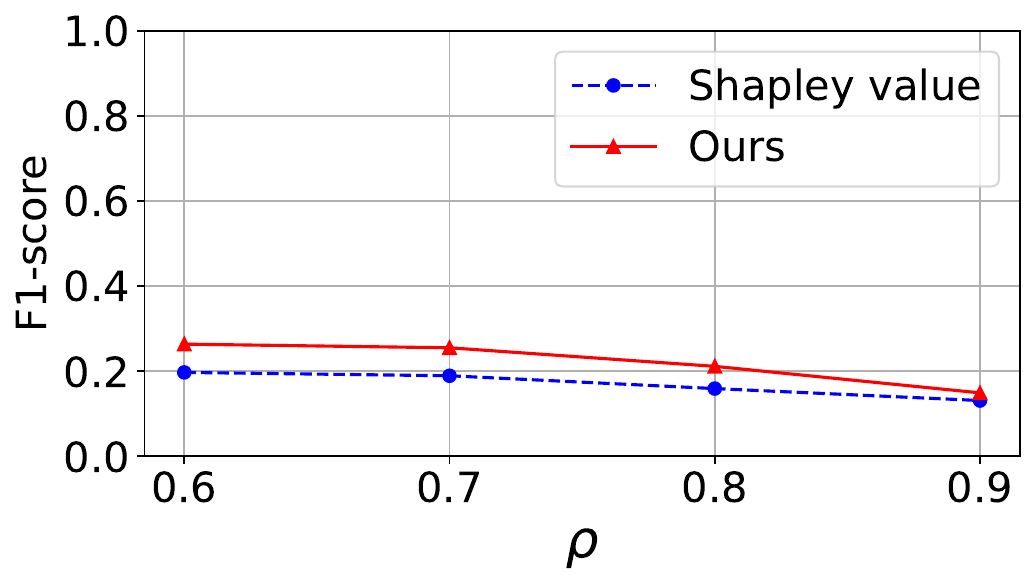}}\vspace{1mm}
\subfloat[AG-news]{\includegraphics[width=0.3\textwidth]{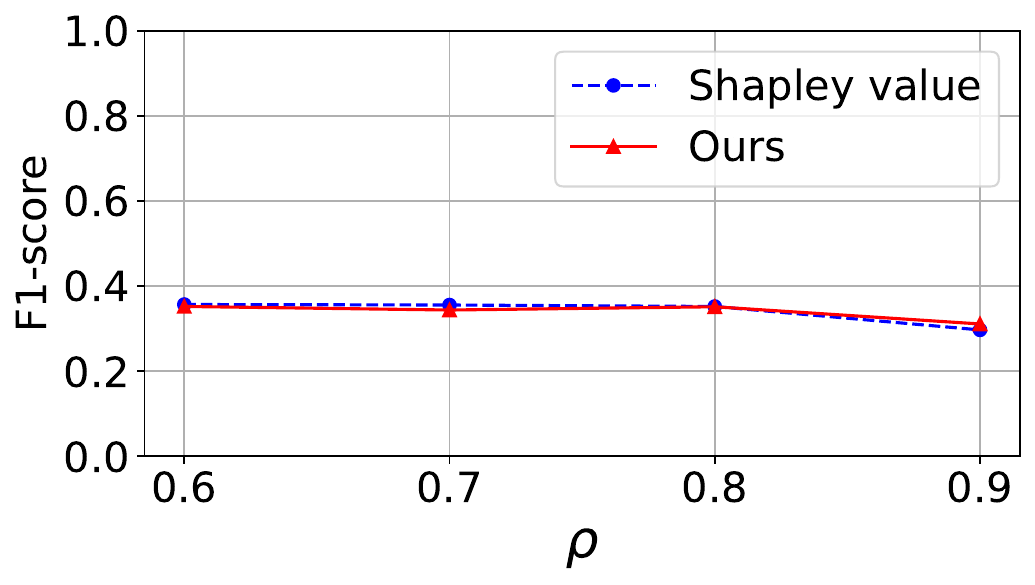}}\vspace{1mm}

\vspace{-4mm}
\caption{Impact of $\rho$ on key word prediction F1-score of the explanation for certified defense. $e = 5$. First row: backdoor attack. Second row: adversarial attack.
}
\label{fig-ablation-rho-certified-defense-f1}
\vspace{-7mm}
\end{figure*}

\begin{figure*}
\vspace{5mm}
\centering

{\includegraphics[width=0.3\textwidth]{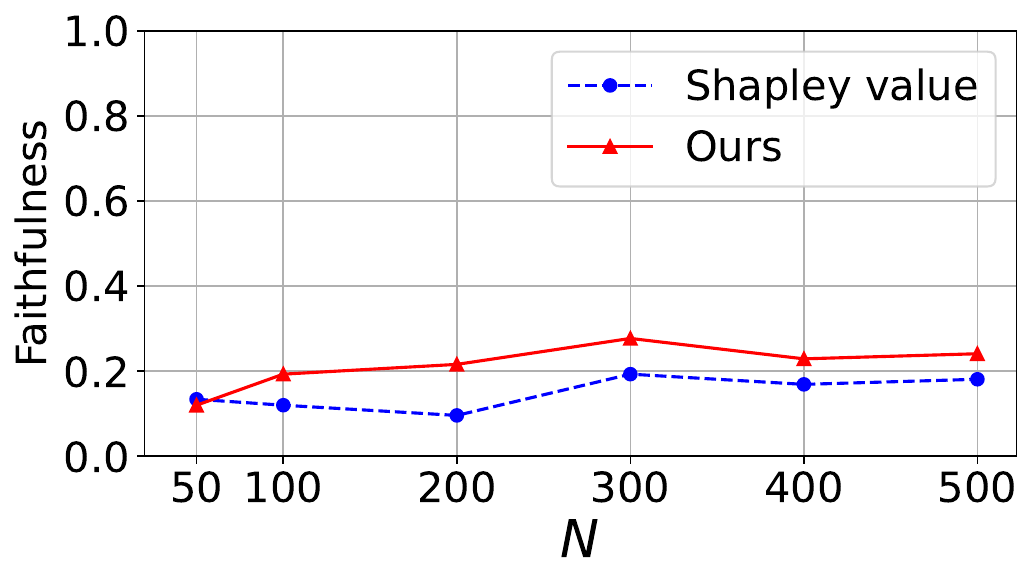}}
{\includegraphics[width=0.3\textwidth]{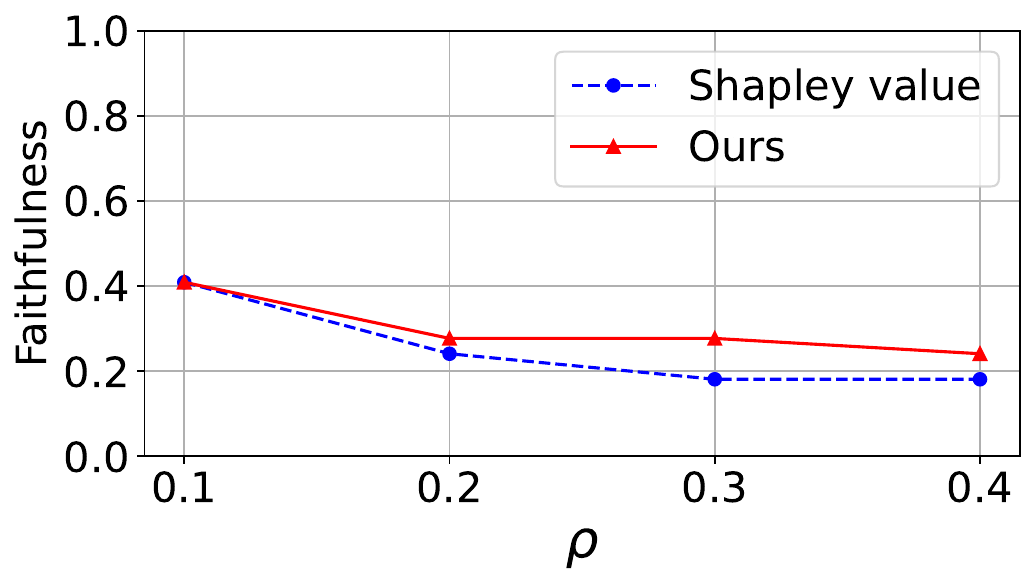}}

\subfloat[Impact of $N$]{\includegraphics[width=0.3\textwidth]{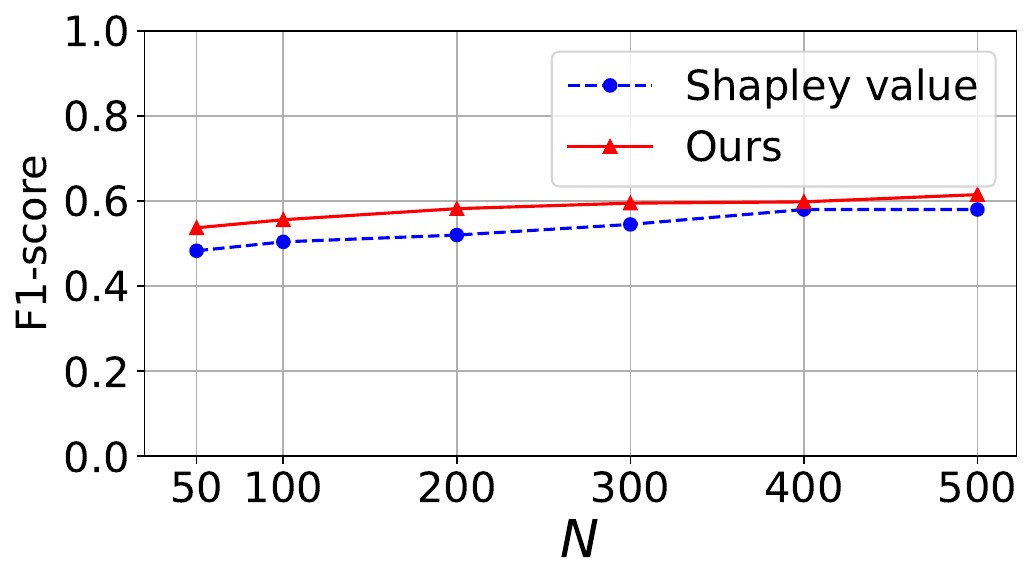}}\vspace{1mm}
\subfloat[Impact of $\rho$]{\includegraphics[width=0.3\textwidth]{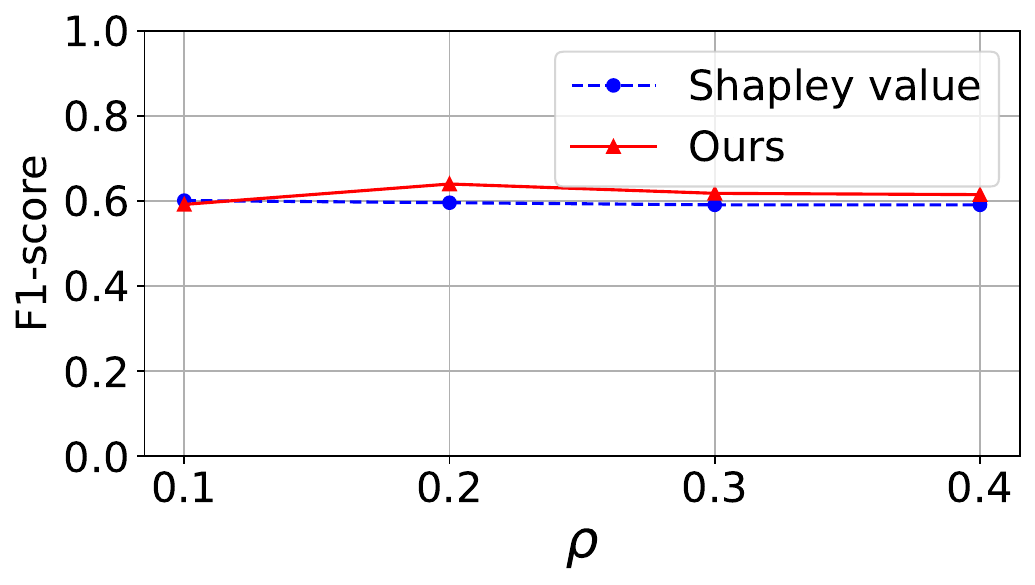}}\vspace{1mm}

\vspace{-4mm}
\caption{Impact of $N$ and $\rho$ on the performance of the explanation for jailbreaking attacks. The jailbreaking attack type is GCG. First row: faithfulness (deletion ratio is $20\%$). Second row: key word prediction F1-score ($e = 10$). 
}
\label{fig-ablation-jailbreaking-attack}
\vspace{-7mm}
\end{figure*}

\begin{figure*}
\vspace{5mm}
\centering
{\includegraphics[width=0.3\textwidth]{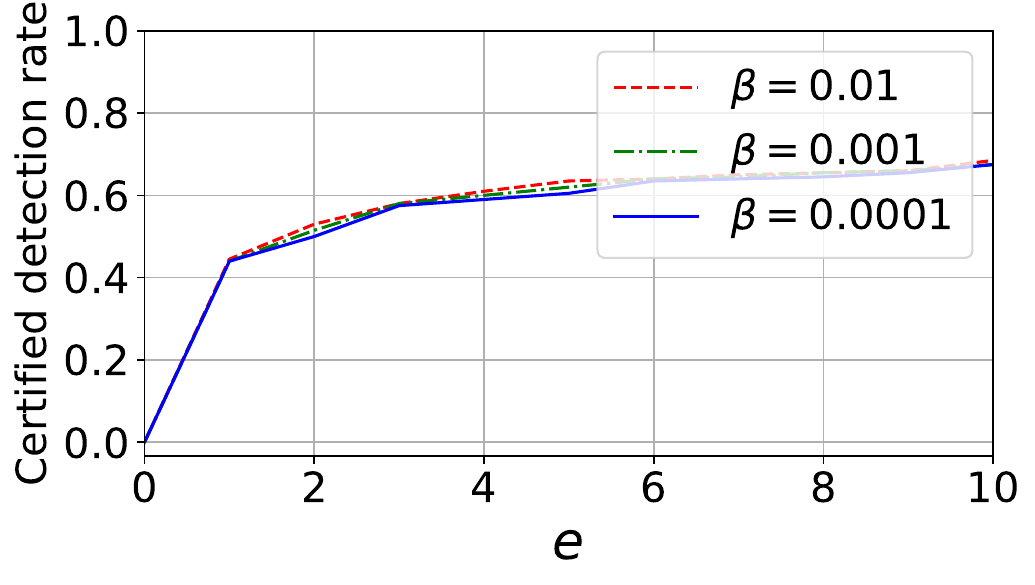}}
{\includegraphics[width=0.3\textwidth]{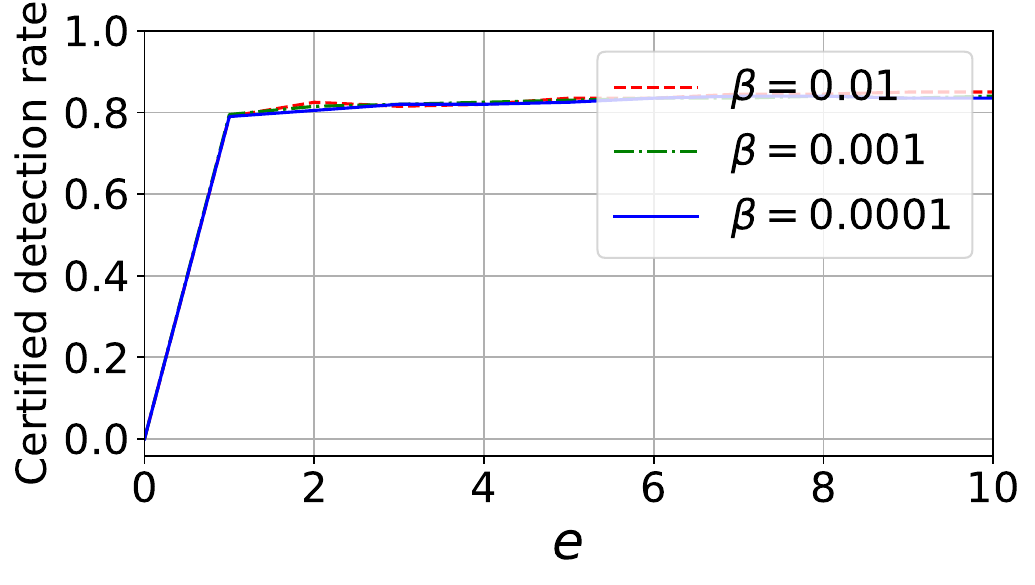}}
{\includegraphics[width=0.3\textwidth]{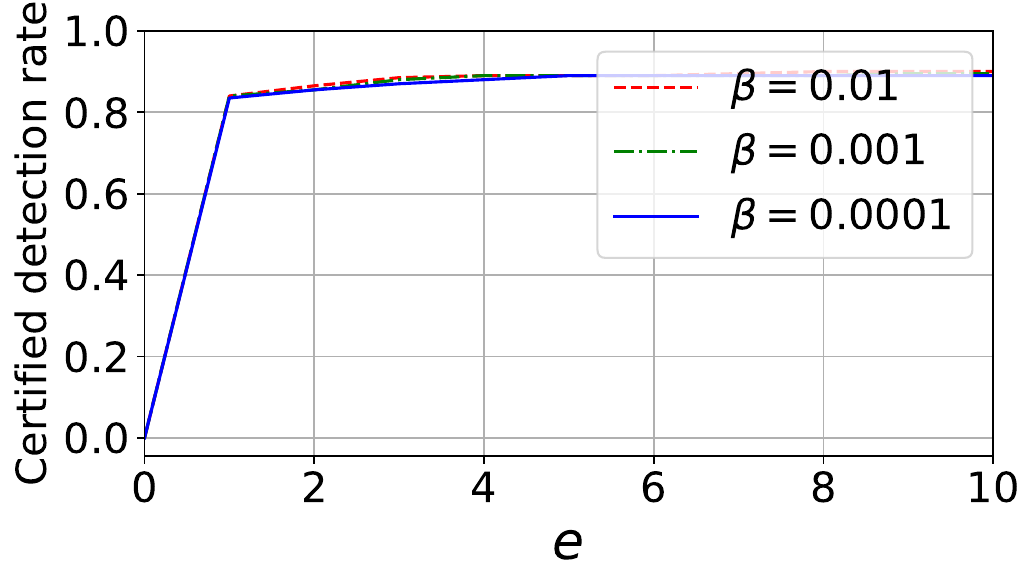}}

{\includegraphics[width=0.3\textwidth]{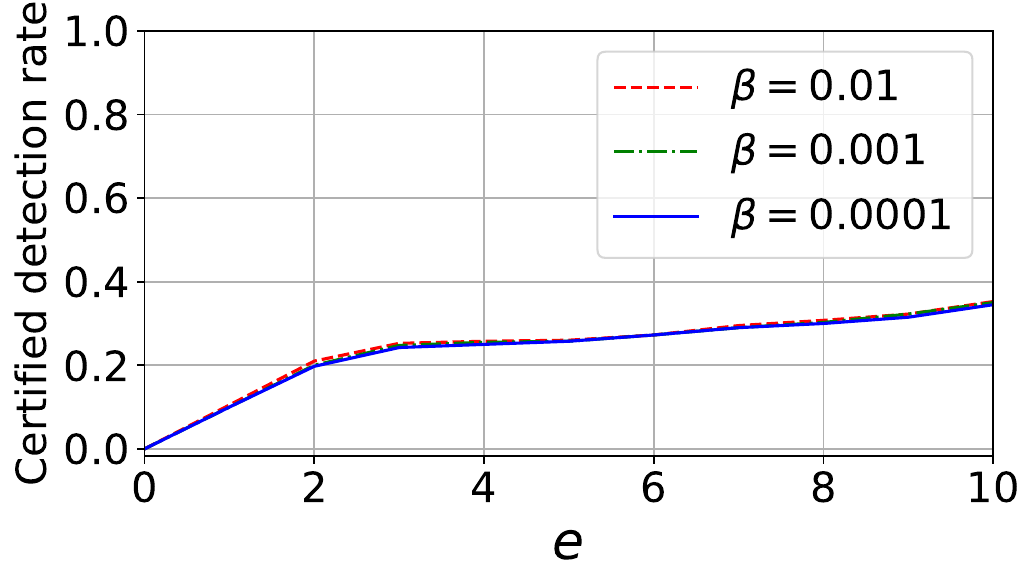}}
{\includegraphics[width=0.3\textwidth]{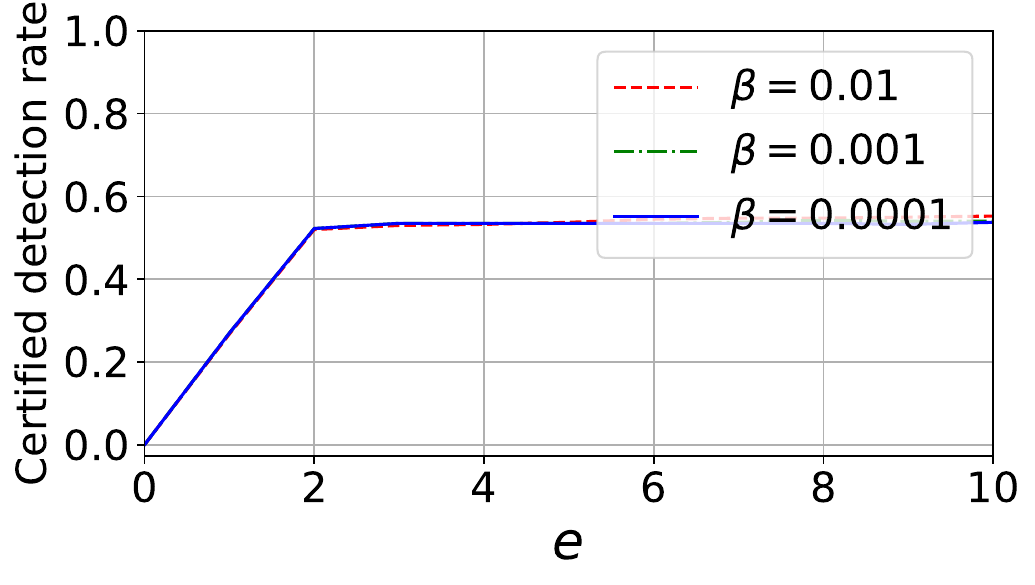}}
{\includegraphics[width=0.3\textwidth]{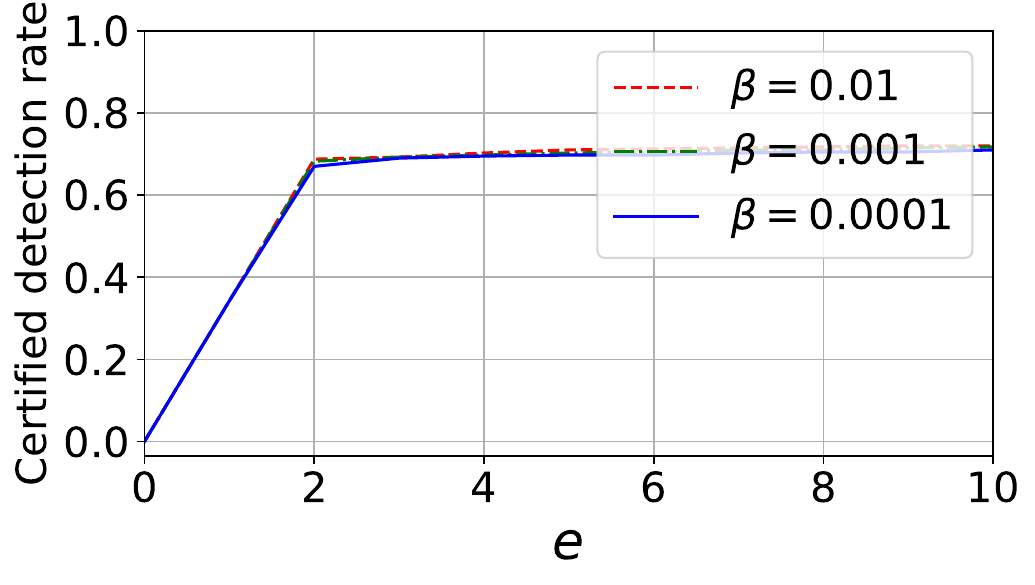}}

\subfloat[SST-2]{\includegraphics[width=0.3\textwidth]{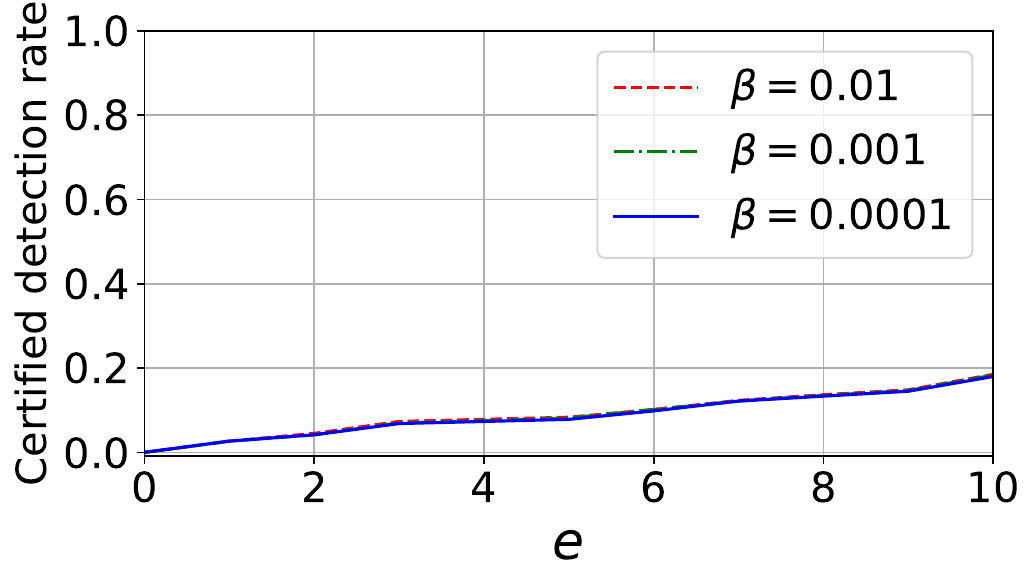}}\vspace{1mm}
\subfloat[IMDb]{\includegraphics[width=0.3\textwidth]{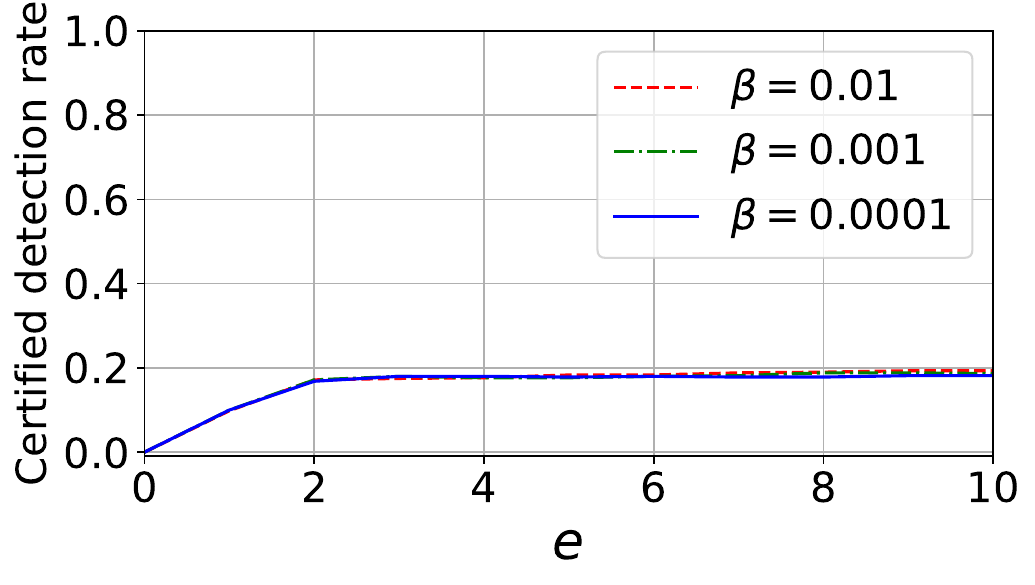}}\vspace{1mm}
\subfloat[AG-news]{\includegraphics[width=0.3\textwidth]{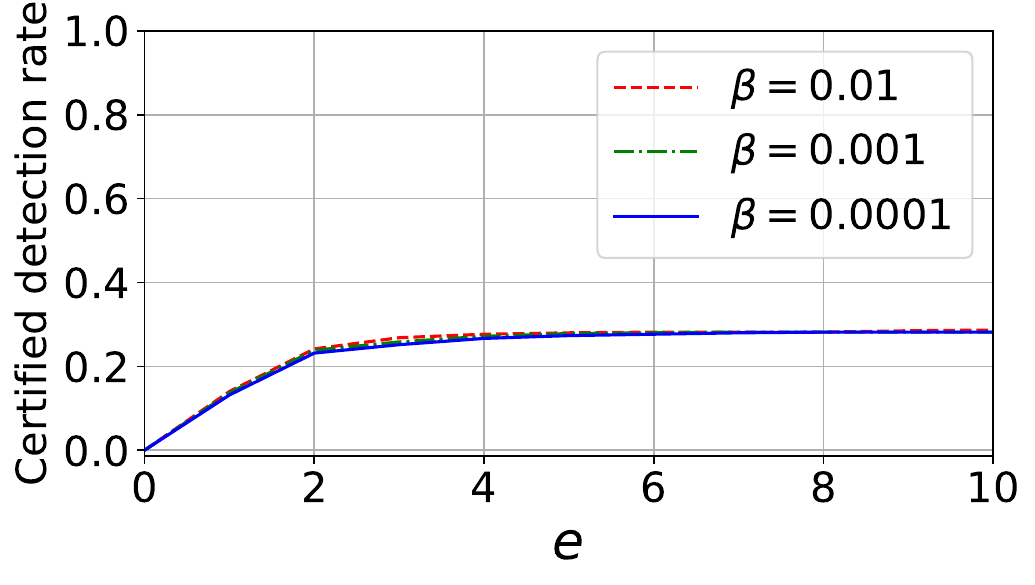}}\vspace{1mm}

\vspace{-4mm}
\caption{Impact of $\beta$ on certified detection rate for varying number of modified features (denoted by $T$). First row: $T=1$. Second row: $T=2$. Third row: $T=3$.
}
\label{fig-ablation-beta-certified-detection}
\vspace{-7mm}
\end{figure*}

\begin{figure*}
\vspace{5mm}
\centering
{\includegraphics[width=0.3\textwidth]{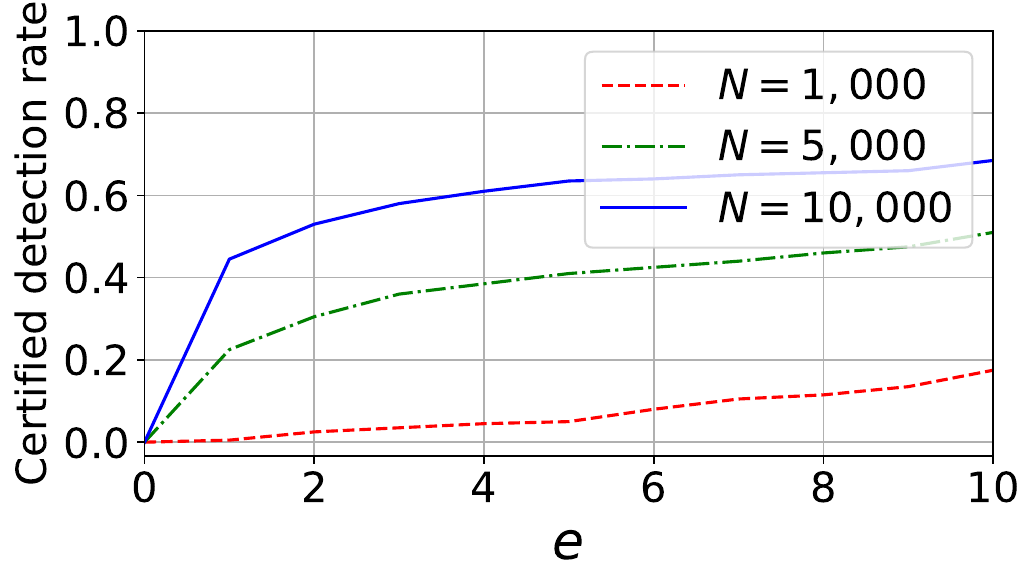}}
{\includegraphics[width=0.3\textwidth]{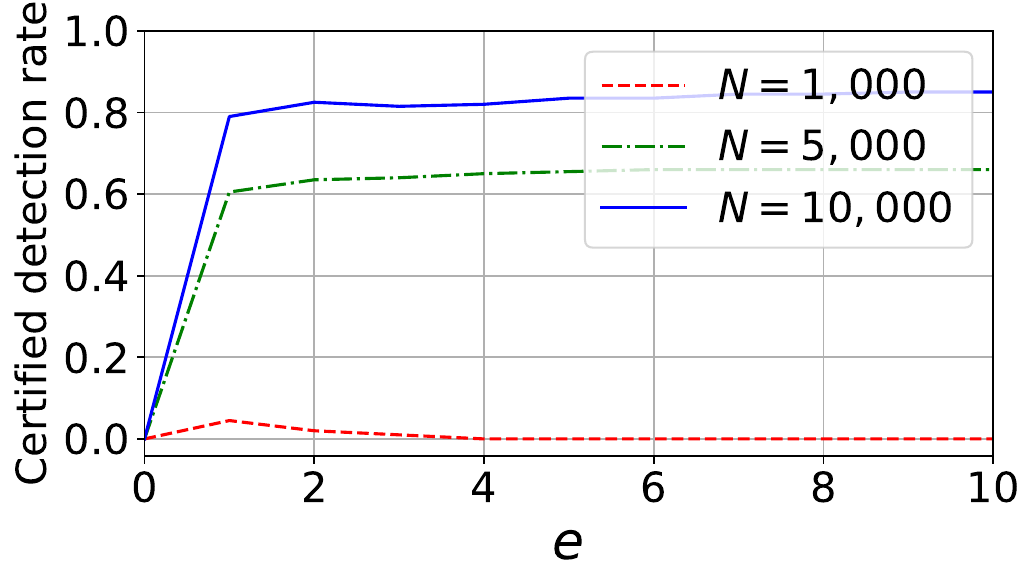}}
{\includegraphics[width=0.3\textwidth]{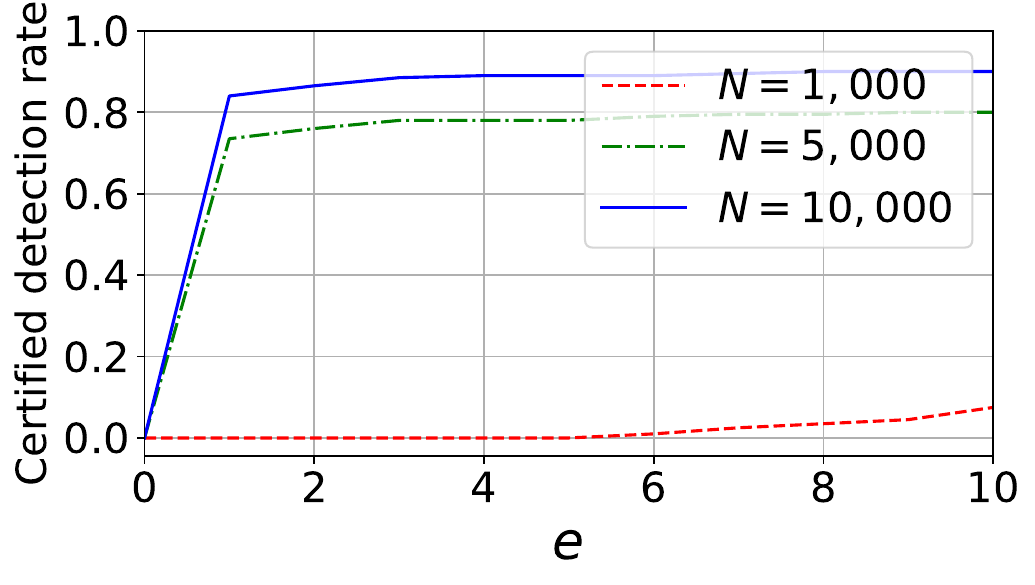}}

{\includegraphics[width=0.3\textwidth]{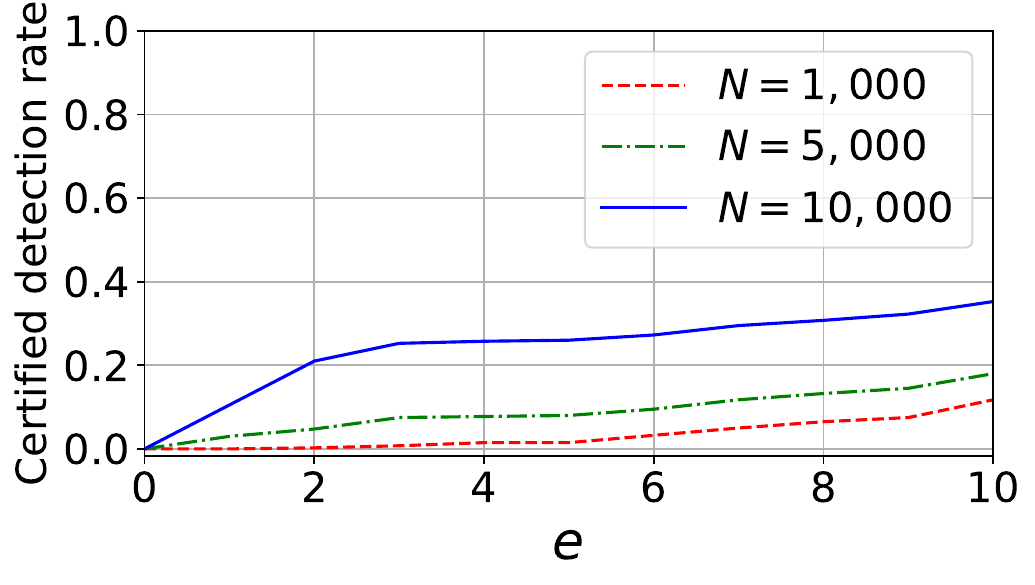}}
{\includegraphics[width=0.3\textwidth]{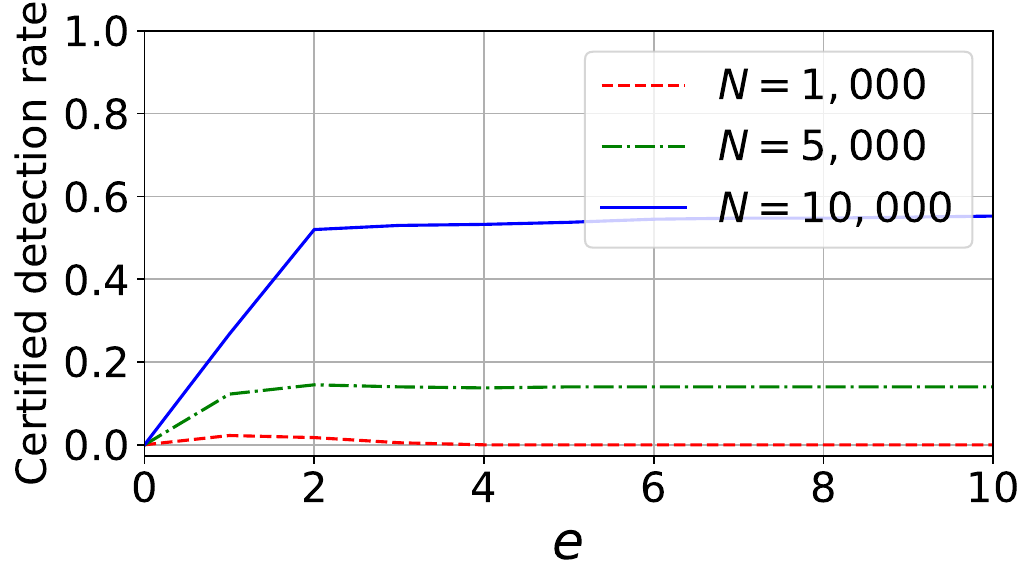}}
{\includegraphics[width=0.3\textwidth]{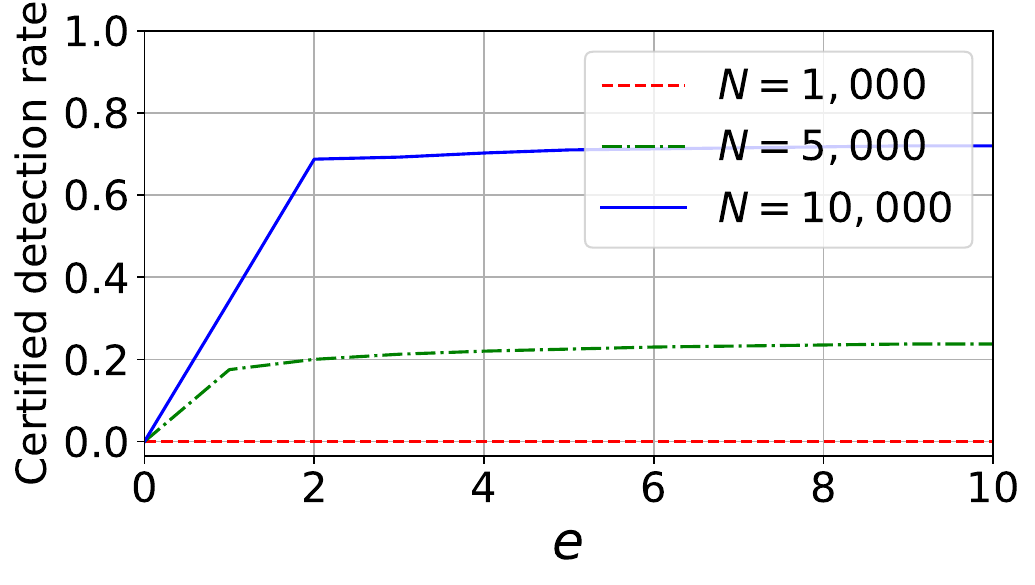}}

\subfloat[SST-2]{\includegraphics[width=0.3\textwidth]{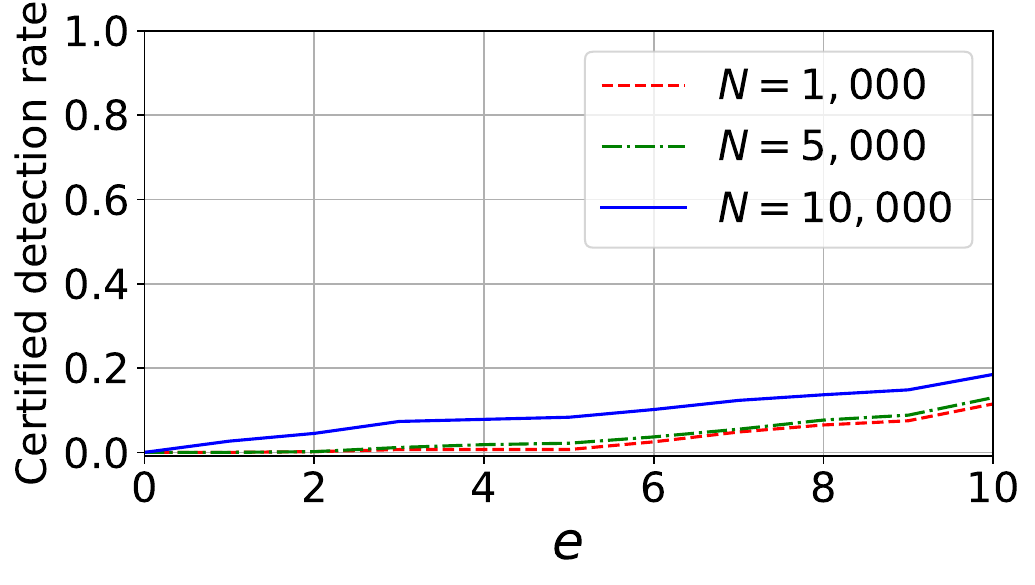}}\vspace{1mm}
\subfloat[IMDb]{\includegraphics[width=0.3\textwidth]{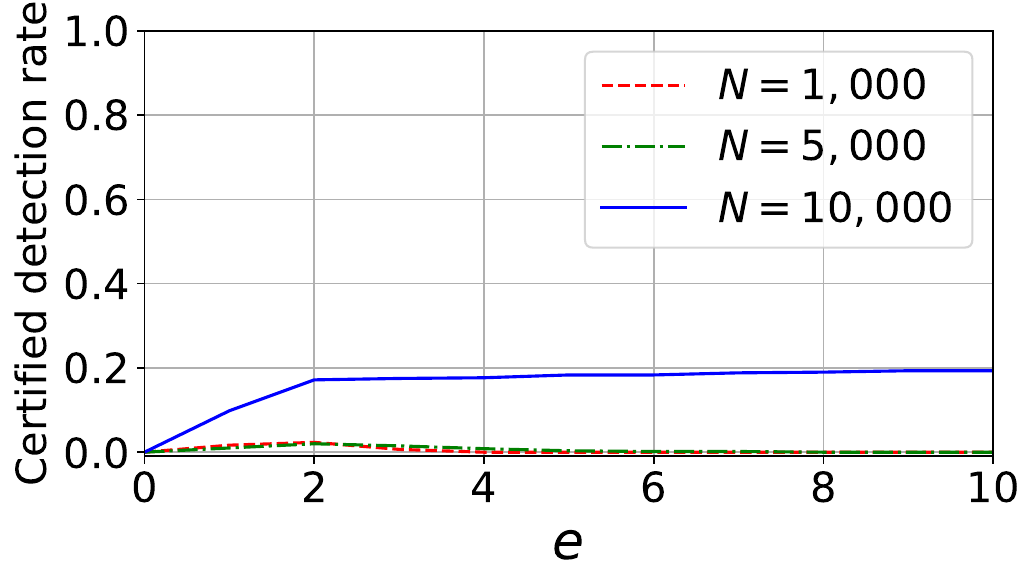}}\vspace{1mm}
\subfloat[AG-news]{\includegraphics[width=0.3\textwidth]{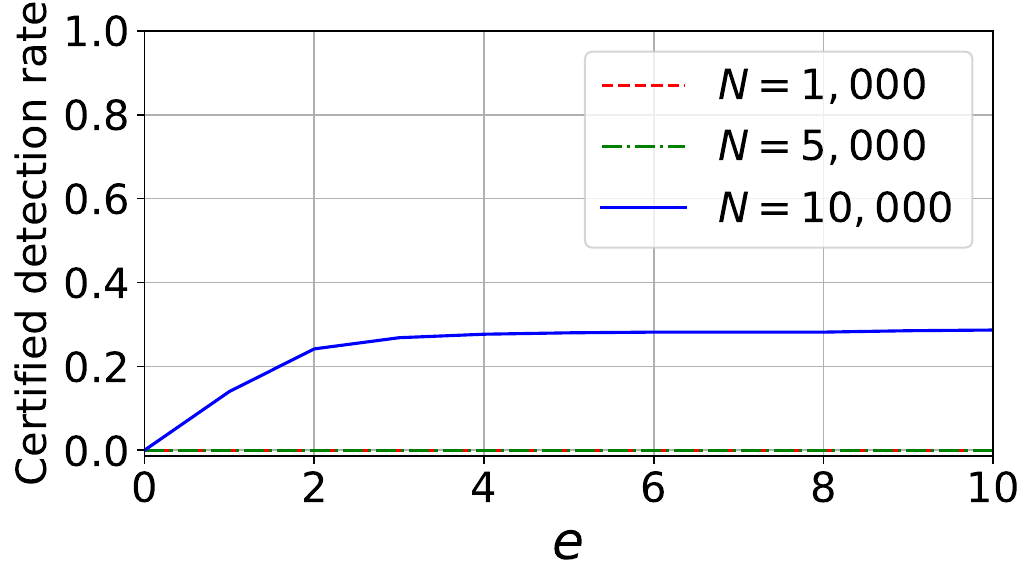}}\vspace{1mm}

\vspace{-4mm}
\caption{Impact of $N$ on certified detection rate for varying number of modified features (denoted by $T$). First row: $T=1$. Second row: $T=2$. Third row: $T=3$.
}
\label{fig-ablation-N-certified-detection}
\vspace{-7mm}
\end{figure*}

\begin{figure*}
\vspace{5mm}
\centering
{\includegraphics[width=0.3\textwidth]{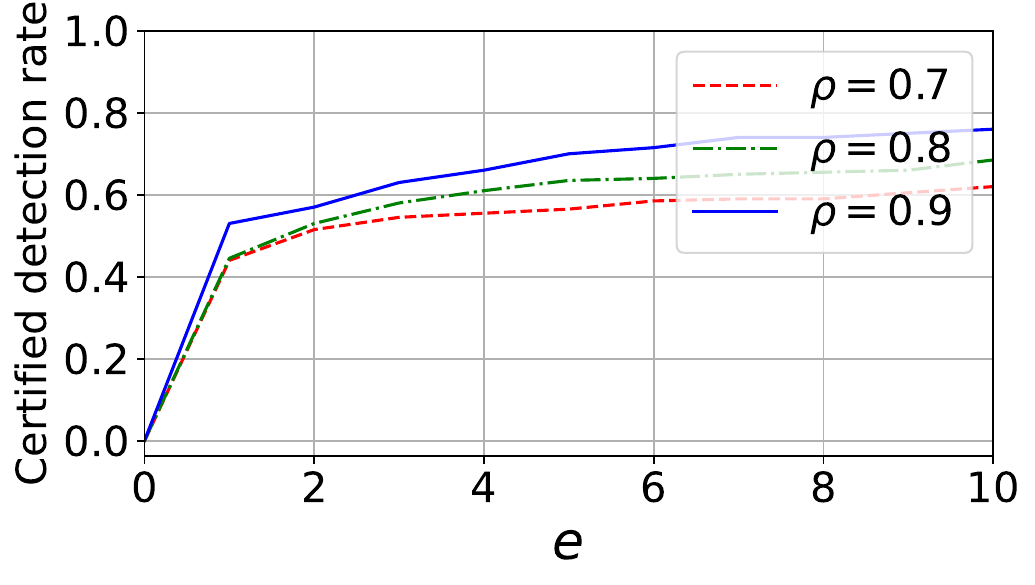}}
{\includegraphics[width=0.3\textwidth]{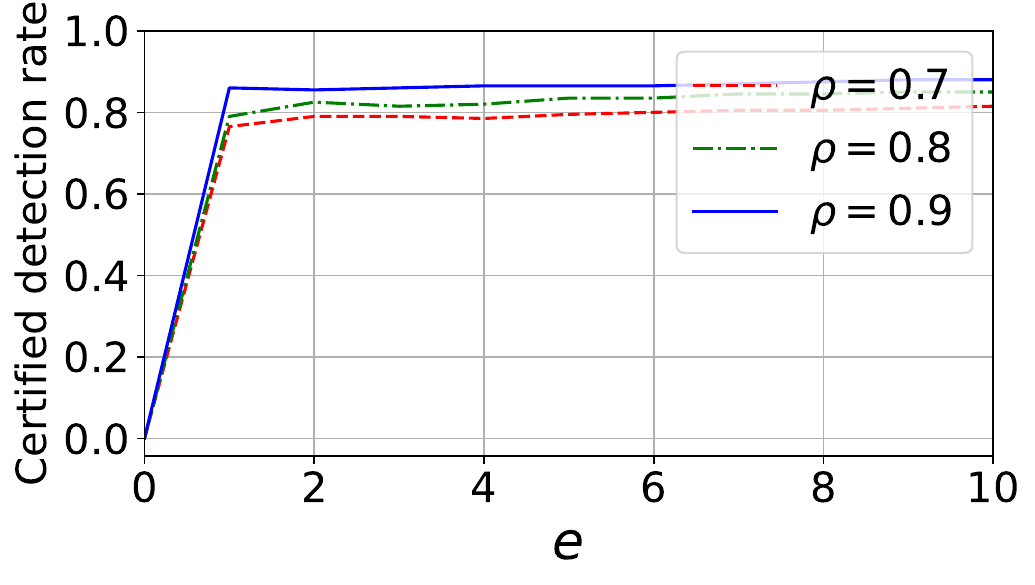}}
{\includegraphics[width=0.3\textwidth]{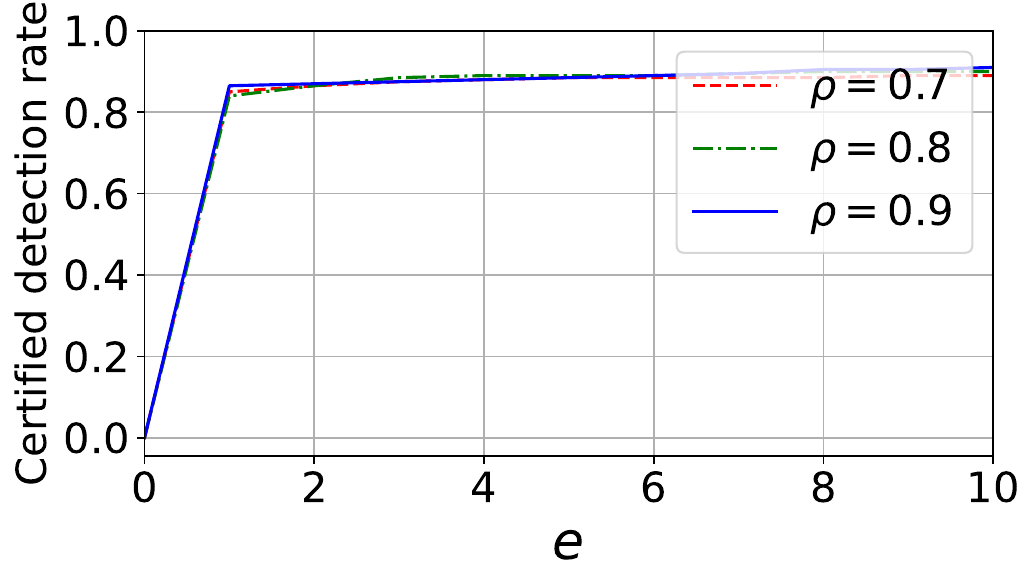}}

{\includegraphics[width=0.3\textwidth]{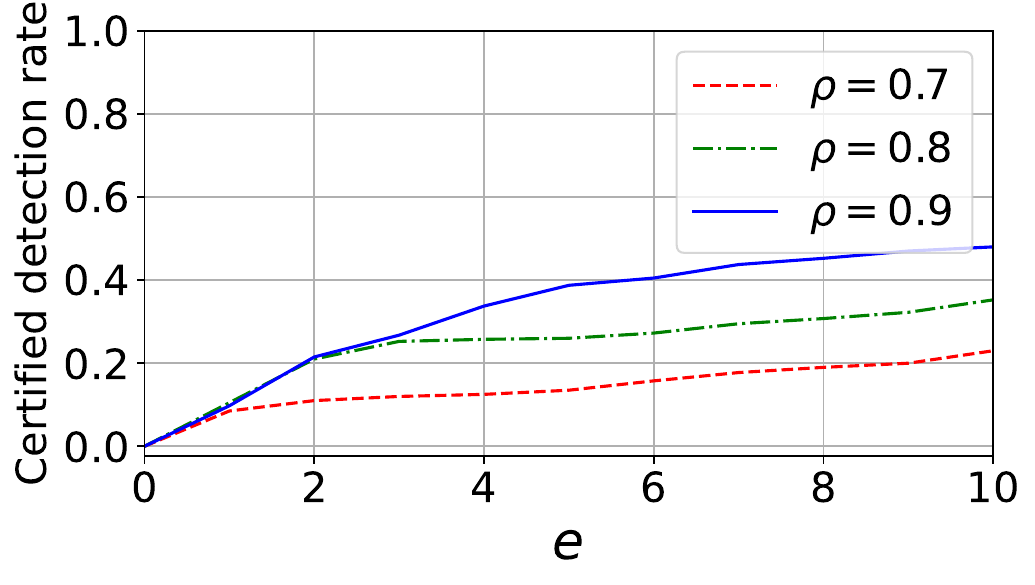}}
{\includegraphics[width=0.3\textwidth]{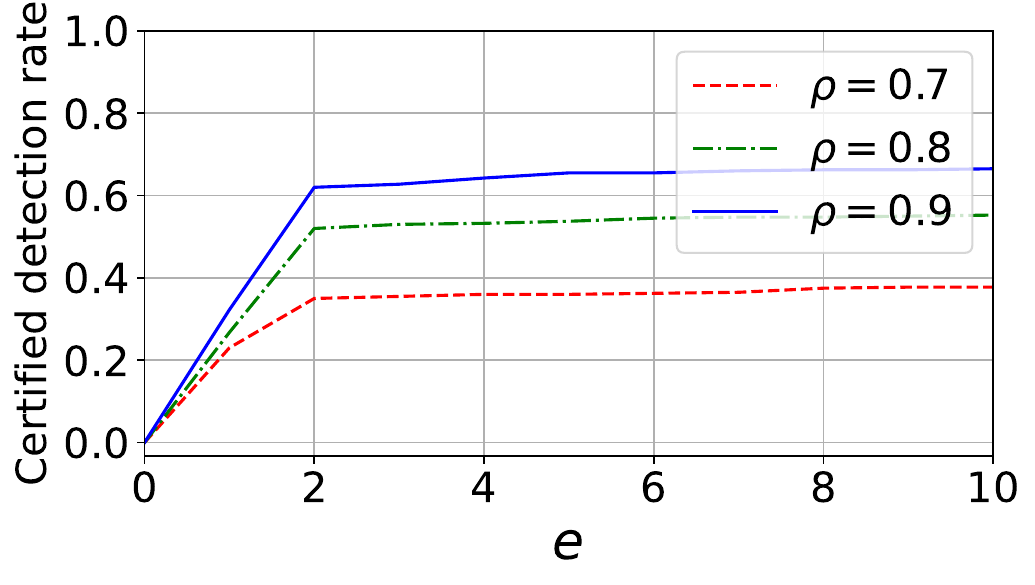}}
{\includegraphics[width=0.3\textwidth]{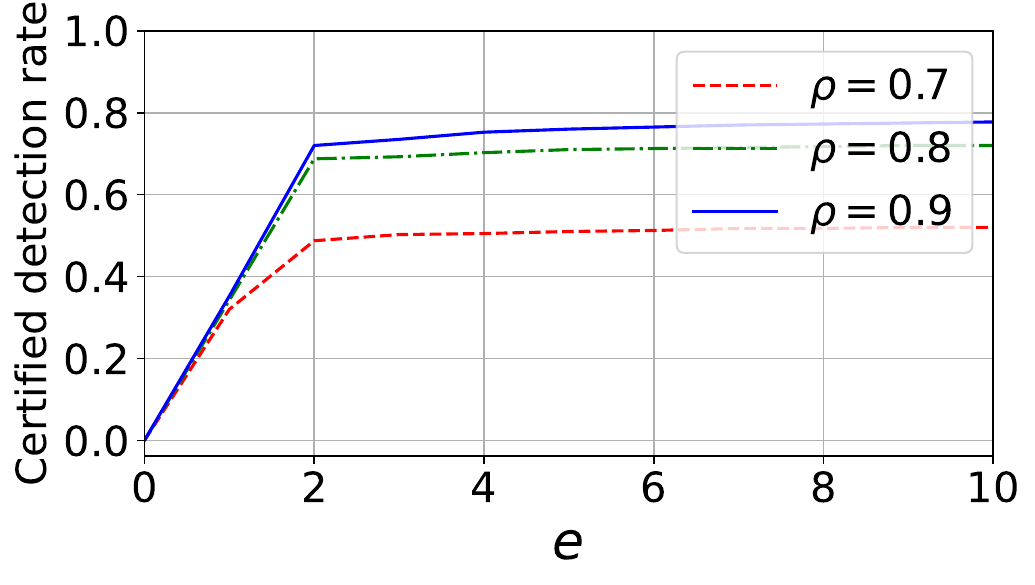}}

\subfloat[SST-2]{\includegraphics[width=0.3\textwidth]{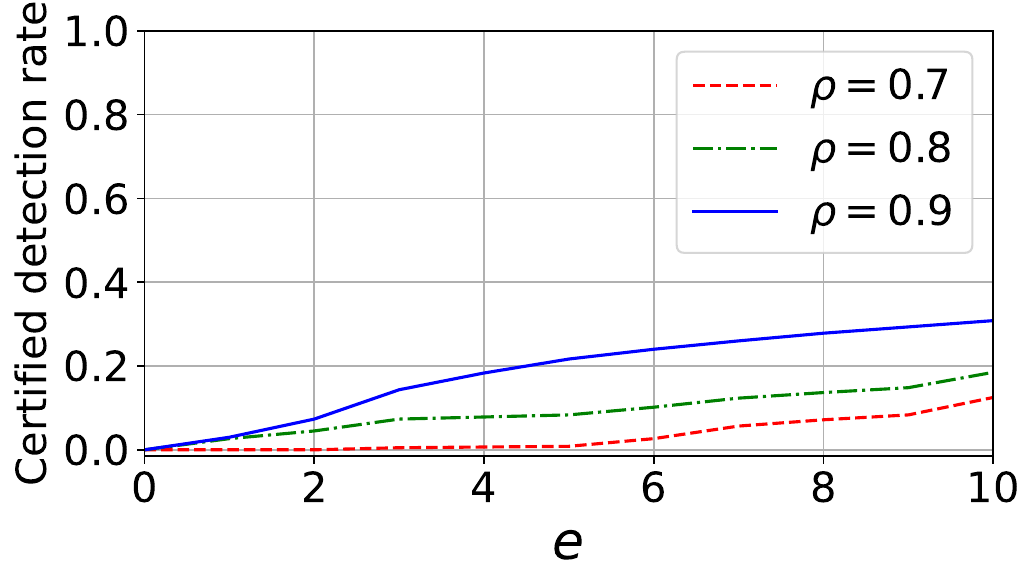}}\vspace{1mm}
\subfloat[IMDb]{\includegraphics[width=0.3\textwidth]{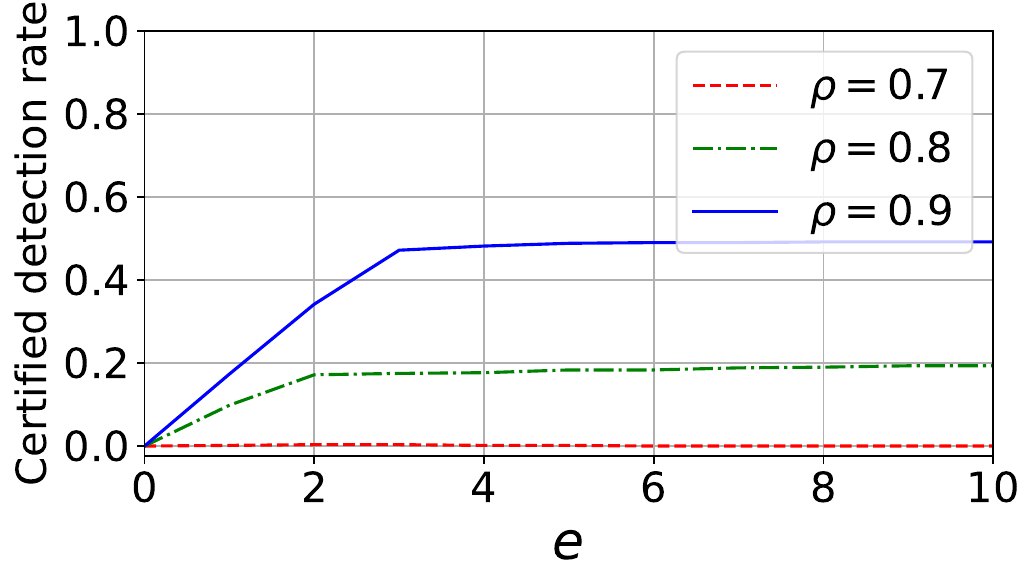}}\vspace{1mm}
\subfloat[AG-news]{\includegraphics[width=0.3\textwidth]{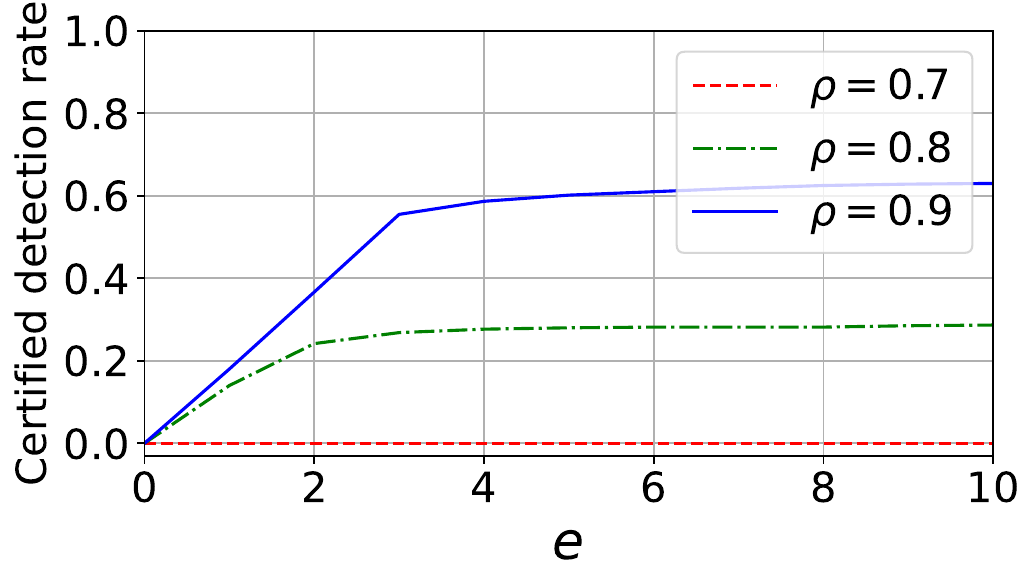}}\vspace{1mm}

\vspace{-4mm}
\caption{Impact of $\rho$ on certified detection rate for varying number of modified features (denoted by $T$). First row: $T=1$. Second row: $T=2$. Third row: $T=3$.
}
\label{fig-ablation-rho-certified-detection}
\vspace{-7mm}
\end{figure*}

\end{document}